\documentclass[12pt]{article}

\usepackage[top=1in,bottom=1in, left=1in, right=1in]{geometry} 

\usepackage{amsmath,amsthm,amssymb,amsopn,amsfonts,pdfpages,algorithmic,algorithm,dsfont,float, fourier,enumerate, setspace, xspace, graphicx, booktabs, mathtools, tikz, enumitem, varwidth, bbold, wrapfig, cite,dsfont, soul,xcolor}
\usepackage[colorlinks=true,linkcolor=cyan,citecolor=blue]{hyperref}
\usepackage{graphicx}
\usepackage{caption}
\usepackage{arydshln}

\DeclareMathOperator*{\argmin}{arg\,min}
\newcommand{\csp}{\mathcal{S}_k}
\newcommand{\csq}{\mathcal{S}_\ell}

\newcommand{\bX}{\textbf{X}}
\newcommand{\bA}{\textbf{A}}
\newcommand{\bB}{\textbf{B}}
\newcommand{\bW}{\textbf{W}}
\newcommand{\bP}{\textbf{P}}
\newcommand{\hbP}{\widehat{\textbf{P}}}
\newcommand{\hbX}{\widehat{\textbf{X}}}
\newcommand{\hX}{\widehat{X}}

\newcommand{\bE}{\mathbf{E}}
\newcommand{\bQ}{\textbf{Q}}
\newcommand{\tbQ}{\widetilde{\textbf{Q}}}
\newcommand{\bC}{\textbf{C}}

\newcommand{\op}{O_{\mathbb{P}}}

\newcommand{\bY}{\textbf{Y}}

\DeclarePairedDelimiter{\norm}{\lVert}{\rVert} 
\newcommand{\rvline}{\hspace*{-\arraycolsep}\vline\hspace*{-\arraycolsep}}

\newtheorem{theorem}{Theorem}[section]
\newtheorem{manualtheoreminner}{Theorem}
\newenvironment{manualtheorem}[1]{%
  \renewcommand\themanualtheoreminner{#1}%
  \manualtheoreminner
}{\endmanualtheoreminner}
\newtheorem{prop}{Proposition}[section]

\newtheorem{cor}{Corollary}[section]
\newtheorem{ass}{Assumption}
\theoremstyle{definition}
\newtheorem{defn}{Definition}[section]
\newtheorem{rem}{Remark}[section]

\theoremstyle{definition}
\newtheorem{example}{Example}[section]

\usepackage{amssymb,authblk}

\begin{document}

\title{Lost in the Shuffle: Testing Power in the Presence of Errorful Network Vertex Labels}


\author[$1$]{Ayushi Saxena}
\author[$1$]{Vince Lyzinski}
\affil[$1$]{\small Department of Mathematics, University of Maryland}

\maketitle

\begin{abstract}
Two-sample network hypothesis testing is an important inference task with applications across diverse fields such as medicine, neuroscience, and sociology.
Many of these testing methodologies operate under the implicit assumption that the vertex correspondence across networks is a priori known.
This assumption is often untrue, and the power of the subsequent test can
degrade when there are misaligned/label-shuffled vertices across networks.
This power loss due to shuffling is theoretically explored in the context of random dot product and stochastic block model networks for a pair of hypothesis tests based on Frobenius norm differences between estimated edge probability matrices or between adjacency matrices.
The loss in testing power is further reinforced by numerous simulations and experiments, both in the stochastic block model and in the random dot product graph model, where the power loss across multiple recently proposed tests in the literature is considered.
Lastly, the impact that shuffling can have in real-data testing is demonstrated in a pair of examples from neuroscience and from social network analysis.
\end{abstract}

\section{Introduction}
\label{sec:intro}
\noindent Interest in graph and network-valued data has soared in the last decades \cite{Goldenberg2009}, as networks have become a common data type for modeling complex dependencies and interactions in many fields of study, ranging from neuroscience \cite{bullmore2009complex,durante2018bayesian,chung2021statistical} 
to sociology \cite{mitchell1974social,carrington2005models}
to biochemistry \cite{vazquez2003global,temkin2020chemical}, among others. 
As network data has become more commonplace, a number of statistical tools tailored for handling network data have been developed \cite{kolaczyk2009statistical,kolaczyk2014statistical}, 
including methods for hypothesis testing \cite{tang2017semiparametric,Tang2017,testing1_2017}, 
goodness-of-fit analysis \cite{goodness1,goodness2,GoF_ModSel}, 
clustering \cite{newman2001clustering,clauset2004finding,blondel2008fast,rohe2011spectral,Sussman2012,lei2015consistency}, and classification 
\cite{vogelstein2015shuffled,tang2013universally,lyzinski2016community,zhang2018end}, among others.

While early work on network inference often considered a single network as the data set, there has been a relatively recent proliferation of datasets that consist of 
multiple networks on the same set of vertices (which we will referred to as the ``paired'' setting) and multiple networks on differing vertex sets (the ``unpaired'' setting).
In the paired setting, example datasets include the DTMRI and FMRI connectome networks considered in 
\cite{tang2017semiparametric,durante2017nonparametric,Arroyo2019},
and the paired social networks from Twitter (now X) in \cite{lyzinski2020matchability} and from Facebook \cite{viswanath2009evolution,heimann2018regal} to name a few.
In the unpaired setting, example datasets include the social networks that partially share a user base from \cite{magnani2011ml,patsolic2020vertex} and the friendship networks in \cite{mastrandrea2015contact}.
When considering multiple paired networks for a single inference task, these methods often make an implicit assumption that the graphs are ``vertex-aligned,'' i.e., that there is an a priori known, true, one--to--one correspondence across the vertex labels of the graphs.
Often, this is not the case in practice (see the voluminous literature on network alignment and graph matching \cite{30yrs,10yrs,short_survey_GM}), as node alignments may be obscured by numerous factors, including, for example, different usernames across social networks \cite{patsolic2020vertex},
 unknown correspondence of neurons across brain hemispheres \cite{pedigo2023bisected},
and misalignments/misregistrations to a common brain atlas (to create connectomes) as discussed in \cite{fiori2013robust}.
Moreover, these errors across the observed vertex labels can have a dramatically detrimental impact on subsequent inferential performance \cite{lyzinski2018information}.

An important inference task in the multiple network setting is two-sample network hypothesis testing \cite{chen2020multiple,chen2023hypothesis}.
Two-sample network testing has been used, for example, to compare neurological properties (captured via patient connectomes) of populations of patients along various demographic characteristics like age, sex, and the presence/absence of a neurological disorder (indeed, this connectomic testing example will serve as motivation for us in Section \ref{sec:motive}).  
Among the work in this area, 
numerous methods exist for (semi)parametric hypothesis testing across network samples, where one of the parameters being leveraged is the correspondence of labels across networks; see for example \cite{tang2017semiparametric,du2023hypothesis,asta2014geometric,levin2017central}.
There are also nonparametric methods that estimate and compare network distributions, and ignore the information contained in the vertex labels; see, for example, \cite{Tang2017,agterberg2020nonparametric}.

This paper considers an amalgam of the above settings in which there is signal to be leveraged in one portion of the vertex labels, and uncertainty across the remaining labels.
Our principle task is then two-fold. First, we seek to understand how this label uncertainty impacts testing power both theoretically (Section \ref{sec:theory}) and in practice (Sections \ref{sec:PhatvsA}--\ref{sec:SN}). 
Second, we seek to better understand how to mitigate the impact of this uncertainty via a graph matching preprocessing step that aligns the graphs before testing (Section \ref{sec:GM}).
Before formally defining the shuffled testing problem, we will first set here some of the notational conventions that will appear throughout the paper.
Note also that all necessary code to reproduce the figures in the paper can be found at \url{https://www.math.umd.edu/~vlyzinsk/Shuffled_testing/}.

%%%%%%%%%%%%%%%%%%%%%%%%%%%%%%%%%%
%%%%%%%%%%%%%%%%%%%%%%%%%%%%%%%%%%

\subsection{Notation}
\label{sec:notation}

Given an undirected graph $n$-vertex graph $G$ (all graphs considered will be undirected graphs with no self-loops), we let $[n]:=\{1,2, ..., n\}$ denote the vertices of $G$. 
The adjacency matrix \textbf{A} $\in \{0,1\}^{n\times n}$ of $G$ is given by:
   $A_{ij} = \mathds{1}\{\{i,j\} \in E(G)\}$
for all $i,j \in [n]$. 
We note here that we will refer to a graph and its adjacency matrix interchangeably, as these objects (in the setting we consider herein) encode equivalent information.
We denote the $i^{th}$ row of any matrix \textbf{M} with the notation $M_i$ (or via $\mathbf{M}_n[i,:]$ if the subscript is needed to denote explicit dependence on $n$).

We define the usual Frobenius norm $\|\cdot\|_F$ 
of a matrix \textbf{A} via 
   $$\|\textbf{A}\|_F = (\sum_{i,j=1}^n  A_{ij}^2)^{1/2}.$$
For positive integers $d$ and $n$, we will define the set of orthogonal matrices in $\mathbb{R}^{d\times d}$ via  $\mathcal{O}_d$ and the set of $n\times n$ permutation matrices via $\Pi_n$. We indicate a matrix of all ones of size $d\times d$ by $\mathbf{J}_{d}$, and $\mathbf{J}_{n,d}$ is the matrix of all one of size $n\times d$. Similarly, we denote the corresponding matrices of all 0's by $\textbf{0}_d = \{0\}^{d \times d}$ and $\textbf{0}_{n,d} = \{0\}^{n \times d}.$ 
Lastly, the direct sum of two matrices \textbf{A} and \textbf{B} is denoted by $\textbf{A} \oplus \textbf{B}$.

For functions $f,g:\mathbb{Z}_{\geq0}\mapsto \mathbb{R}_{\geq0}$, we use here the standard asymptotic notations:

\begin{align*}
    f=O(g)&\text{ if } \exists\, C> 0,\text{ and }n_0\in\mathbb{Z}_{\geq0}\text{ s.t. }f(n)\leq C g(n)\text{ for }n\geq n_0;\\
    f=\Omega(g)&\text{ if } \exists\, C> 0,\text{ and }n_0\in\mathbb{Z}_{\geq0}\text{ s.t. }Cg(n)\leq  f(n)\text{ for }n\geq n_0;\\
    f=\Theta(g)&\text{ if } f=\Omega(g),\text{ and }f=O(g);\,\,f\sim g\text{ if } \lim_{n\rightarrow\infty}f(n)/g(n)= 1;\\
    f=o(g)&\text{ if } \lim_{n\rightarrow\infty} f(n)/g(n)= 0;\,
    f=\omega(g)\text{ if } \lim_{n\rightarrow\infty}g(n)/f(n)= 0.
\end{align*}
Note that when $f=o(g)$ (resp., $f=\omega(g)$ and $f=\Theta(g)$) when $g$ is a complicated function of $n$, we will often write $f\ll g$ (resp., $f\gg g$ and $f\approx g$) to ease notation.

%%%%%%%%%%%%%%%%%%%%%%%%%%%%%%%%%%%%%%%%%%%%
%%%%%%%%%%%%%%%%%%%%%%%%%%%%%%%%%%%%%%%%%%%%

\subsection{Shuffled testing}

To formalize our partially aligned graph setting, suppose we have two networks $\bA_1$ and $\bA_2$ on a common vertex set $V(\bA_1)=V(\bA_2)=[n]$, and that the label correspondence across networks is known for $n-k$ of the vertices; denote the set of these vertices via $M_{n,k}$.
We further assume that the user has knowledge of which vertices are in $M_{n,k}$.
This, for example, could be the result of graph matching algorithms that provide a measure of certainty for the validity of each matched vertex (see, for example, the soft matching approach of \cite{fishkind2019seeded} or the vertex nomination work of \cite{coppersmith2014vertex,fishkind2015vertex}).
The veracity of the correspondence across the remaining $k$ vertices not in $M_{n,k}$ (which we shall denote via $U_{n,k}:=[n]\setminus M_{n,k}$), is assumed unknown a priori.  This may be due, for example, to algorithmic uncertainty in aligning nodes across networks or noise in the data.
We then have that there exists an unknown (where $\Pi_n$ is the set of $n\times n$ permutation matrices)
$$\tbQ\in \Pi_{n,k}:=\{\bQ\in \Pi_n\text{ s.t. }\!\!\!\sum_{i\in U_{n,k}}\!\!\! Q_{ii}\leq k;\, \!\!\!\sum_{j\in M_{n,k}}\!\!\! Q_{jj}= n-k\}$$
such that the practitioner observes $\bA_1$ and $\bB_2=\tbQ\bA_2\tbQ^T$. 
Given the above framework, it is natural to consider the following semiparametric adaptation of the traditional parametric tests.
From $\bA_1$ and $\bB_2$, the user seeks to test if the distribution of $\bA_1$ is different than from that of $\bA_2$, i.e., to test the following hypotheses (where $\mathcal{L}(\bA_i)$ denotes the distribution (law) of $\bA_i$),
$H_0: {\mathcal{L}}(\bA_1)=\mathcal{L}(\bA_2)$ versus 
$H_1: \mathcal{L}(\bA_1)\neq\mathcal{L}(\bA_2).$
In this work, we will be considering testing within the family of (conditionally) edge-independent graph models, so that $\mathcal{L}(\bA_i)$ is completely determined by $\bP_i=\mathbb{E}(\bA_i)$.
Focusing our attention first on a simple Frobenius-norm based hypothesis test (later considering the spectral embedding based tests of \cite{levin2017central,tang2017semiparametric,Tang2017}), we reject $H_0$ if 
$$T=T(\bA_1,\bB_2):=\|\hbP_1-\hbP_{2,\widetilde{\bQ}}\|_F^2$$ 
is suitably large; here $\hbP_1$ is an estimate of $\bP_1$ obtained from $\bA_1$, and $\hbP_{2,\widetilde{\bQ}}$ an estimate of $\tbQ\bP_2\tbQ^T=\mathbb{E}(\bB_2)$
(this will be formalized later in Section \ref{sec:model}; for intuition and experimental validation on why the test statistic using $\hbP$ is preferable to the adjacency-based test using $\|\bA_1-\bB_2\|_F^2$ as its statistic, see Section \ref{sec:PhatvsA}).
To account for the uncertainty in the labeling of $\bB_2$, for each $\alpha>0$ and $\bQ\in \Pi_{n,k}$, define $c_{\alpha,\bQ}>0$ to be the smallest value such that 
$$\mathbb{P}_{H_0}(\|\hbP_1-\mathbf{Q}\hbP_2\mathbf{Q}^T\|_F\geq c_{\alpha,\bQ})\leq \alpha.$$
As we do not know which element of $\Pi_{n,k}$ yields the shuffling from $\bA_2$ to $\bB_2$, a valid (conservative) level-$\alpha$ test using the Frobenius norm test statistic
would reject $H_0$ if 
$$T(\bA_1,\bB_2) > \max_{\bQ\in\Pi_{n,k}}c_{\alpha,\bQ}.$$
The price of this validity is a loss in testing power against any fixed alternative, especially in the scenario where  $\tbQ$ (the true, but unknown, shuffling) shuffles fewer than $k$ vertices in $\bA_2$.
In this case the conservative test is over-correcting for the uncertainty in $U_{n,k}$.
The question we seek to answer is how much testing power is lost in this shuffle, and how robust the adaptations of different testing methods (i.e., different $T$ test statistics) are to this shuffling.

Note that our choice of Frobenius norm for the test statistic is natural here in light of the metric's extensive use for network comparison; see, e.g., its use in the estimation, testing, and matching literatures \cite{fishkind2019alignment,levin2019bootstrapping,lyzinski2018information}.

\subsection{Motivating example: DTMRI connectome testing}
\label{sec:motive}

To motivate the shuffled testing problem further, we first present the following real-data example.
We consider the test/retest connectomic dataset from \cite{zuo2014open} processed via the algorithmic pipeline at \cite{kiar2017high} (note that this data is available for download at \url{http://www.cis.jhu.edu/~parky/Microsoft/JHU-MSR/ZMx2/BNU1/DS01216-xyz.zip}).
This dataset represents human connectomes derived from DTMRI scans, where there are multiple (i.e., test/retest) scans per each of the 57 individuals in the study.
We consider three such scans, yielding connectomes $\textbf{A}_1$ and $\textbf{A}_2$ and $\textbf{A}_3$.
Here, $\textbf{A}_1$ and $\textbf{A}_2$ represent test/retest scans from one subject (subject 1) and $\textbf{A}_3$ a scan from a different subject (subject 2 scan 1).
In each scan, vertices represent voxel regions of the brain with edges denoting whether a neuronal fiber bundle connects the two regions or not (so that the graphs are binary and undirected).
Considering only vertices common to the three connectomes, we are left with three graphs each with $n=1085$ vertices.

For $k>0$, let $\bQ\in\Pi_{n,k}$ be an unknown permutation.
Observing $\textbf{A}_1$ and $\bA_2$ and $\bB_3=\textbf{Q}\bA_3\textbf{Q}^T$
(rather than $\bA_1,\bA_2,$ and $\bA_3$)
we seek to test whether $\bB_3$ is a connectome drawn from the same person as $\textbf{A}_1$ and $\bA_2$ or from a different person (i.e., under the reasonable assumption that $\mathcal{L}(\textbf{A}_1)=\mathcal{L}(\textbf{A}_2)$, we seek to test whether $\bA_3$ is from this same distribution).
Ideally, we would then construct our test statistic as 
\begin{align}
T(\bA_i,\bA_j)=\|\hbP_i-\hbP_j\|_F
\label{eq:teststat}
\end{align} where $\hbP_i$ is the estimate of $\bP_1=\mathbb{E}(\bA_1)=\mathbb{E}(\bA_2)=\bP_2$ or $\bP_3=\mathbb{E}(\bA_3)$ derived from $\bA_i$ as in Section \ref{sec:model}.

Incorporating the unknown shuffling of $U_{n,k}$ in $\bB_3$ is tricky here, as for moderate $k$ it is computationally infeasible to compute $c_{\alpha,\mathbf{Q}}$ for all $\mathbf{Q}\in\Pi_{n,k}$ (where we recall that $c_{\alpha,\mathbf{Q}}$ is the smallest value such that 
$\mathbb{P}(\|\hbP_1-\mathbf{Q}\hbP_2\mathbf{Q}^T\|_F^2\geq c_{\alpha,\bQ})\leq \alpha;$
indeed, the order of $\Pi_{n,k}$ is $k!$ given we know the vertices in $M_{n,k}$), and so it is difficult to compute the conservative critical value $\max_{\bQ\in\Pi_{n,k}} c_{\alpha,\mathbf{Q}}$.
Here, the task of finding the worse-case shuffling in $\Pi_{n,k}$ can be cast as finding an optimal matching between one graph and the complement of the second, which we suspect is computationally intractable.
To proceed forward, then, we consider the following
modification of the overall testing regime:  We consider a fixed (randomly chosen) sequence of nested permutations $\bQ_k\in\Pi_{n,k}$ for $k=0,\, 50,\, 150,\, 200,\,250,\,350$ and consider shuffling $\bA_2$ by $\bQ_k$ and $\bA_3$ by $\bQ_{\ell}$ for all $\ell\leq k$.
We then repeat this process $nMC=100$ times, each time obtaining an estimate (via bootstrapping as outlined below) of testing power against $H_0$.
This is done here out of computational necessity, and although the test does not achieve level-$\alpha$ here, this will nonetheless be sufficient to demonstrate the dramatic loss in testing performance due to shuffling.

\begin{figure}[t!]
\begin{center}
\includegraphics[width=1\textwidth]{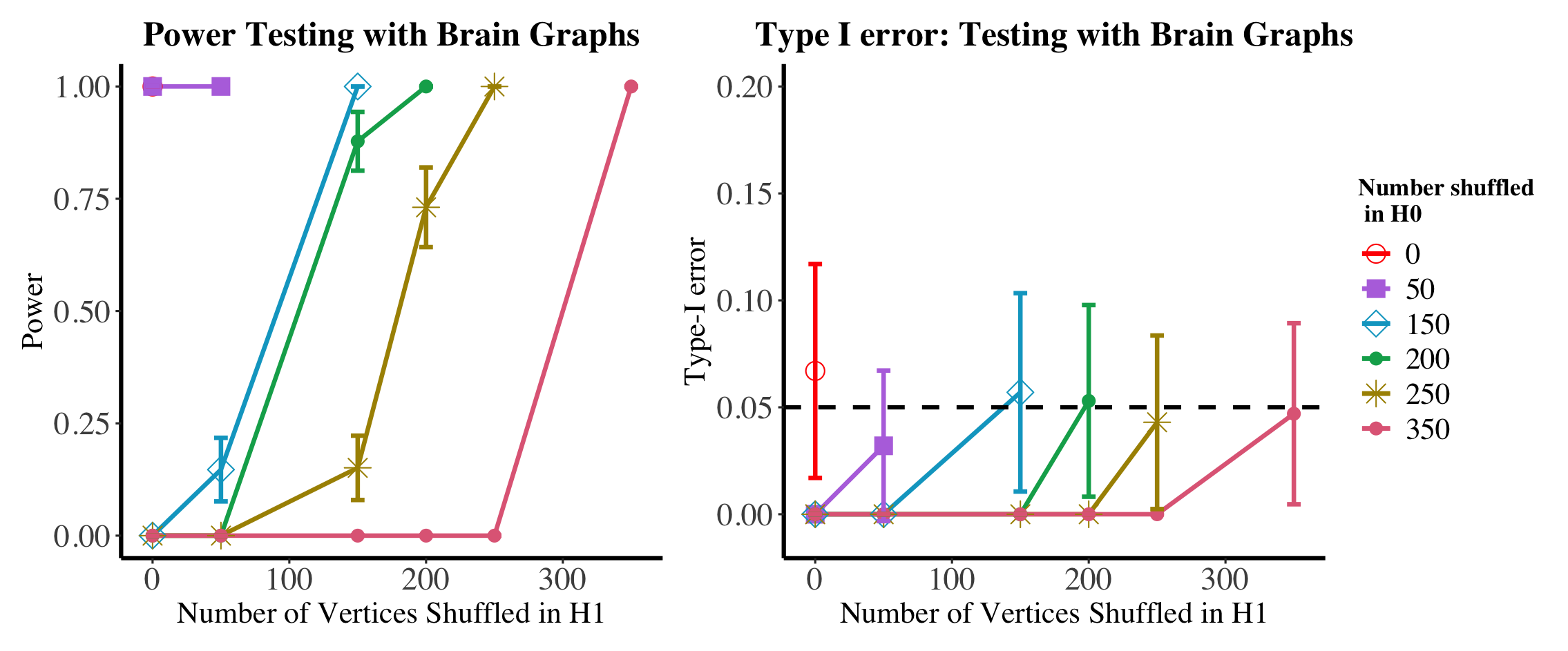} 
\caption{
Results for 200 bootstrapped samples of the test statistic in Eq. \ref{eq:teststat} at approximate level $\alpha=0.05$.
The $x$-axis represents the number of vertices shuffled via $\textbf{Q}_\ell$ (from $0$ to $k$) while the curve colors represent the maximum number of vertices potentially shuffled via $\Pi_{k,n}$, here all shuffled by $\bQ_k$.
The left panel displays testing power, the right type-I error; results are averaged over 100 Monte Carlo iterates (error bars are $\pm2$s.e.).}
\label{fig:brains}
\end{center}
\end{figure}

In order to estimate the testing power here, we rely on the bootstrapping heuristic inspired by \cite{levin2017central,levin2019bootstrapping} for latent position networks.
Informally, we will model $\bA_1$, $\bA_2$, and $\bA_3$ as instantiations of the Random Dot Product Graph (RDPG) modeling framework of \cite{young2007random}.
\begin{defn} 
\label{def:RDPG}
Let $\textbf{X}=[X_1|X_2|\cdots|X_n]^T\in\mathbb{R}^{n\times d}$ be such that $\langle X_i,X_j\rangle\in[0,1]$ for all $i,j\in[n]$.
We say that $\textbf{A} \sim RDPG(\textbf{X}, \nu)$ is an instance of a d-dimensional \emph{Random Dot Product Graph (RDPG)} with latent positions $\textbf{X}$ and sparsity parameter $\nu$ if given $\bX$, the entries of the random, symmetric, hollow adjacency matrix $\textbf{A} \in \{0,1\}^{n \times n}$ satisfy for all $i < j$, 
$    A_{ij} \overset{ind.}\sim  \text{Bernoulli}(\nu X_i^TX_j).$
\end{defn}
\noindent In this framework (where we here have $\nu=1)$, we posit matrices of latent positions $\bX\in\mathbb{R}^{n\times d}$, for $\bA_1$, $\bA_2$, and $\bY\in\mathbb{R}^{n\times d}$  for $\bA_3$ (with $\bX\bX^T, \bY\bY^T\in[0,1]^{n\times n}$), such that the $i$-th row of $\bX$, (resp., $\bY$) corresponds to the latent feature vector associated with the $i$-th vertex in $\bA_1$ and $\bA_2$ (resp., $\bA_3$).
In this setting, the distribution of networks from the same person (i.e., $\bA_1$ and $\bA_2$) will serve as our null distribution, and we seek to test  $H_0:\bA_3\sim \text{RDPG}(\bX,\nu=1)$, where $\bA_3\sim\text{RDPG}(\bY,\nu=1)$.

The RDPG model posits a tractable parameterization of the networks provided by the latent position matrices, and there are a number of statistically consistent methods for estimating these parameters under a variety of model variants.
Here, we will use the Adjacency Spectral Embedding (ASE) to estimate $\bX$ and $\bY$; see \cite{Athreya2018} for a survey of recent work in estimating and inference in RDPGs.
\begin{defn} 
\label{def:ASE}
(Adjacency Spectral Embedding) Given the adjacency matrix $\text{\textbf{A}} \in \{0,1\}^{n \times n}$ of an $n$-vertex graph, the \emph{Adjacency Spectral Embedding (ASE)} of $\textbf{A}$ into $\mathbb{R}^d$ is given by
\begin{equation} \label{ase}
   \text{ASE(\textbf{A}}, d) =  \hbX = U_AS_A^{1/2} \in \mathbb{R} ^{n \times d}
\end{equation}
where $ [U_A | \Tilde{U_A}][S_A \oplus \Tilde{S_A}][U_A | \Tilde{U_A}] \hspace{2pt}$ is the spectral decomposition of $|\text{\textbf{A}}| = (\text{\textbf{A}}^T \text{\textbf{A}})^{\frac{1}{2}},$ $S_A$ is the diagonal matrix with the ordered $d$ largest singular values of $\textbf{A}$ on its diagonal, and $U_A \in \mathbb{R}^{n \times d}$ is the matrix whose columns are the corresponding orthonormal eigenvectors of \textbf{A}.
\end{defn}
\noindent Once suitable estimates of the parameters are obtained---denoted $\widehat{\bX}_{1},\widehat{\bX}_{2}$ for those derived via ASE of $\bA_1,\bA_2$ respectively, and $\widehat{\bY}$ derived via ASE of $\bA_3$---we use a parametric bootstrap (here with 200 bootstrap samples) to estimate the null distribution critical value of
$T(\bA_1,\bA_2)=\|\hbP_1-\bQ_k\hbP_2\bQ_k^T\|_F,$ 
where in the $b$-th bootstrap sample, the test statistic $$T_b(\bA_1,\bA_2)=\|\hbP_{1,b}-\bQ_k\hbP_{2,b}\bQ_k^T\|_F$$ is computed as follows: 
For each $i=1,2$, 
\begin{itemize}
\item[(i)] sample independent $\bA_{i,b}\sim$RDPG$(\hbX_i,\nu=1)$; 
\item[(ii)] compute $\hbX_{i,b}=$ASE($G_{i,b},d$); 
\item[(iii)] set $\hbP_{i,b}=\hbX_{i,b}(\hbX_{i,b})^T$. 
\end{itemize}
Note that we use a single embedding dimension $d$ for all the ASE's in the null, estimated as detailed in Remark \ref{rem:pickd}.
Given this estimated critical value, we then estimate the testing level and testing power as follows.
For each $\ell\leq k$, we mimic the above procedure (again with 200 bootstrap samples) to estimate the distributions of 
$T^{(\ell)}(\bA_1,\bA_2)=\|\hbP_1-\bQ_\ell\hbP_2\bQ_\ell^T\|_F$
and 
$T^{(\ell)}(\bA_1,\bA_3)=\|\hbP_1-\bQ_\ell\hbP_3\bQ_\ell^T\|_F$
where $\hbP_2$ (resp., $\hbP_3$) is the ASE derived estimate of $\bP_2=\mathbb{E}(\bA_2)$ (resp., $\bP_3=\mathbb{E}(\bA_3)$); 
note that similar bootstrapping procedures were considered in \cite{tang2017semiparametric,levin2017central}.
The estimated testing level (resp., power) under the null hypothesis that $\mathcal{L}(\bA_1)=\mathcal{L}(\bA_2)$ (resp., under the alternative that $\mathcal{L}(\bA_1)=\mathcal{L}(\bA_3)$) is then computed by calculating the proportion of $T^{(\ell)}(\bA_1,\bA_2)$ (resp., $T^{(\ell)}(\bA_1,\bA_3)$) greater than the estimated critical value.

Estimated power for 200 bootstrapped samples of these test statistics at approximate level $\alpha=0.05$ are plotted in Figure \ref{fig:brains} left panel (averaged over the $nMC=100$ Monte Carlo replicates).
In the figure, the $x$-axis represents the number of vertices shuffled via $\textbf{Q}_\ell$ (from $0$ to $k$) while the curve colors represent the maximum number of vertices potentially shuffled via $\Pi_{k,n}$; here all shuffled by $\bQ_k$.
As seen in figure, the power of this test increases as both $k$ and $\ell$ increase, implying that the test is able to correctly distinguish the difference between the two subjects when the effect of the shuffling is either minimal (small $k$) or when the shuffling is equally severe in both the null and alternative cases (i.e., $\textbf{Q}_\ell$ shuffles as much as $\textbf{Q}_k$).
When $k$ is much bigger than $\ell$, the test is overly conservative (see Figure \ref{fig:brains} right panel), as expected.
In this case the shuffling in $H_0$ has the effect of inflating the critical value compared to the true (i.e., unshuffled) testing critical value, yielding an overly conservative test that cannot distinguish between the different test subjects.

%%%%%%%%%%%%%%%%%%%%%%%%%%%%%%%%%%%%%%%%%%%%
%%%%%%%%%%%%%%%%%%%%%%%%%%%%%%%%%%%%%%%%%%%%

\subsection{Random graph models}
\label{sec:models}

As referenced above, to tackle the question of power loss statistically, we will anchor our analysis in commonly studied random graph models from the literature.
In addition to the Random Dot Product Graph (RDPG) model \cite{hoff2002latent,young2007random} mentioned above, we will also consider the Stochastic Blockmodel (SBM) \cite{holland1983stochastic} as a data generating mechanism.
These models provide tractable settings for the analysis of graphs where the connectivity is driven by latent features---community membership in the SBM and the latent position vector in the RDPG.

The Stochastic Blockmodel---and its myriad variants including mixed membership \cite{airoldi2008mixed}, degree corrected \cite{karrer2011stochastic}, and hierarchical SBMs \cite{lyzinski2016community,li2020hierarchical}---provide a simple framework for networks with latent community structure.
\begin{defn}
\label{def:SBM}
We say that an $n$-vertex random graph $\bA\sim$ SBM$(K,\Lambda, b,
\nu)$ is distributed according to a stochastic block model random graph with parameters $K\in\mathbb{Z}^+$ the number of blocks in graph, 
$ \Lambda\in[0,1]^{K\times K}$ the block probability matrix, $b\in\mathbb{Z}^K$ the block membership function, and $\nu$ the sparsity parameter, if
\begin{itemize}
    \item[i.] The vertex set $V$ is partitioned into $K$ blocks 
    $V=V_1\sqcup V_2\sqcup\cdots\sqcup V_K,$ where for each $i\in[K]$,
    we have $|V_i|= n_i$ denotes the size of the $i^{th}$ block (so that $\sum_{i=1}^K n_i=n$);
    \item[ii.] The block membership function $b:V\mapsto K$ is such that ${b}(v)=i$ iff $v\in V_i$, and we have for each $\{u,v\}\in\binom{V}{2}$,
    $A_{uv}\stackrel{ind.}{\sim}\text{Bernoulli}(\nu\Lambda_{{b}(u),{b}(v)}).$
\end{itemize}
\end{defn}
\noindent Note that the block membership vector in an SBM is often modeled as a random multinomial vector with block probability parameter $\vec\pi\in\mathbb{R}^K$ giving the probabilities of assigning vertices randomly to each of the $K$ blocks.
Our analysis is done in the fixed block membership setting, although it translates immediately to the random membership setting.
Note also that we will often be considering cases where the number of vertices in $G$ satisfies $n\rightarrow\infty$.
In this case, we write $G\sim\text{SBM}(K_n,\Lambda_n,b_n,\nu_n)$ so that the model parameters may vary in $n$. However, to ease notation, we will suppress the $n$ subscript throughout, although the dependence on $n$ is implicitly understood.

In SBMs, the connectivity structure is driven by the latent community membership of the vertices.
In the space of latent feature models, a natural extension of this idea is to have connectivity modeled as a function of more nuanced, vertex-level, features.
In this direction, we will also consider framing our inference in the popular Random Dot Product Graph model introduced in Definition
\ref{def:RDPG}.
Note that the RDPG model encompasses SBM models with positive semidefinite $\Lambda$.
Indefinite and negative definite SBM's, are encompassed via the generalized RDPG \cite{Rubin-Delanchy2017}, though the ordinary RDPG will be sufficient for our present purposes.
Note also that our theory will be presented for the fixed latent position RDPG above, though it translates immediately to the random latent position setting (i.e., where the rows of $\mathbf{X}$, namely the $X_i$, are i.i.d. from an appropriate distribution $F$).
\begin{rem}
An inherent non-identifiability of the RDPG model comes from the fact that for any orthogonal matrix $\textbf{W}\in\mathcal{O}_d$, we get $\mathbf{A}|\mathbf{X}\stackrel{\mathcal{L}}{=}\mathbf{A}|(\mathbf{XW})$. 
With this caveat, RDPGs are more suitable for modeling in inference tasks that are rotation invariant, such as clustering \cite{Sussman2012,rohe2011spectral}, classification \cite{tang2013universally}, and appropriately defined hypothesis testing settings \cite{tang2017semiparametric,levin2017central}.
\end{rem} 
\begin{rem}
If the RDPG graphs are directed or weighted, then appropriate modifications to the ASE are required to embed the networks (see, for example, \cite{Sussman2014} and \cite{gallagher2023spectral}).
Analogous concentration results are available in both settings, and we suspect that we can derive theory analogous to Theorems \ref{thm:power2}--\ref{thm:power3}.  
Herein, we restrict ourselves to the unweighted RDPG, and leave the necessary modification to handle directed and weighted graph to future work.
\end{rem}
As in \cite{du2023hypothesis}, we will choose to control the sparsity of our graphs via $\nu$ and not through the latent position matrix $\mathbf{X}$ or block probability matrix $\Lambda$.
As such, we will implicitly make the following assumption throughout the remainder for all RDPGs and positive semidefinite SBMs (when viewed as RDPGs):
\begin{ass}
\label{ass:ass1}
If we consider a random graph sequence $\bA_n\sim\text{RDPG}(\bX_n,\nu_n)$ where $\bX_n\in\mathbb{R}^{n\times d}$, then we will assume that for all $n$ sufficiently large, we have that:
\begin{itemize}
    \item[i.] $\bX_n$ is rank $d$, and if $\sigma_1(\bX_n)\geq 
    \sigma_2(\bX_n)\geq \cdots\geq
    \sigma_d(\bX_n)$ are the singular values of $\bX_n$, we have $\sigma_1(\bX_n)\approx \sigma_d(\bX_n)=\Theta(n)$;
    \item[ii.] There exists a fixed compact set $\mathcal{X}$ such that the rows of $\bX_n$ are in $\mathcal{X}$ for all $n$;
    \item[iii.] There exists a fixed constant $a>0$ such that $\bX_n\bX_n^T\geq a$ entry-wise.
\end{itemize} 
\end{ass}
%%%%%%%%%%%%%%%%%%%%%%%%%%%%%%%%%%%%%%%%%%%%
%%%%%%%%%%%%%%%%%%%%%%%%%%%%%%%%%%%%%%%%%%%%

\subsection{Model estimation}
\label{sec:model}

In the RDPG (and positive semidefinite SBM) setting, our initial hypothesis test will be predicated upon having a suitable estimate of $\bP=\mathbb{E}(\bA|\bX)$.
In this setting, the Adjacency Spectral Embedding (ASE) (see Definition \ref{def:ASE}) of \cite{Sussman2012} has proven to be practically useful and theoretically tractable means for obtaining such an estimate.
% \begin{defn} 
% \label{def:ASE}{}
% (Adjacency Spectral Embedding) Given the adjacency matrix $\text{\textbf{A}} \in \{0,1\}^{n \times n}$ of an $n$-vertex graph, the \emph{Adjacency Spectral Embedding (ASE)} of $\textbf{A}$ into $\mathbb{R}^d$ is given by
% \begin{equation} \label{ase}
%    \text{ASE(\textbf{A}}, d) =  \hbX = U_AS_A^{1/2} \in \mathbb{R} ^{n \times d}
% \end{equation}
% where $ [U_A | \Tilde{U_A}][S_A \oplus \Tilde{S_A}][U_A | \Tilde{U_A}] \hspace{2pt}$ is the spectral decomposition of $|\text{\textbf{A}}| = (\text{\textbf{A}}^T \text{\textbf{A}})^{\frac{1}{2}},$ $S_A$ is the diagonal matrix with the ordered $d$ largest singular values of $\textbf{A}$ on its diagonal, and $U_A \in \mathbb{R}^{n \times d}$ is the matrix whose columns are the corresponding orthonormal eigenvectors of \textbf{A}.
% \end{defn}
The adjacency spectral embedding has a rich, recent history in the literature (see \cite{Athreya2018}) as a tool for estimating tractable graph representations, achieving its greatest estimation strength in the class of latent position networks (the RDPG being one such example).
In these settings, it is often assumed that the rank of the latent position matrix \textbf{X} is $d$, and that $d$ is considerably smaller than $n$, the number of vertices in the graph.

A great amount of inference in the RDPG setting is predicated upon $\hbX$ being a suitably accurate estimate of 
$\bX$.
To this end, the key statistical properties of consistency and asymptotic residual normality are established for the ASE in \cite{Sussman2012,Lyzinski2014,Rubin-Delanchy2017} and \cite{Athreya2013,Athreya2018} respectively.
These results (and analogues for unscaled variants of the ASE) have laid the groundwork for myriad subsequent inference results, including clustering \cite{Sussman2012,Lyzinski2014, lei2015consistency,sanna2021spectral}, classification \cite{tang2013universally}, time-series analysis \cite{chen2020multiple,Pantazis2020}, and vertex nomination \cite{fishkind2015vertex,yoder2020vertex}, among others.

\begin{rem}
\label{rem:pickd}
In practice, there are a number of heuristics for estimating the unknown embedding dimension $d$ in the ASE (see, for example, the work in \cite{fishkind2013consistent,chatterjee2015matrix,li2020network}). 
In the real data experiments below, we will adopt an automated elbow-finder applied to the SCREE plot (as motivated by
\cite{zhu2006automatic} and \cite{chatterjee2015matrix}); for the simulation experiments, we use the true $d$ value for the underlying RDPG's/SBM's.
Estimating the correct dimension $d$ is of paramount importance in spectral graph inference, as underestimating $d$ introduces bias into the embedding estimate and overestimating $d$ introduces additional variance into the estimate.
In our experience, underestimation of $d$ would have a more dramatic impact on subsequent inferential performance.

\end{rem}

\section{Shuffled graph testing in theory}
\label{sec:theory}
In complicated testing regimes (e.g., the embedding-based tests of \cite{tang2017semiparametric,levin2017central,asta2014geometric}), analyzing the distribution of the test under the alternative is itself a challenging proposition (see, for example, the work in \cite{Draves2020,tang2017semiparametric}).
Accounting for a second layer of uncertainty due to the shuffling adds further complexity to the analysis.
In order to build intuition for these more complex settings in the context of the RDPG and SBM models (which we will explore empirically in Section \ref{sec:ASE}), we examine the effect on testing power of shuffling in the simple Frobenius norm test considered in Section \ref{sec:motive}.

We consider first the case where $\bA_1\sim$RDPG$(\bX_n, \nu_n)$ and we have an independent $\bA_2\sim$RDPG$(\bY_n, \nu_n)$.
Under the null $\bX_n=\bY_n$, and we will consider elements of the alternative that satisfy the following:  for all but $r=r_n$ rows of $\bY_n$, we have $\bY_n[i,:]=\bX_n[i,:]$ so that we have
$\bY\bY^T=\bX\bX^T+\bE$
where (with the proper vertex reordering) 
\begin{align}
\label{eq:form_of_E}
\bE=
\begin{pmatrix}
\mathbf{E}_{r}&\mathbf{E}_{r}'\\
(\mathbf{E}_{r}')^T&\textbf{0}_{n-r}
\end{pmatrix}
\end{align}
We will further assume that there exists  constants $c_2>c_1>0$ and $\epsilon_n=\epsilon>0$ such that 
$c_1\epsilon\leq |e_{ij}|\leq c_2\epsilon$
for all entries of $\mathbf{E}_{r}$ and $\mathbf{E}_{r}'$.
We note that we will assume throughout that both $\bX_n$ and $\bY_n$ satisfy the conditions of Assumption \ref{ass:ass1}.

The principle challenge of testing in this regime is that the veracity of the (across graph) labels of vertices in $U_{n,k}$ is unknown a priori.  It could be the case that these vertices were all shuffled or all correctly aligned, and it is difficult to disentangle the effect on testing power of $\mathbf{E}$ versus the potential shuffling.
To model this, we consider shuffled elements of the alternative, so that we observe $\bA_1$ and $\bB_2=\widetilde\bQ\bA_2(\widetilde\bQ)^T$, where the true but unknown shuffling of $\bA_2$ is $\widetilde\bQ$, which shuffles $\ell\leq k$ labels in $U_{n,k}$.  

\subsection{Power analysis and the effect of shuffling: $\widehat P$ test}
\label{sec:pwr1}

In this section, we will present a trio of theorems, namely Theorems \ref{thm:power_to_1}--\ref{thm:power_to_0}, in which we characterize the impact on power of the two distinct sources of noise here: the shuffling error ($k$ and $\ell$) and the error in the alternative captured here by $\epsilon$.
When the difference between $k$ and $\ell$ is comparably large (for example, when $\ell\leq r$ $\epsilon\ll\sqrt{(k-\ell)/r}$ and
$\epsilon\ll \frac{k-\ell}{\ell}$), 
then the power of the resulting test will be low even in the presence of modest error $\epsilon$. 
In this case, the relative size of the error in the alternative is overwhelmed by the excess shuffling in the null which is needed to maintain testing level $\alpha$.
The actual shuffling error (i.e., $\ell$) is much less than the conservative null shuffling (i.e., $k$), and the test is not able to distinguish the two graphs in light of the conservative test's overcompensation.
Even in the case where $k-\ell$ is relatively small, if the error $\epsilon$ is sufficiently small, we will have low testing power, as expected.
However, when the difference between $k$ and $\ell$ is relatively small compared to $\epsilon$, or $k$ and $\ell$ are both relatively small compared to $\epsilon$ (see the conditions in Theorem \ref{thm:power_to_1} and Theorem \ref{thm:power_to_12}), then the difference in the number of vertices being shuffled across the conservative null and truly shuffled in the alternative is overwhelmed by the error in the alternative. 
In this case, the noise created by the relatively small differences in shuffling between null and alternative can be overcome, and high power can still be achieved.

\subsubsection{Small $k-\ell$ regime}
\label{sec:kl_small}
Before presenting the trio of theorems, we will first establish the following notation 
\begin{itemize}
\item[i.] For $i=1,2,$ let $\hbP_i$ be
the ASE-based estimate of $\bP_i$ derived from $\bA_i$;
\item[ii.] For $i=1,2,$ and for any $\bQ\in\Pi_{n,k}$, let $\hbP_{i,\bQ}$ be
the ASE-based estimate of  $\bP_{i,\bQ}=\bQ\bP_i\bQ^T$ derived from $\bQ\bA_i\bQ^T$;
\item[iii.] Let $\bQ^*\in\text{argmax}_{Q\in\Pi_{n,k}}\|P_1-P_{1,\bQ}\|_F$; let $\widetilde \bQ$ be the shuffling of $\ell\leq k$ vertices in $U_{n,k}$ such that we observe $\bB_2=\widetilde \bQ \bA_2\widetilde \bQ^T$.
\end{itemize}
In this section, we will be concerned with conditions under which power is asymptotically almost surely 1; specifically conditions under which the following holds for all $n$ sufficiently large
\begin{equation}
\label{eq:power_to_1_2m}
\mathbb{P}_{H_1}\left(\|\hbP_1-\hbP_{2,\tilde\bQ}\|_F>\max_{\bQ\in\Pi_{n,k}}c_{\alpha,\bQ}\right)\geq 1-n^{-2}.
\end{equation}
Our first result tackles the case in which $k$ is relatively small, and only modest error $\epsilon$ is needed to achieve high testing power.
Note that the proof of Theorem \ref{thm:power_to_1} can be found in Appendix \ref{pf:power_to_1}.
\begin{theorem}
\label{thm:power_to_1}
With notation as above, assume there exist $\alpha\in(0,1]$ such that $r=\Theta(n^{\alpha})$ and $k,\ell\ll n^{\alpha}$, and that $\frac{\|\bP_1-\bP_{1,\bQ^*}\|_F^2-\|\bP_1-\bP_{1,\widetilde\bQ}\|_F^2 }{\nu^2 n}=O(k)$.
In the sparse setting, consider $\nu\gg\frac{\log^{4c}(n)}{n^\beta}$ for $\beta\in(0,1]$ where $\alpha\geq\beta$.  If either \begin{itemize}
\item[i.]$k=O\left(\frac{n^\beta}{\log^{2c}n}\right)$ 
and
$\epsilon\gg  \sqrt{\frac{n^{\beta-\alpha}}{\log^{2c}(n)}}$; or
\item[ii.] $k\gg \frac{n^\beta}{\log^{2c}(n)}$ and $\epsilon\gg \sqrt{\frac{k}{n^\alpha}}$
\end{itemize}
then Eq. \ref{eq:power_to_1_2m} holds for all $n$ sufficiently large.
In the dense case where $\nu=1$, if either
\begin{itemize}
\item[i.]$k\gg\log^{2c}(n)\text{ and }\epsilon\gg \sqrt{k/n^{\alpha}};$ or 
\item[ii.]$k\ll\log^{2c}(n)\text{ and }\epsilon\gg\sqrt{(\log^{c}(n))/ n^{\alpha}},$
\end{itemize}
then Eq. \ref{eq:power_to_1_2m} holds for all $n$ sufficiently large.
\end{theorem}
\noindent Note that $\|\bP_1-\bP_{1,\bQ^*}\|_F^2=O(nk\nu^2)$ so the assumption on the growth rate of $\|\bP_1-\bP_{1,\bQ^*}\|_F^2-\|\bP_1-\bP_{1,\widetilde\bQ}\|_F^2$
in Theorem \ref{thm:power_to_1} considers the case where the shuffling due to $\ell$ does not compensate for the shuffling due to $k$, and the shuffling due to $k$ needs to be relatively minor to achieve the desired power (here, the growth rates on $\epsilon$ in terms of $k$). 

Our second result tackles the case in which $k-\ell$ is relatively small, and only modest error $\epsilon$ is needed to achieve high testing power.
The proof of Theorem \ref{thm:power_to_12} can be found in Appendix \ref{pf:power_to_1}.
\begin{theorem}
\label{thm:power_to_12}
With notation as above, assume there exist $\alpha\in(0,1]$ such that $r=\Theta(n^{\alpha})$ and $k,\ell\ll n^{\alpha}$, and that $\frac{\|\bP_1-\bP_{1,\bQ^*}\|_F^2-\|\bP_1-\bP_{1,\widetilde\bQ}\|_F^2 }{\nu^2 n}=O(k-\ell)$.
In the sparse setting where $\nu\gg\frac{\log^{4c}(n)}{n^\beta}$ for $\beta\in(0,1]$ where $\alpha\geq\beta$, if $k\gg \frac{n^\beta}{\log^{2c}(n)}$ and either
\begin{itemize}
\item[i.]  $\frac{k-\ell}{k^{1/2}}\geq\frac{n^{\beta/2}}{\log^{2c}(n)}$; $\epsilon\gg \frac{\ell}{n^\alpha}$; and $\epsilon\gg \sqrt{\frac{k-\ell}{n^\alpha}}$; or
\item[ii.] $\frac{k-\ell}{k^{1/2}}\leq\frac{n^{\beta/2}}{\log^{2c}(n)}$;  
$\epsilon\gg\frac{\ell}{n^\alpha}$; and $\epsilon\gg\sqrt{\frac{n^{\beta/2}}{\log^{2c}(n)}\frac{k^{1/2} }{n^{\alpha}}}$
\end{itemize}
then Eq. \ref{eq:power_to_1_2m} holds for all $n$ sufficiently large.
In the dense case where $\nu=1$ and $k=\omega(\log^{2c}n)$, if either
\begin{itemize}
\item[i.] $\frac{k-\ell}{k^{1/2}}\geq\log^{c}(n)$; $\epsilon\gg \frac{\ell}{n^\alpha}$; and $\epsilon\gg \sqrt{\frac{k-\ell}{n^{\alpha}}}$; or
\item[ii.] $\frac{k-\ell}{k^{1/2}}\leq\log^{c}(n)$; $\epsilon\gg \frac{\ell}{n^\alpha}$; and $\epsilon\gg \sqrt{\frac{k^{1/2}\log^{c}(n) }{n^{\alpha}}}$,
\end{itemize}
then Eq. \ref{eq:power_to_1_2m} holds for all $n$ sufficiently large.
\end{theorem}
\noindent 
The assumption on the growth rate of $\|\bP_1-\bP_{1,\bQ^*}\|_F^2-\|\bP_1-\bP_{1,\widetilde\bQ}\|_F^2$
in Theorem \ref{thm:power_to_12} considers the case where the shuffling due to $\ell$ can compensate for the shuffling due to $k$ (i.e., when $k-\ell\ll k$).
In this setting, it is possible to achieve the desired power in alternative regimes with significantly smaller $\epsilon$.
Under mild assumptions, this growth rate condition will hold, for example, in the SBM where the shuffling is across blocks; see Section \ref{sec:SBM_test}.

\subsubsection{Shuffled graph testing in SBMs}
\label{sec:SBM_test}

We consider next the case where $\bA_1\sim$SBM$(K,\Lambda, b, \nu)$, and 
we assume there exists a matrix $\mathbf{E}=[e_{ij}]\in\mathbb{R}^{n\times n}$ of the form (up to vertex reordering) of Eq. \ref{eq:form_of_E}.
such that, under $H_1$, $\bA_2=[A_{2,ij}]$ is an independently sampled graph with independently drawn edges sampled according to
$$A_{2,ij}\stackrel{\text{ind.}}{\sim} \text{Bernoulli}\left(\nu\left[ \Lambda_{b(i),b(j)}+ e_{ij}\right]\right).$$
We will consider here $\bE$ being block-structured (in which case $\bA_2$ is itself an SBM).  As before, we will assume that there exists  constants $c_2>c_1>0$ and $\epsilon_n=\epsilon>0$ such that 
$c_1\epsilon\leq |e_{ij}|\leq c_2\epsilon$
for all entries of $\bE_r, \bE_r'$.

Consider the setting where
$U_{n,k}\subset V_1\cup V_2,$ and $|U_{n,k}\cap V_1|=|U_{n,k}\cap V_2|=k/2,$
so that at most $k$ vertices have shuffled labels and $k/2$ of these are in each of blocks 1 and 2.
Note that, as block labels are arbitrary, this captures the setting where vertices may be flipped between any two different blocks.
In what follows below (see Proposition \ref{prop:crit-value} in the Appendix), we will see that we can bound $\max_{\bQ\in\Pi_{n,k}}c_{\alpha,\bQ}$ in terms of any permutation that interchange exactly $k/2$ vertices between blocks 1 and 2.
Without loss of generality we can then bound $\max_{\bQ\in\Pi_{n,k}}c_{\alpha,\bQ}$ in terms of $\bQ_k$ defined via
$$ \bQ_k = 
\begin{pmatrix}
\mathbf{R}_k&\textbf{0}_{n_1+n_2, n-n_1-n_2}\\
\textbf{0}_{n-n_1-n_2, n_1+n_2}&\textbf{I}_{n-n_1-n_2}
\end{pmatrix}
$$
where
\[ \mathbf{R}_k = 
\begin{pmatrix}
  \begin{matrix}
  \textbf{0}_{k/2, k/2}  & \textbf{0}_{k/2, n_1-k/2} \\
  \textbf{0}_{n_1-k/2,k/2} & \textbf{I}_{n_1-k/2}
  \end{matrix}
  & \rvline & \begin{matrix}
  \textbf{I}_{k/2} & \textbf{0}_{k/2, n_2-k/2} \\
  \textbf{0}_{n_1-k/2,k/2} & \textbf{0}_{n_1-k/2, n_2-k/2}
  \end{matrix} \\
\hline
 \begin{matrix}
  \textbf{I}_{k/2} & \textbf{0}_{k/2, n_1-k/2} \\
  \textbf{0}_{n_2-k/2,k/2} & \textbf{0}_{n_2-k/2, n_1-k/2}
  \end{matrix}  & \rvline & 
  \begin{matrix}
  \textbf{0}_{k/2, k/2} & \textbf{0}_{k/2, n_2-k/2} \\
  \textbf{0}_{n_2-k/2,k/2} & \textbf{I}_{n_2-k/2}
  \end{matrix} 
  \end{pmatrix}\]
\noindent We again consider shuffled elements of the alternative, so that we observe $\bA_1$ and $\bB_2=\bQ_\ell\bA_2(\bQ_\ell)^T$, where $\bQ_\ell$ is defined analogously to $\bQ_k$ (i.e., for any $h\leq k$, $\bQ_h$ shuffles the first $h/2$ vertices between blocks 1 and 2).
In this SBM setting, note that 
\begin{align*}
\|\bP_1&\!-\!\bQ_k\bP_1(\bQ_k)^T\|_F^2=
k^2\nu^2(\Lambda_{11}\!-\!\Lambda_{22})^2/2 + 2k\sum_{i=3}^K n_i \nu^2(\Lambda_{i1}-\Lambda_{i2})^2\\
&+ 2k(n_1-k/2)\nu^2(\Lambda_{11}\!-\!\Lambda_{12})^2 + 2k(n_2-k/2)\nu^2(\Lambda_{22}-\Lambda_{12})^2\\
=&k^2\nu^2(\Lambda_{11}-\Lambda_{22})^2/2 + 2k\sum_{i=1}^K n_i \nu^2(\Lambda_{i1}-\Lambda_{i2})^2\\
&- k^2\nu^2(\Lambda_{11}-\Lambda_{12})^2 - k^2\nu^2(\Lambda_{22}-\Lambda_{12})^2
\end{align*}
If $n_i=\Theta(n)$ for each $i\in[K]$ and $k,\ell\ll n$ (as assumed in Theorems \ref{thm:power2} and \ref{thm:power3}), then $\|\bP_1-\bQ_k\bP_1(\bQ_k)^T\|_F^2-\|\bP_1-\bQ_\ell\bP_1(\bQ_\ell)^T\|_F^2=\Theta(n\nu^2(k-\ell))$ under mild assumptions on $\Lambda$, and Theorem \ref{thm:power_to_12} applies.

\subsubsection{Large $k-\ell$ regime}
\label{sec:kl_big}

We next tackle the power lost by an overly conservative test (i.e., when $k$ is much bigger than $\ell$).
In this case, it is reasonable to expect the power of the resulting test to be small,
as in this setting the shuffling noise could hide the true discriminatory signal in the alternative (here presented by $\bE$).  
Note that the proof of Theorem \ref{thm:power_to_0} can be found in Appendix \ref{pf:power_to_0}.

\begin{theorem}
\label{thm:power_to_0}
With notation as in Section \ref{sec:kl_small}, assume that 
$$\frac{\|\bP_1-\bP_{1,\bQ^*}\|_F^2-\|\bP_1-\bP_{1,\widetilde\bQ}\|_F^2 }{\nu^2 n}=\Omega(k-\ell)$$ and $r=\Theta(n^{\alpha}).$
Suppose further that
$\frac{k-\ell}{\sqrt{k}}\gg \frac{\log^cn}{\sqrt{\nu}}.$
Then if either
\begin{itemize}
 \item[i.] $\epsilon\ll  \frac{k-\ell}{n^\alpha}\text{ and }\ell\geq r$; or
 \item[ii.]$\epsilon\ll  \sqrt{\frac{k-\ell}{n^\alpha}};\,\, \epsilon\ll  \frac{k-\ell}{\ell};\text{ and }\ell\leq r,$
\end{itemize}
we have that for all $n$ sufficiently large
\begin{equation}
\label{eq:pwr_to_0}
\mathbb{P}_{H_1}\left(\|\hbP_1-\bQ_\ell\hbP_{2}\bQ_\ell^T\|_F>\max_{\bQ\in\Pi_{n,k}}c_{\alpha,\bQ}\right)\leq n^{-2}.
\end{equation}
\end{theorem}
\noindent Note again that under mild assumptions, the growth rate requirements of $\|\bP_1-\bP_{1,\bQ^*}\|_F^2-\|\bP_1-\bP_{1,\widetilde\bQ}\|_F^2$ holds in the SBM setting considered in Section \ref{sec:SBM_test}, and Theorem \ref{thm:power_to_0} holds (given the growth rate of $k,r$).

\section{$\widehat \bP$ versus A in the Frobenius test}
\label{sec:PhatvsA}

The issue that is at the heart of the problem with the Frobenius-norm test using adjacency matrices (rather than using $\hbP$) can be best understood via the following simple example:
\begin{example}
\label{ex:badA}
Consider the simple setting where we have independent random variables 
\begin{align*}
X\sim\text{Bernoulli}(p);\,\,
Y\sim\text{Bernoulli}(p);\,\,
Z\sim\text{Bernoulli}(q).
\end{align*}
In this case
\begin{align*}
\mathbb{E}|X-Y|&=2p(1-p)\\
\mathbb{E}|X-Z|&=p(1-q)+q(1-p)
\end{align*}
Note that 
\begin{align*}
    p(1-q)+q(1-p) - 2p(1-p) =
    p(p-q) + (q-p)(1-p)=(q-p)(1-2p) 
\end{align*}
is greater than $0$ when $q>p$ and $p<1/2$, or when $q<p$ and $p>1/2$; and is less than $0$ when $q>p$ and $p>1/2$, or when $q<p$ and $p<1/2$.

Consider next the task of testing
$H_0:\mathcal{L}(\bA)=\mathcal{L}(\bB)$ for a pair of graphs $\bA$ and $\bB$.  A natural first test statistic to use is $T=\|\bA-\bB\|_F^2$, and it is natural to then reject the null when $T$ is relatively large.
In the case where
$\bA\sim$ER$(n,p)$ (i.e., all edges appear in $\bA$ with i.i.d. probability $p$ independent of all other edges) and $\bB\sim$ER$(n,q)$, the test becomes $H_0:p=q$.
However, under $H_0$ we have $\mathbb{E}T=n(n-1)2p(1-p)$ and under the alternative 
$\mathbb{E}T=n(n-1)(p(1-q)+q(1-p))$.
If $p<1/2$, then $p>q$ implies $\mathbb{E}_1T<\mathbb{E}_0T$, and rejecting for large values of $T$ would fail to reject for this range of alternatives.
Of course, in the homogeneous Erd\H os-R\'enyi (ER) case, we would want a two-sided rejection region (or we can appropriately scale $T$ to render a one-sided test suitable), though in heterogeneous 
ER models, adapting $T$ is more nuanced as we shall show below.
While the test using $T=\|\hbP_1-\hbP_2\|_F^2$ does not suffer from this particular quirk, we do not claim it is the optimal test in the low-rank heterogeneous 
ER model.
Indeed, we suspect the more direct spectral tests of \cite{tang2017semiparametric,levin2017central} would be more effective, though the effect of the shuffling is more nuanced in those tests. There, it is considerably more difficult to disentangle the shuffling from the embedding alignment steps of the testing regimes (Procrustes alignment in \cite{tang2017semiparametric}, and the Omnibus construction in \cite{levin2017central}).
\end{example}

Note that, as here we are working with a $2$-block SBM setting of Section \ref{sec:SBM_test}, we adopt the notation $Q_{2h}$ for $0\leq h\leq k$ to emphasize that $2h$ total vertices are being shuffled, with $h$ coming from each block. 
For $i=1,2$ and all $0\leq h\leq k$, let 
$\bA_{i,h}$ be shorthand for $\bQ_{2h}\bA_i\bQ_{2h}^T$ and $\bE_{h}=[e_{i,j}^{(h)}]$ be shorthand for $\bQ_{2h}\bE\bQ_{2h}^T$.
Consider the hypothesis test for testing
$H_0:\mathcal{L}(\bA_1)=\mathcal{L}(\bA_2)$ using the test statistic (where $\mathcal{G}_n$ is the set of all $n$-vertex undirected graphs)
$T_A:\mathcal{G}_n\times \mathcal{G}_n\mapsto \mathbb{R}^{\geq 0}$
defined via $T_A(\bA_1,\bA_2):=\frac{1}{2}\|\bA_{1}-\bA_{2}\|_F^2$.

Assume that we are in the dense setting (i.e., $\nu_n= 1$ for all $n$) and that the following holds:
\begin{itemize}
\item[i.] There exists an $\eta\in(0,1/2)$ such that $\eta\leq \Lambda\leq1-\eta$ entry-wise;
\item[ii.] There exists a $\tilde\eta\in[0,\eta)$ such that for all $\{ij\}$, $|e_{ij}|\leq\tilde\eta$;
\item[iii.] $\min_i n_i=\Theta(n)$, and $\max_i |\Lambda_{1i}-\Lambda_{2i}|,|\Lambda_{11}-\Lambda_{22}|=\Theta(1).$
\end{itemize}
In this case, we have that $T_A(\bA_1,\bA_{2,k})$ is stochastically greater than $T_A(\bA_1,\bA_{2,h})$ for $h<k$, and so the conservative level $\alpha$ test---to account for the uncertainty in $U_{n,2k}$---using $T_A$ would reject $H_0$ if
$T_A(\bA_1,\bA_{3,\ell})> \mathfrak{c}_{\alpha,k}$
where $\mathfrak{c}_{\alpha,k}$ is the smallest value such that 
$\mathbb{P}_{H_0}(T_A(\bA_1,\bA_{2,k})> \mathfrak{c}_{\alpha,k})\leq \alpha.$
As the following proposition shows (proven in Appendix \ref{sec:APpf}), the decay of power for this adjacency-based test exhibits pathologies not present in the $\hbP$-based test (where $\sum_{\{ij\}}$ denotes the sum over unordered pairs of elements of $[n]$, and $n_*=\min_i n_i$, and $\delta:=\max_i |\Lambda_{1i}-\Lambda_{2i}|$ and $\gamma:=|\Lambda_{11}-\Lambda_{22}|$).
\begin{prop}
\label{prop:AP}
With notation as above, let $r=n$ and define $\xi_{ij}:=(2p^{(1)}_{ij}-1)e^{(\ell)}_{ij}$ and 
$\mu_\xi:=\sum_{\{ij\}}\xi_{ij}$.
We have that 
$$\mathbb{P}_{H_1}(T_A(\bA_1,\bA_{2,\ell})\geq \mathfrak{c}_{\alpha,k})=o(1)$$
if 
$$(k-\ell)\frac{n_*}{n}-\frac{k^2-\ell^2}{n}\delta^2-\frac{k-\ell}{n}\gamma^2+\frac{\mu_{\xi}}{n}=\omega(1).$$
\end{prop}
Digging a bit deeper into this proposition, we see the phenomena of Example \ref{ex:badA} at play.  
Even when $k-\ell$ is relatively small, if sufficiently often we have that
\begin{align}
\label{eq:ep1}
e^{(\ell)}_{ij}<0&\text{ when }p^{(1)}_{ij}<1/2\\
\label{eq:ep2}
e^{(\ell)}_{ij}>0&\text{ when }p^{(1)}_{ij}>1/2
\end{align}
then $\frac{\mu_{\xi}}{n}$ can itself be positive and divergent, driving power to $0$.

\begin{figure*}[t!]
\includegraphics[width=1\textwidth]{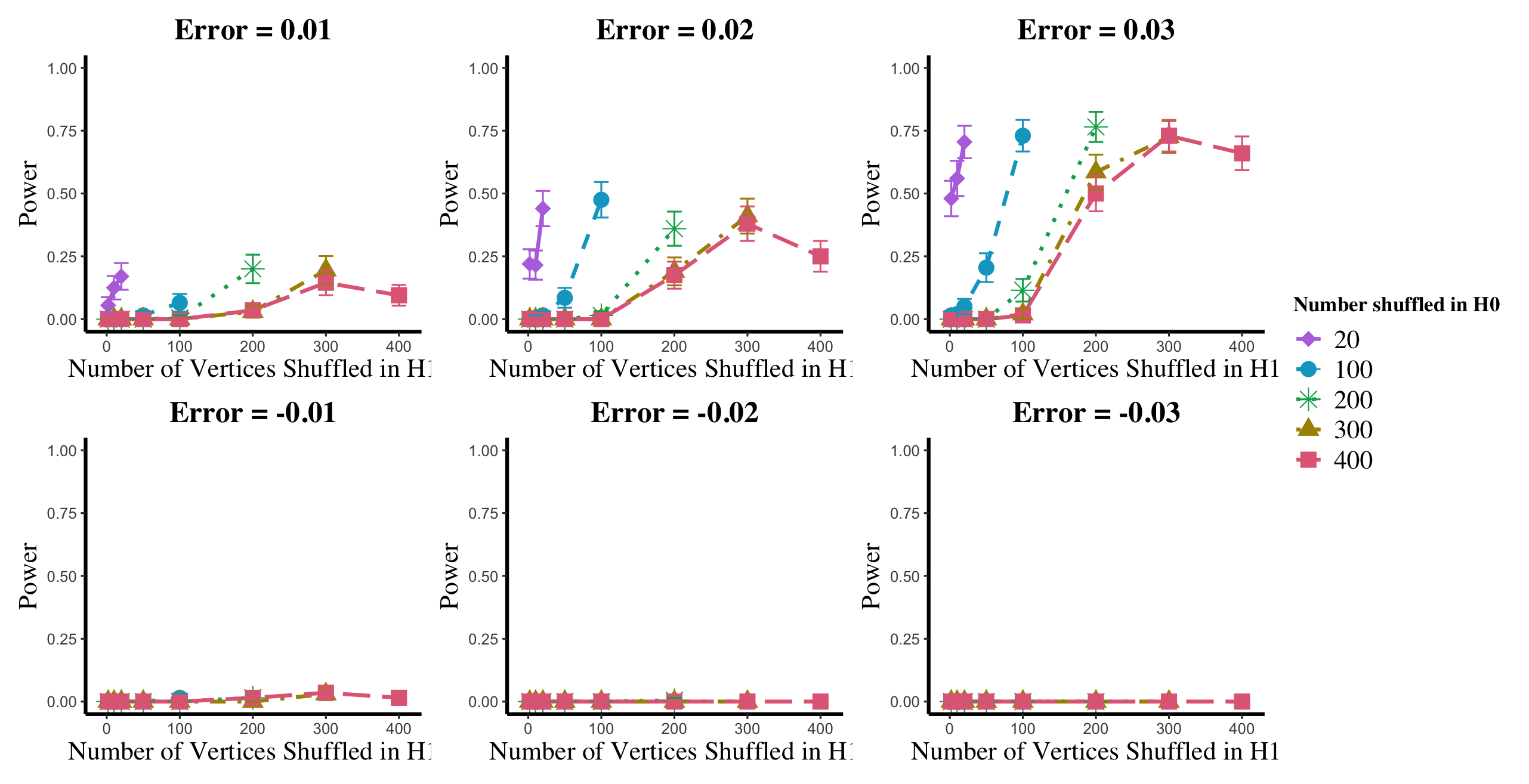}
\caption{
Power results at level $\alpha=0.05$ for $nMC=200$ Monte Carlo replicates of the adjacency matrix-based test statistic and null and alternative distributions presented in Eq. \ref{eq:AandE} (error bars are $\pm2$s.e.). 
In the figure the $x$-axis represents the number of vertices shuffled via $\textbf{Q}_\ell$ (from $0$ to $k$) while the curve colors represent the maximum number of vertices potentially shuffled via $\Pi_{k,n}$, here all shuffled by $\bQ_k$.}
\label{fig:bigoldanalysis1}
\end{figure*}
\subsubsection{Power loss in Presence of Shuffling: Adjacency versus $\widehat \bP$-based tests}
\label{sec:Phat_Adj_eps}
Note that in this section, the testing power was computed by directly sampling the distributions of the test statistic under the null and alternative.  In this 2-block stochastic block model setting, all shufflings permuting the same number of vertices between blocks are stochastically equivalent, and hence we can directly estimate the testing critical value under the null, and statistic under the alternative (for all $\ell\leq k$).
\begin{figure*}[t!]
\includegraphics[width=1\textwidth]{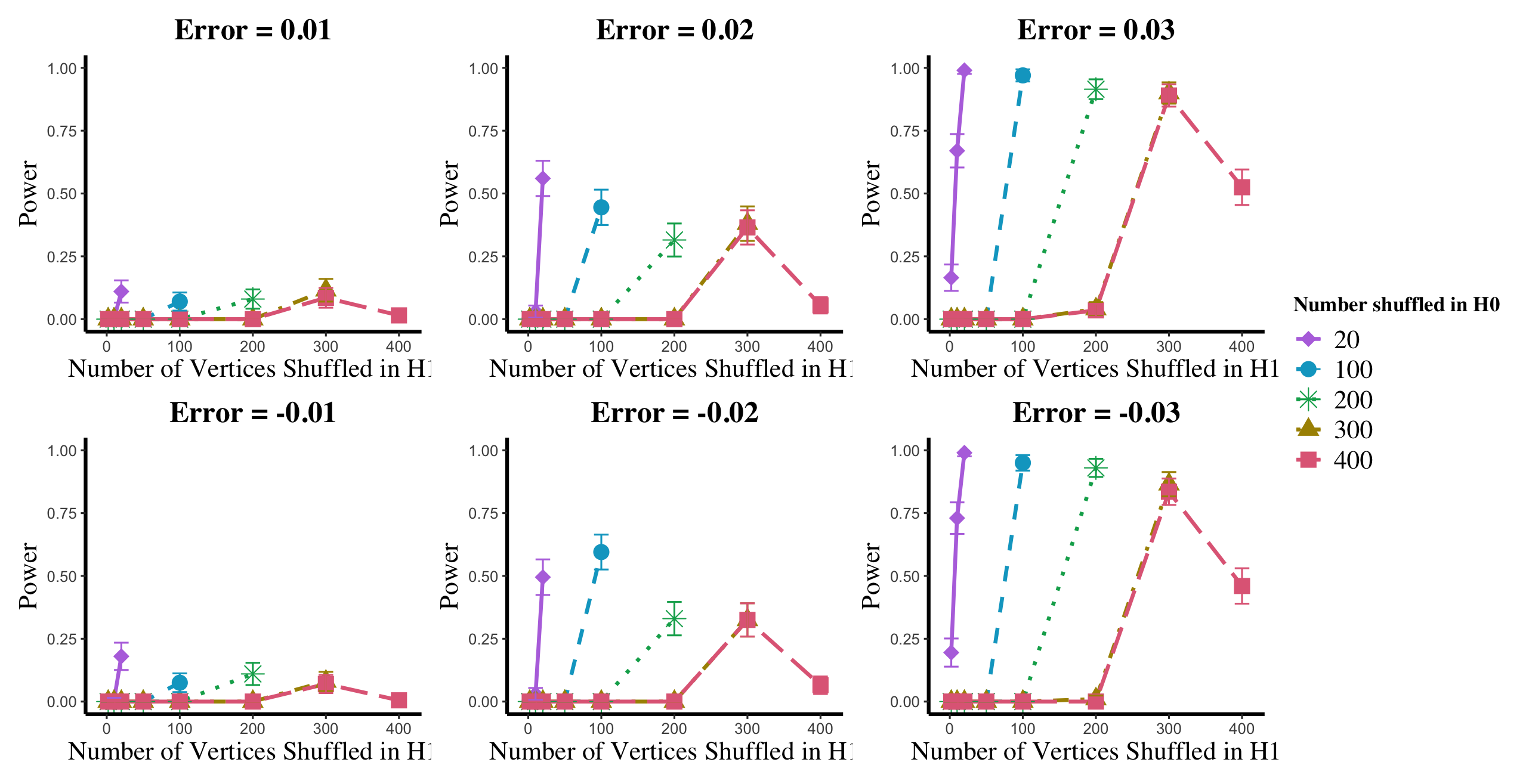}
\caption{
Power results at approximate level $\alpha=0.05$ for $nMC=200$ Monte Carlo replicates of the $\widehat{\bf P}$-based test statistic and null and alternative distributions presented in Eq. \ref{eq:AandE} (error bars are $\pm2$s.e.).
In the figure the $x$-axis represents the number of vertices shuffled via $\textbf{Q}_\ell$ (from $0$ to $k$) while the curve colors represent the maximum number of vertices potentially shuffled via $\Pi_{k,n}$, here all shuffled by $\bQ_k$.}
\label{fig:bigoldanalysis2}
\end{figure*}
Here the number of Monte Carlo replicates used to estimate the critical values is $nMC=200$, and $200$ further Monte Carlo replicates are used to estimate the power curves under each value of $\epsilon$ under consideration.
We note here that the true $d=2$ dimension was used to embed the graphs in this section.

To further explore the theoretical analysis of $T_{A}$ considered above, we consider the following simple, illustrative experimental setup.
With
$b(v)=2-\mathds{1}\{v\in\{1,2,\cdots,250\}$, we consider two $n=500$ vertex SBMs defined via
\begin{equation}
\label{eq:AandE}
\bA\sim \text{SBM}\left(2, \begin{bmatrix}
0.55 & 0.4\\
0.4 & 0.45
\end{bmatrix},b,1\right);\quad
\bB\sim \text{SBM}\left(2, \begin{bmatrix}
0.55 & 0.4\\
0.4 & 0.45
\end{bmatrix}+\textbf{E}_{\epsilon},b,1\right)
\end{equation}
where $\textbf{E}_{\epsilon} = \epsilon \times \textbf{J}_{500 \times 500}$, and
$\epsilon$ ranges over $\{\pm0.01,\pm0.02,\pm0.03\}$ (for an example of an analogous test---and result---in the sparse regime, see \hyperref[app:addexp]{Appendix A.4}).
According to Eq. \ref{eq:ep1} and \ref{eq:ep2}, we would expect the power of the adjacency matrix-based test to be poor for the $\epsilon<0$ values, even when $k-\ell$ is relatively small (i.e., even when the shuffling has a  negligible effect).
We see this play out in Figure \ref{fig:bigoldanalysis1}, 
where the adjacency matrix-based test (i.e., the test where $T(\bA,\bB)=\frac{1}{2}\|\bA-\bB\|_F^2$) demonstrates the following:
diminishing power in the $\epsilon>0$ setting when $k-\ell$ is large, and uniformly poor power in the $\epsilon<0$ setting.
Notably, when $k-\ell$ is small, the test power is (relatively) large when $\epsilon>0$ and is near 0 when $\epsilon<0$.
In Figure \ref{fig:bigoldanalysis2}, 
we see the above phenomena does not occur for the $\hbP$-based test (i.e., the test where $T(A,B)=\|\hbP_A-\hbP_B\|_F$), as for this test we see (nearly identical) diminishing power in the $\epsilon>0$ and $\epsilon<0$ settings when $k-\ell$ is large, and relatively high power when $k-\ell$ is small.
In both cases, the power is increasing as $\epsilon$ increases as expected.
We note here the odd behavior in Figures \ref{fig:bigoldanalysis1} and \ref{fig:bigoldanalysis2} when $k=400$ and $\ell=300,$ and $400$.
When $k=400$, most of the vertices are being shuffled, and the least favorable shuffling under the null does not shuffle the full $k=400$ vertices (indeed, the critical value when $\ell=300$ is larger here).   
This is because when $k=\ell=400$ here, the graphs are closer to the original 2-dimensional, 2-block SBM's with the blocks reversed.
In this case, the ``overshuffling'' of $\ell=400$ results in a smaller test statistic than the least favorable $k<400$ shuffling. 

One possible solution to the issue presented in Example \ref{ex:badA} (and exemplified in Eq. \ref{eq:ep1}--\ref{eq:ep2}) is to normalize the adjacency matrices to account for degree discrepancy.
With the setup the same as in Example \ref{ex:badA},
consider 
$$T(U,V)=\frac{\mathbb{E}|U-V|}{\mathbb{E}U(1-\mathbb{E}U)+\mathbb{E}V(1-\mathbb{E}V)},
$$
so that
\begin{align*}
T(X,Y)=1;\quad T(X,Z)=
    \frac{p(1-q)+q(1-p)}{p(1-p)+q(1-q)} \geq 1.
\end{align*}
With $\bA\sim$ER$(n,p)$ and $\bB\sim$ER$(n,q)$, rejecting $H_0:p=q$ for large values of $T$ will be asymptotically strongly consistent.  However, in heterogeneous ER settings (see Figure \ref{fig:bigoldanalysis3}) this degree normalization is less effective (especially when the expected degrees are equal across networks).
In Figure \ref{fig:bigoldanalysis3}, with 
$$b(v)=2-\mathds{1}\{v\in\{1,2,\cdots,250\},$$ we consider two $n=500$ vertex SBMs defined via
\begin{equation}
\label{eq:normtest}
\bA\sim \text{SBM}\left(2, \begin{bmatrix}
0.55 & 0.4\\
0.4 & 0.45
\end{bmatrix},b,1\right);\quad
\bB\sim \text{SBM}\left(2, \begin{bmatrix}
0.6& 0.35\\
0.35& 0.5
\end{bmatrix},b,1\right)
\end{equation}
and we consider testing $H_0:\mathcal{L}(A)=\mathcal{L}(B)$ in the presence of shuffling using the test statistic 
$$T_{\text{norm}}(\bA,\bB)=\frac{\frac{1}{2\binom{n}{2}}\|\bA-\bB\|_F^2}{\frac{1}{2\binom{n}{2}}\|\bA\|_F^2\left(1-\frac{1}{2\binom{n}{2}}\|\bA\|_F^2\right)+\frac{1}{2\binom{n}{2}}\|\bB\|_F^2\left(1-\frac{1}{2\binom{n}{2}}\|\bB\|_F^2\right)},$$
for the normalized adjacency matrix test (left panel), and the usual $T(\bA_1,\bA_2)=\|\widehat \bP_1-\widehat \bP_2\|_F$ for the $\widehat\bP$ test (right panel).
From the figure, we see that in settings such as this where $\|\bA\|_F^2\approx \|\bB\|_F^2$, the degree normalization is (unsurprisingly) unable to overcome the issues with the adjacency based test outlined in Example \ref{ex:badA}.

\begin{figure}[t!]
\includegraphics[width=1\textwidth]{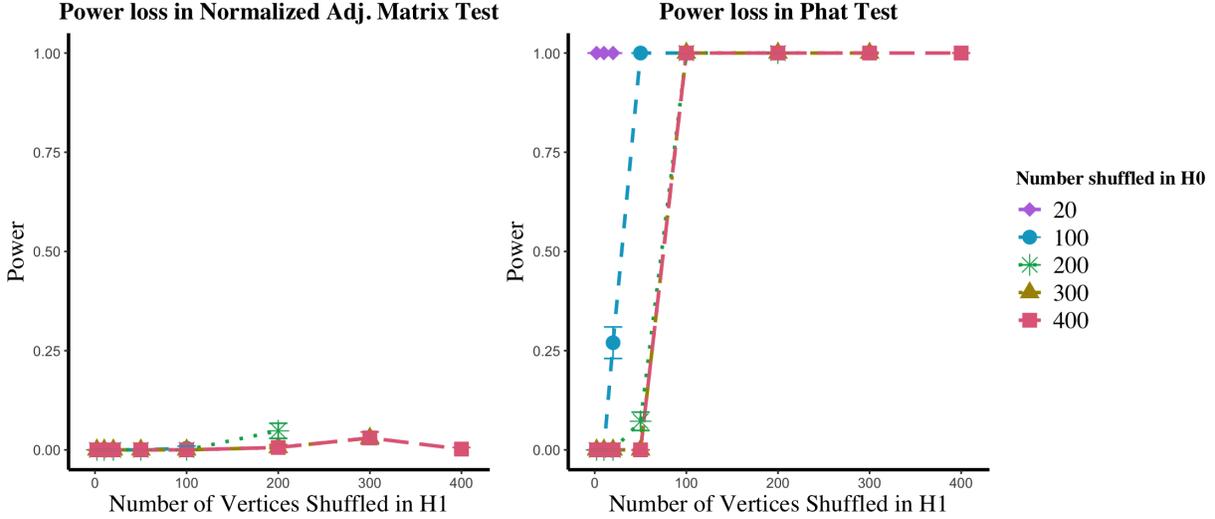}
\caption{
Power results at approximate level $\alpha=0.05$ for $nMC=500$ Monte Carlo replicates of the normalized test statistic and null and alternative distributions presented in Eq. \ref{eq:normtest} (error bars are $\pm2$s.e.). 
In the figure the $x$-axis represents the number of vertices shuffled via $\textbf{Q}_{2\ell}$ (from $0$ to $k$) while the curve colors represent the maximum number of vertices potentially shuffled via $\Pi_{k,n}$, here all shuffled by $\bQ_{2k}$.}
\label{fig:bigoldanalysis3}
\end{figure}

\section{Empirically exploring shuffling in ASE-based tests}
\label{sec:ASE}

\begin{figure}[t!]
\centering
\includegraphics[width=1\textwidth]{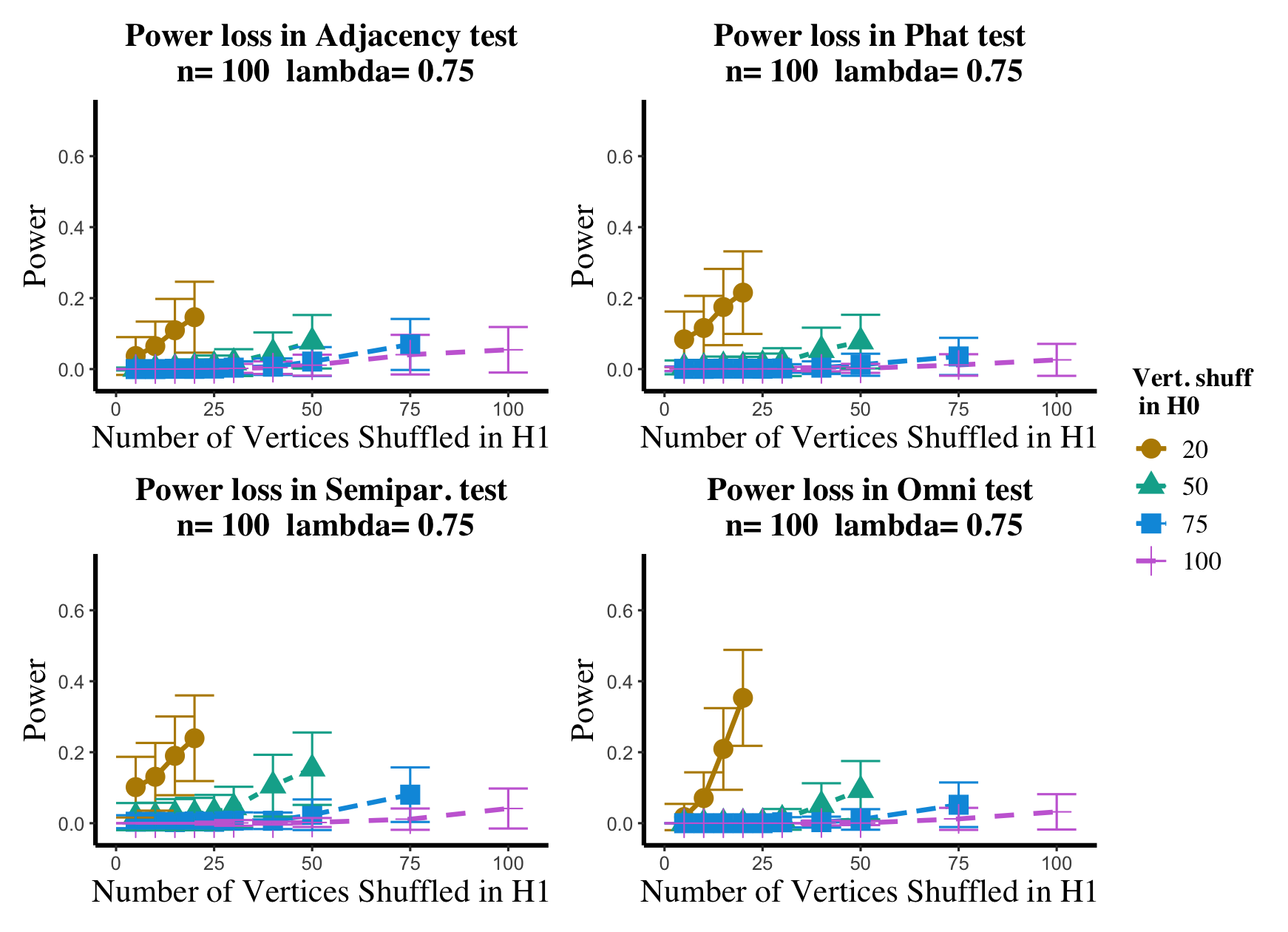}
\caption{For the experimental setup considered in Section \ref{sec:ASE}, we plot the empirical testing power in the presence of shuffling for the four tests: the Frobenius norm difference between the adjacency-matrices, between $\hbP$'s, $T_{\text{Omni}}$ and $T_{\text{Semipar}}$.
In the figure
the x-axis represents the number of vertices actually shuffled in $U_{n,k}$ (i.e., the number shuffled in the alternative) while the curve colors
represent the maximum number of vertices potentially shuffled via in $U_{n,k}$.}
\label{fig:levinvsomni075}
\end{figure}
\begin{figure}[t!]
\centering
\includegraphics[width=1\textwidth]{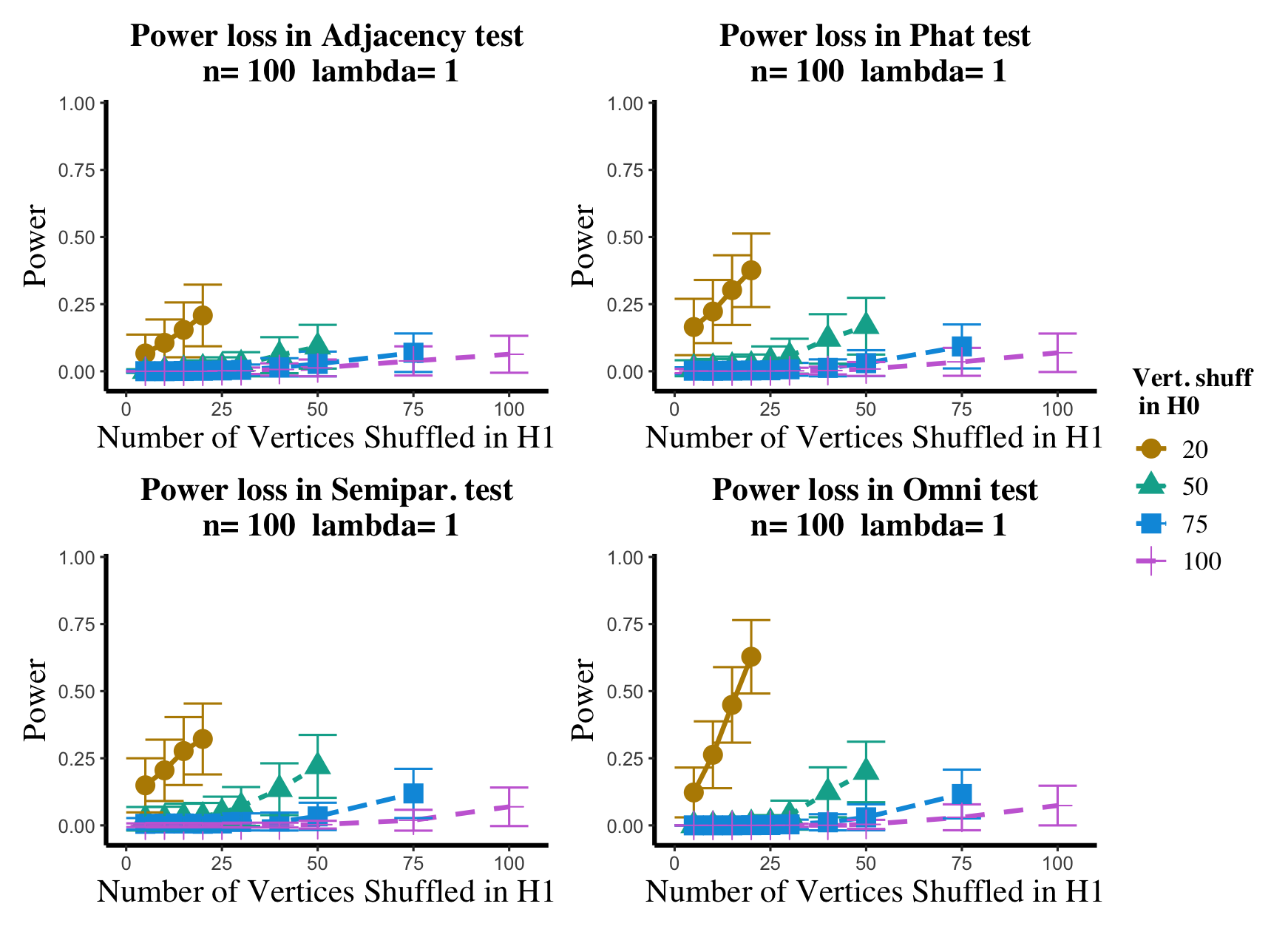}
\caption{For the experimental setup considered in Section \ref{sec:ASE}, we plot the empirical testing power in the presence of shuffling for the four tests: the Frobenius norm difference between the adjacency-matrices, between $\hbP$'s, $T_{\text{Omni}}$ and $T_{\text{Semipar}}$.
In the figure
the x-axis represents the number of vertices actually shuffled in $U_{n,k}$ (i.e., the number shuffled in the alternative) while the curve colors
represent the maximum number of vertices potentially shuffled via in $U_{n,k}$.}
\label{fig:levinvsomni01}
\end{figure}
As mentioned in previous sections, multiple spectral-based hypothesis testing regimes have been proposed in the literature over the past several years; see, for example, \cite{levin2017central,tang2017semiparametric,Tang2017,asta2014geometric}.
One of the chief advantages of the $\hbP$-based test considered herein is the ease in which the analysis lends itself to understanding the effect of shuffling; indeed, this power analysis is markedly more complex for the tests considered in  \cite{levin2017central,tang2017semiparametric}, for example.

In the ASE-based tests in \cite{levin2017central,tang2017semiparametric}, the authors
consider $n$ vertex, $d$ dimensional RDPG's
$\bA \sim \text{RDPG}(\bX, \nu=1)$
and 
$\bB \sim \text{RDPG}(\bY, \nu=1)$, and seek to test 
\begin{equation}
    \label{eq:hypoth}
H_0:\textbf{X}\stackrel{\perp}{=}\textbf{Y},\text{ versus }H_1:\textbf{X}\not\stackrel{\perp}{=}\textbf{Y},
\end{equation}
where $\textbf{X}\stackrel{\perp}{=}\textbf{Y}$ holds if there exists an orthogonal matrix $\mathbf{W}\in\mathcal{O}_d$ such that $\textbf{X}=\textbf{Y}\mathbf{W}$.
This rotation is to account for the inherent non-identifiability of the RDPG model, as latent positions $\textbf{X}$ and $\textbf{Y}$ satisfying $\textbf{X}\stackrel{\perp}{=}\textbf{Y}$ yield the same graph distribution.
The semiparametric test of \cite{tang2017semiparametric} used as its test statistic a suitably scaled version of the Frobenius norm between suitably rotated ASE estimates of $\bX$ and ${\bf Y}$; namely, an appropriately scaled version of 
$T_{\text{Semipar}}(\bA,\bB) = \min_{\textbf{W} \in \mathcal{O}_d} \norm{\widehat{\textbf{X}} \textbf{W} - \widehat{\textbf{Y}} }_F,$
where $\widehat{\textbf{X}}=$ASE($\bA,d$) and $\widehat{\textbf{Y}}=$ASE($\bB,d$).
While consistency of the test based on $T_{\text{Semipar}}$ is shown in \cite{tang2017semiparametric}, the effect of shuffling vertex labels in the presence of the Procrustean alignment step is difficult to parse here, and is the subject of current research.

The separate graph embeddings cannot be compared in $T_{\text{Semipar}}$ above without first being aligned (e.g., without the Procrustes alignment provided by $\bW$) due to the non-identifiability of the RDPG model, and this added variability (uncertainty) motivated the Omnibus joint embedding regime of \cite{levin2017central}.
The Omnibus matrix in the $m=2$ setting--- $m$ here the number of graphs to embed---is defined as follows.
Given two adjacency matrices $\textbf{A}, \textbf{B} \in \mathbb{R}^{n \times n}$ on the same vertex set with known vertex correspondence, the omnibus matrix $\textbf{M} \in \mathbb{R}^{2n \times 2n}$ is defined as
$${\bf M}=\begin{bmatrix}\bA& \frac{\bA+\bB}{2}\\
\frac{\bA+\bB}{2}& \bB
\end{bmatrix}.
$$
Note that this definition can easily extend to a sequence of matrices  $\textbf{A}_1, \dots, \textbf{A}_m$, where the $i,j$-th block of the omnibus matrix 
$\textbf{M}_{ij}= (\textbf{A}_i+\textbf{A}_j)/2$ for all $i,j \in [m]$.
We present the case when $m=2$ for simplicity.
When combined with ASE, the Omni framework allows for us to simultaneously produce directly comparable estimates of the latent positions of each network without the need for a rotation.
Let the ASE of $\mathbf{M}$ be defined as
$\text{ASE}(\mathbf{M},d)=\widetilde{\mathbf{Z}}=[\widetilde\bX^T,\widetilde{\bf Y}^T]^T,$
where $\widetilde{\mathbf{Z}}\in\mathbb{R}^{2n\times d}$ provides, via its first $n$ rows denoted $\widetilde\bX$, an estimate of $\bX$ and, via its second $n$ rows denoted $\widetilde{\bf Y}$, an estimate of ${\bf Y}$.
The Omnibus test, as proposed in \cite{levin2017central}, seeks to test the hypotheses in Eq. \ref{eq:hypoth} via the test statistic
$T_{\text{Omni}} =  \norm{\widetilde{\textbf{X}} - \widetilde{\textbf{Y}} }_F.$
Concentration and asymptotic normality of $T_{\text{Omni}}$ under $H_0$ is established in \cite{levin2017central} (see also the work analyzing $T_{\text{Omni}}$ under the alternative in \cite{Draves2020}), and in \cite{levin2017central} the Omni-based test demonstrates superior empirical testing performance compared to the test in \cite{tang2017semiparametric}.
As in the case with $T_{\text{Semipar}}$, the effect of shuffling vertex labels in the presence of the omnibus structural/construction alignment is difficult to theoretically understand, and is the subject of current research.

In the setting above, we will compare the performance of the $\hbP$ and adjacency-based tests with $T_{\text{Omni}}$ and $T_{\text{Semipar}}$ in a
(slightly) modified version of the experimental setup of \cite{levin2017central}.
To wit, we consider paired $100$-vertex RDPG graphs where
the rows of $\bX$ are i.i.d.\@ Dirichlet($\alpha=(1,1,1)$) random vectors, with the exception that the first five rows of $\bX$ are fixed to be $(0.8, 0.1, 0.1)$ (in \cite{levin2017central} they consider one fixed row).
The rows of $\bY$ are identical to those of $\bX$, with the exception that the first five rows of $\bY$ are fixed to be 
$(1-\lambda)(0.8, 0.1, 0.1)+\lambda (0.1, 0.1, 0.8)$; here we consider $\lambda$ ranging over $(0,0.25,0.5,0.75,1)$ (plots for the $\lambda$ and $k$ values not shown here can be found in Appendix \ref{app:addexp}).
In the language of Section \ref{sec:theory}, we are setting here $r=5$, with the $\lambda$ controlling the level of error $e_{i,j}$ in each entry of $\bE$.
This change in notation (from $\bE$ to $\lambda$) is to maintain consistency with the notation used in the motivating work of \cite{levin2017central,levin2017central2}.
Note that in each ASE computed in this experiment, the true $d=3$ dimension was used for the embedding.

As in the connectomic real-data example,
incorporating the unknown shuffling of $U_{n,k}$ into the adjacency and $\hbP$-based tests and into $T_{\text{Omni}}$ and $T_{\text{Semipar}}$ is tricky here, as for moderate $k$ it is computationally infeasible to compute the conservative critical values exactly. 
Our compromise is that in the Monte Carlo simulations below, we sample random permutations that fix no element of $U_{n,k}$ to act as the elements generating the least favorable null; while this will not guarantee the worst case shuffling is sampled (and so the test will most-likely not achieve its desired level of $\alpha=0.05$), this seems reasonable in light of the non-fixed rows of $\bX$ being i.i.d.

In order to simulate testing $
H_0:\textbf{X}\stackrel{\perp}{=}\textbf{Y}$, we consider the following two-tiered Monte Carlo simulation approach.
\begin{itemize}
\item[1.] For $i=1,\ldots,nMC_1=50$
\begin{itemize}
   \item[i.] Simulate $\bX$ and $\bY$, drawn as above
   \item[ii.] For $j=1,\ldots,nMC_2=100$
   \begin{itemize}
   \item[a.] Generate a nested sequence of random derangements of the elements of $U_{n,k}$ (ranging over $k$);
   \item[b.] Simulate $\bA_1,\bA_2\stackrel{\text{i.i.d}}{\sim} \text{RDPG}(\bX,\nu=1)$ and independently simulate $\bA_3\sim \text{RDPG}(\bY,\nu=1)$;
   \item[c.] Compute the test statistic under the null and alternative for each $k$, namely compute $T^{(k,j)}=\|\hat P_1-\hat P_{2,k}\|_F$ for the null, and $T^{(\ell,j)}=\|\hat P_1-\hat P_{3,\ell}\|_F$ for the alternative;
\end{itemize}
\item[iii.] Estimate the critical value for the test for each $k$ using $\{T^{(k,j)}=\|\hat P_1-\hat P_{2,k}\|_F\}_{j=1}^{100}$
\item[iv.] Estimate the power for the test at each $k$ by computing the proportion of $\{T^{(\ell,j)}=\|\hat P_1-\hat P_{3,\ell}\|_F\}_{j=1}^{100}$ greater than the critical value in step iii.
\end{itemize}  
\item[2.] Compute the power average over the $nMC_1=50$ outer Monte Carlo iterates.
\end{itemize} 
% each of $nMC_1=50$ Monte Carlo replicates of $\bX$ and $\bY$ drawn as above, we consider simulating $nMC_2=100$ pairs of random graphs under the null shuffled according to a random derangement of $U_{n,k}$ (ranging over $k$) to estimate the test's critical value in the presence of shuffling; we further simulate $nMC_2=100$ pairs of random graphs under the alternative shuffled according to random derangements of $\ell$ elements of $U_{n,k}$ (ranging over $\ell\leq k$) to then estimate the testing power.
% We finally present these estimated powers averaged over the outer $nMC_1=50$ Monte Carlo replicates.

\noindent Results are displayed in Figure \ref{fig:levinvsomni075}-- \ref{fig:levinvsomni01} where we range $k=( 20, 50, 75, 100)$ (in the figures in Appendix \ref{app:addexp}, we range $k=(10, 20, 30, 50, 75, 100)$) and different $k$ values are represented by different colors/shapes of the plotted lines. 
Here $\ell\leq k$, the number shuffled in the alternative (i.e., the number of actual incorrect labels in $U_{n,k}$), ranges over the values plotted on the $x$-axis.
Figure \ref{fig:levinvsomni075} (resp., Figure \ref{fig:levinvsomni01}) show results for $\lambda=0.75$ (resp., $\lambda=1$); plots for the remaining values of $\lambda$ can be found in Appendix \ref{app:addexp}.
From the figures, we see that, as expected, larger values of $\lambda$ (i.e., more signal in the alternative) yield higher testing power, and that power diminishes greatly as the number of vertices shuffled in the null is increasing relative to the number shuffled in the alternative.
We also note that the Omnibus based test appears to be more robust to shuffling than the other tests; developing analogous theory to Theorems \ref{thm:power_to_1}--\ref{thm:power_to_0} for the Omnibus test is a natural next step, though we do pursue this further here.
We lastly note that the loss in power in the large $k$ settings is more pronounced here than in the SBM simulations, even when $\ell\approx k$; we suspect here this is due to the noise introduced by the large amount of shuffling being of higher order than the signal in the alternative.

\begin{figure}[t!]
\centering 
\includegraphics[width=1\textwidth]{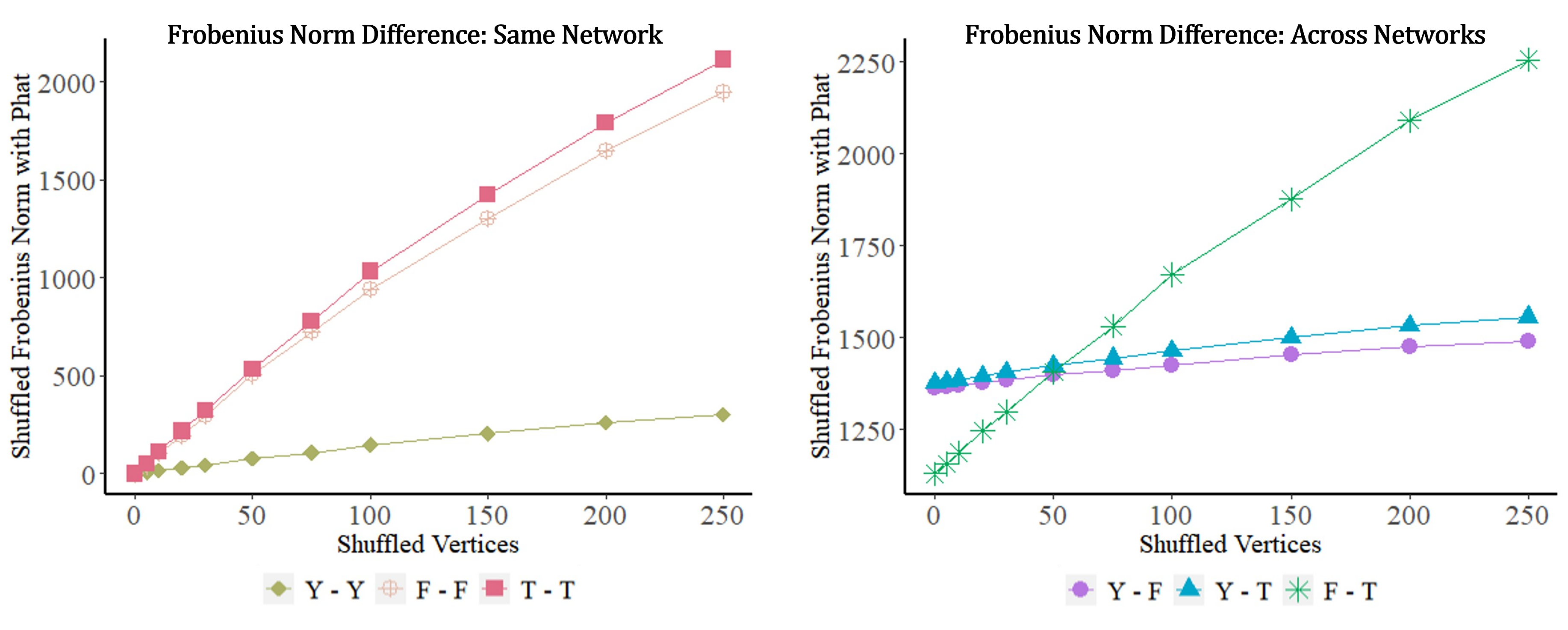}
\caption{In the left (resp., right) panel, we plot the Frobenius norm difference of the estimated $\hbP$ between the same (resp., across) network as the number of shuffled vertex labels is increased within (resp., across) networks; the $x$-axis represents the number of vertices shuffled.
Note the different scales on the $y$-axes of the two panels in the figure.}
\label{fig:sn_frob}
\end{figure}

%%%%%%%%%%%%%%%%%%%%%%%%%%%%%%%%%%%%%%%

\section{Shuffling in social networks}
\label{sec:SN}
In this section, we explore the shuffled testing phenomena in the context of the social media data found in \cite{magnani2011ml}.
The data contains a multilayer social network consisting of user activity across three distinct social networks: YouTube, Twitter, and FriendFeed; after initially cleaning the data (removing isolates, and symmetrizing the networks, etc.), there are a total of 422 common users across the three networks. 
 Given adjacency matrices of our three 422-user social networks $\bA=\bA_\bY$ (Youtube), $\bB=\bB_{\bf T}$ (Twitter), and $\bC=\bC_{\bf F}$ (FriendFeed), we ultimately wish to understand the effect of vertex misalignment on testing the following hypotheses
  \begin{align*}
   H_0^{(1)}&\!\!:\!\mathcal{L}(\bA)\!=\!\mathcal{L}(\bB);\,\,
   H_0^{(2)}\!\!:\!\mathcal{L}(\bA)\!=\!\mathcal{L}(\bC);\,\,
   H_0^{(3)}\!\!:\!\mathcal{L}(\bB)\!=\!\mathcal{L}(\bC)\\
   H_1^{(1)}&\!\!:\!\mathcal{L}(\bA)\!\neq\!\mathcal{L}(\bB);\,\,
   H_1^{(2)}\!\!:\!\mathcal{L}(\bA)\!\neq\!\mathcal{L}(\bC);\,\,
   H_1^{(3)}\!\!:\!\mathcal{L}(\bB)\!\neq\!\mathcal{L}(\bC)
  \end{align*}
In this case, it is difficult to get a handle on the critical values across these three tests under shuffling, so we consider the following simple initial illustrative experiment to shed light on these tests.
Assuming an underlying RDPG model for each of the three networks, and using Friendfeed as our fulcrum (note that similar results are obtained using Twitter as the fulcrum network), we consider testing the simple parametric hypotheses under the effect of shuffling
  \begin{align}
   H_0^{(1)}&:\bA\sim\text{RDPG}(\bX_\bC)\quad\quad
   H_0^{(2)}:\bB\sim\text{RDPG}(\bX_\bC)\notag\\
   H_1^{(1)}&:\bA\not\sim\text{RDPG}(\bX_\bC)\quad\quad
   H_1^{(2)}:\bB\not\sim\text{RDPG}(\bX_\bC) \label{eq:h0abc}
   \end{align}
  Letting $\bP_\bullet$ represents the edge probability matrix for the corresponding social network, we compute
  $T(\bA,\bC)=\frac{1}{2}\|\hbP_\bA-\hbP_\bC\|_F^2$ (similarly $T(\bB,\bC)$ and $T(\bA,\bB)$),
where the embedding dimensions needed to compute $\hbP_\bullet=\hbX_\bullet\hbX^T_\bullet$ ($\hbX_\bullet$ being the ASE of the associated network) are each chosen via an automated elbow finder on the scree plot of $\textbf{A}, \textbf{B}, \textbf{C}$ inspired by \cite{zhu2006automatic} and \cite{chatterjee2015matrix}, and where we then set a common embedding dimension for the three networks to the max of these three estimated dimensions. 

We plot these initial findings in 
Figure \ref{fig:sn_frob}.
In the left (resp., right) panel of the figure, we plot the Frobenius norm difference of the estimated $\hbP$ between the same (resp., across) network as the number of shuffled vertex labels is increased within (resp., across) networks; the $x$-axis represents the number of vertices shuffled.
Note the different scales on the $y$-axes of the two panels in the figure.
From the figure, we see that although all network pairs differ significantly from each other, the FriendFeed and Twitter networks are more similar to each other (according to $T$) than either is to the Youtube network.
However, this is obscured given enough vertex shuffling as seen by the green curve crossing the blue/purple curves in the right panel.
Given enough uncertainty in the vertex labels,
we posit that a conservative test using $T$---i.e., one that must assume the uncertain labels are shuffled---would compute the FriendFeed and Twitter networks to be more similar to each other (if the uncertain labels are, in fact, correct and not shuffled under $H_1$) than either is to themselves, and so we should be less likely to reject $H_0^{(2)}$ of Eq. \ref{eq:h0abc}.  

\begin{figure}[t!]
\begin{flushleft}
\includegraphics[width=1\textwidth]{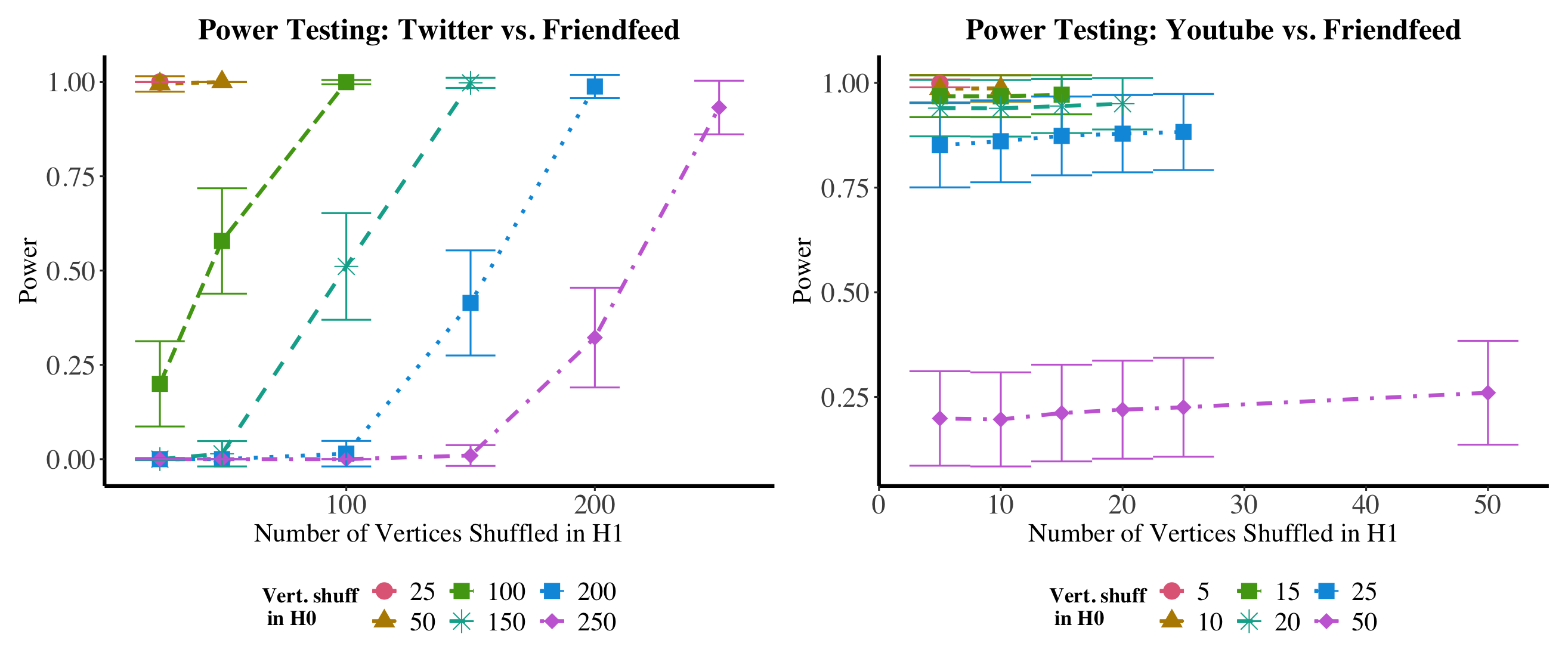}
\caption{Using a parametric bootstrap with 200 bootstrapped replicates to estimate testing power for $H_0^{(1)}$ and $H_0^{(2)}$.
The amount shuffled in the null is $k$; the different colored curves represent the different $k$-values. 
The amount shuffled in the alternative, $\ell\leq k$, is plotted on the $x$-axis. Here, the results are averaged over $nMC=50$ Monte Carlo replicates ($\pm2$s.e.).}
\label{fig:SN_boot}
\end{flushleft}
\end{figure}

We see this play out in Figure \ref{fig:SN_boot}, where we use a parametric bootstrap (assuming the RDPG model framework) with 200 bootstrapped replicates to estimate each of the critical values and the testing power in order to see the effect of vertex shuffling on the testing power for $H_0^{(1)}$ and $H_0^{(2)}$ of Eq. \ref{eq:h0abc}; note that as in section \ref{sec:motive} 
we consider the following
modification of the overall testing regime:  We consider a fixed (randomly chosen) sequence of nested permutations $\bQ_k\in\Pi_{n,k}$ for $k=25,\, 50,\, 100,\, 150,\, 200,\,250$ for $H^{(2)}$ and $k=5, 10, 15, 20, 25,\, 50$ for $H^{(1)}$.
We consider shuffling the null by $\bQ_k$ and the alternative by $\bQ_{\ell}$ for all $\ell\leq k$.
We then repeat this process $nMC=50$ times, each time obtaining a bootstrapped estimate  of testing power against $H_0$.
In the figure, we plot the average empirical testing power, where the amount shuffled in the conservative null is $k$ (so that $U_{n,k}$ is entirely shuffled); the different colored curves represent the different $k$-values. 
The amount actually shuffled in the alternative, $\ell\leq k$, is plotted on the $x$-axis.
From the figure, we see that the test would reject the null in both cases when few vertices are shuffled (i.e., a small $k$-value) or when $k-\ell$ is small in the Twitter versus Friendfeed panel; this is as expected from Figure \ref{fig:sn_frob}, as the networks all seem to differ significantly from each other.
However, testing power degrades significantly when $k-\ell$ is large in the Twitter versus Friendfeed setting and in the large $k$ setting for Youtube versus Friendfeed, and the test no longer rejects the null.
While the former is as expected---indeed, with enough uncertainty the small differences across the Twitter and Youtube networks is lost in the shuffle, even when the networks are different---the latter is surprising.
Further investigation yields that this power degradation as a function of $k$ alone stems from a difference in density, as the Youtube network is sparse and $\|\hbP_C-\hbP_A\textbf{}\|_F\approx\|\hbP_C\|_F$; in this case the error in shuffling under the null when $k$ is large overwhelms any effect of shuffling in the alternative.

If the sparse Youtube network is used as the fulcrum, then both power figures would have power equal to one uniformly.  The sparse Youtube graph is markedly different in degree from the two denser networks, and the shuffling does not affect that difference.  
In general, testing for a measure of equality in graphs with markedly different degree distributions would require a modified version of the hypotheses and test statistic, perhaps of the form (inspired by \cite{tang2017semiparametric}) $\|\widehat{\mathbf{P}}_1-c\widehat{\mathbf{P}}_2\|_F$
for a scaling constant $c$ (to test equality of latent positions up to rotation and constant scaling) or $\|\widehat{\mathbf{P}}_1-\mathbf{D}\widehat{\mathbf{P}}_2\mathbf{D}^T\|_F$ for a diagonal scaling matrix $\mathbf{D}$ (to test equality of latent positions up to rotation and diagonal scaling).

%---------- Section----------------
\section{Graph matching (unshuffling)}
\label{sec:GM}

Once we quantify the added uncertainty due to shuffles in the vertex correspondences, it is natural to try to remedy this shuffling via graph matching. 
In this context, the graph matching problem seeks the correspondence between the vertices of the graphs that minimizes the number of induced edge disagreements \cite{short_survey_GM}. 
Formally, letting $\Pi_n$ again denote the sets of $n \times n$ permutations matrices, the simplest formulation of the graph matching problem can be cast as seeking elements in $\argmin_{\textbf{Q} \in \Pi_n} \norm{\textbf{A} - \textbf{QBQ}^T}_F^2$.  
For a survey on the current state of the graph matching literature, see \cite{30yrs,10yrs,short_survey_GM}.

The problem of graph matching can often be made easier by having prior information about the true vertex correspondence.
This information can come in the form of seeds, or a list of vertices for which the true, latent correspondence is known a priori; see \cite{fang2018tractable,fishkind2019seeded,mossel2020seeded}.
In the current shuffled testing regime, only the vertices of $U_{n,k}$ have unknown correspondence, and so the graph matching problem would reduce to
seeking elements in $\argmin_{\textbf{Q} \in \Pi_{n,k}} \norm{\textbf{A} - \textbf{QBQ}^T}_F^2$; i.e., those vertices in $M_{n,k}$ can be treated as seeds.
While there are no efficient algorithms known for matching in general (with or without seeds), 
there are a number of approximate seeded graph matching algorithms (see, for example, \cite{fishkind2019seeded,mossel2020seeded})
that have proven effective in practice.
In applications below, we will make use of the \texttt{SGM} algorithm of \cite{fishkind2019seeded} to approximately solve the seeded graph matching problem.

\subsection{SBM Shuffling and Matching}
\label{sec:SGM_SBM}

\begin{figure}[t!]
\begin{center}
\includegraphics[width=1\textwidth]{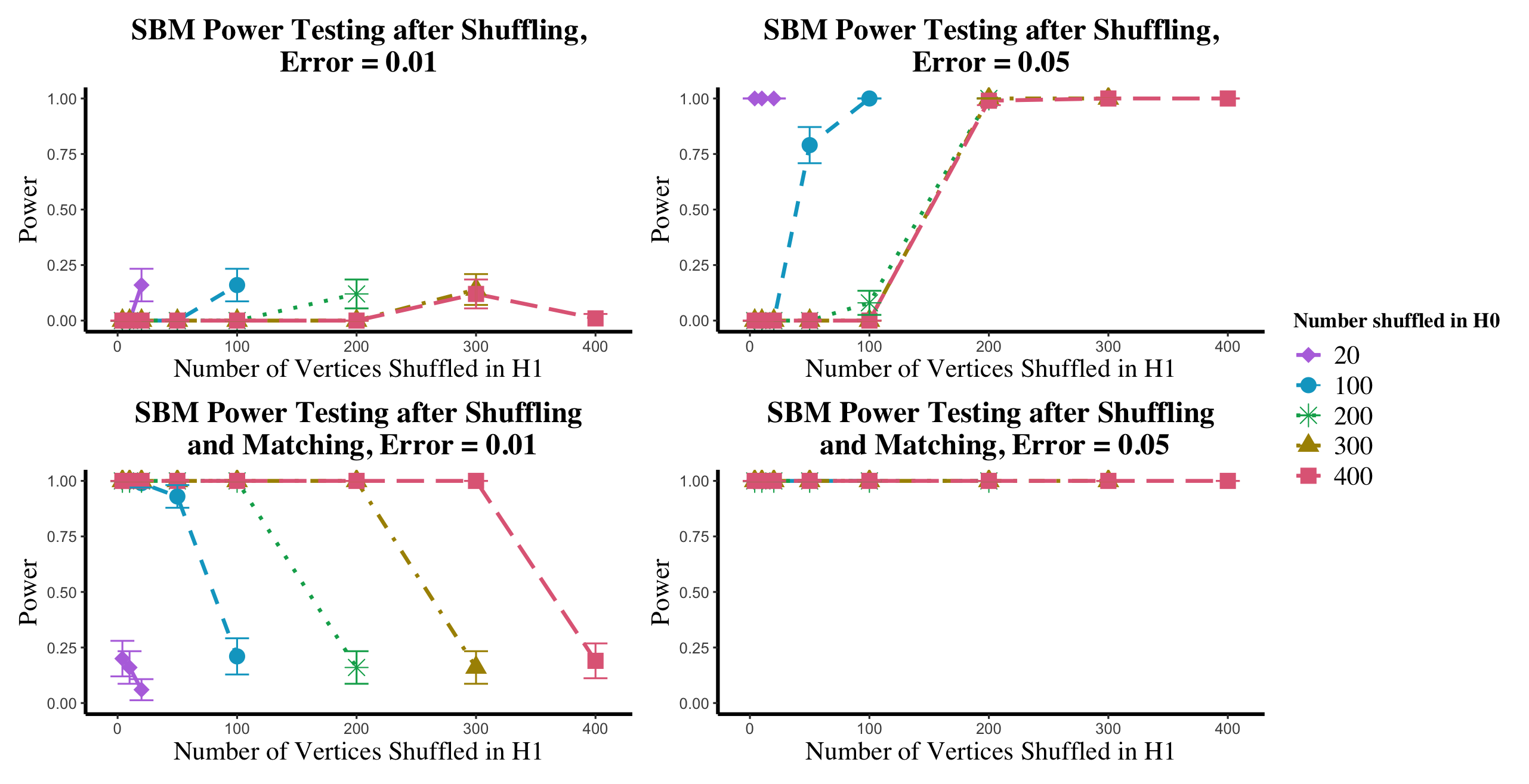}
\caption{(Top row) Shuffling different number of vertices across two blocks in an SBM 
model from Section \ref{sec:SGM_SBM} with $\epsilon = 0.01$, $0.05$
% using Equation \eqref{powershuffunshuff} 
to calculate power for $n=500$ vertices.
(Bottom row) The power of the test after using \texttt{SGM} to match the $k$ vertices in $U_{n,k}$ under the null, versus the $\ell$ vertices truly shuffled in the alternative.
Results are averaged over 100 Monte Carlo replicates with error bars $\pm2$s.e.}
\label{fig:SBMShuffMatch-0.05}
\end{center}
\end{figure}
We first begin by testing the dual effects of shuffling and matching on testing power in the SBM setting (where, as in Section \ref{sec:Phat_Adj_eps}, the critical values and the power can be computed via direct simulation).
Adopting the notation of Section \ref{sec:SBM_test},
we first consider an SBM with $n_1=n_2=250$ vertices, where $\Lambda$ is given by:
$${\Lambda}= 
\begin{pmatrix}
0.55 & 0.4\\
0.4 & 0.45
\end{pmatrix}\in\mathbb{R}^{250\times 250}$$
Letting $\textbf{Q}_h$ ($h\leq k$, so that $h/2$ vertices are shuffled between each of the two blocks) be as defined in Section \ref{sec:SBM_test}, we consider the error matrix $\textbf{E}=\epsilon \mathbf{J}_n$; in this example we will consider $\epsilon \in \{0.01, 0.05\}$. 
In this setting, we will simulate directly (using the true model parameters) from the SBM models to estimate the relevant critical values and the testing powers; here, all critical values and power estimates are based on $nMC=100$ Monte Carlo replicates.

In Figures \ref{fig:SBMShuffMatch-0.05}, we plot (in the upper panels) the power loss due to shuffling in the $\epsilon=0.01$ and $0.05$ settings; in the bottom panels, we plot the effect of shuffling and then matching on testing power (where we first match the graphs, and then use the unshuffled statistic $\|\hbP_1-\hbP_2\|_F$ for the hypothesis test.  For each $k$, we would (in practice) use \texttt{SGM} to unshuffle \textit{all} $k$ vertices in $U_{n,k}$ regardless of the value of $\ell$---represented here by the unshuffling in the $\ell=k$ case.
Here, we also show the effect of unshuffling only the $\ell$ vertices truly shuffled in the alternative.
From the $k=\ell$ case, we see that matching recovers the lost shuffling power.  From the $\ell<k$ setting, we see that the (possible) downside of unshuffling is its propensity to align the two graphs better than the true, latent alignment.  In general, the more vertices being matched, the smaller the graph matching objective function, and hence the smaller the test statistic, which here manifests as larger than desired type-I error probability.
Note that in each ASE used to compute $\widehat\bP$ in this experiment, we used the true value of $d=2$.

\begin{figure}[t!]
\begin{center}
\includegraphics[width=1\textwidth]{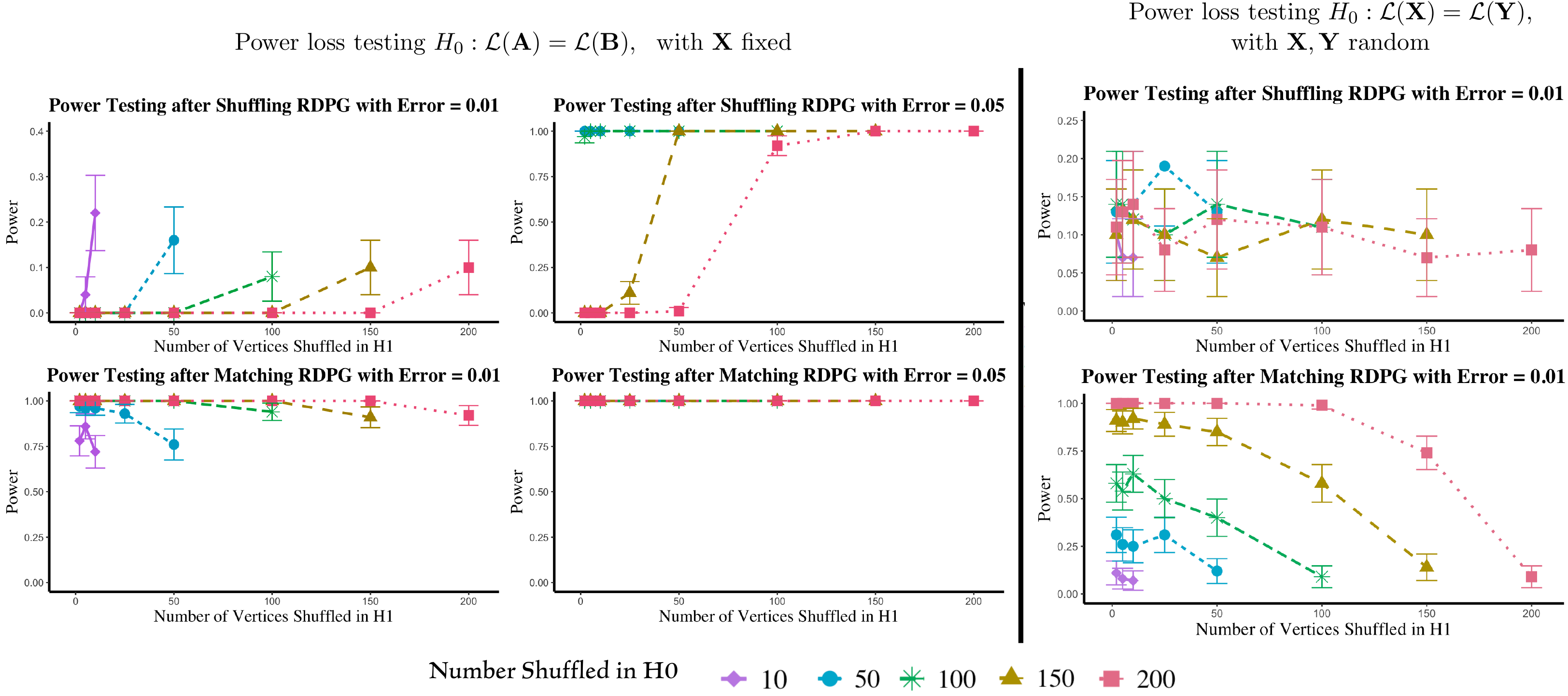}
\caption{(Top row) Shuffling different number of vertices across two blocks in an RDPG 
model from Section \ref{sec:ASE} with $\epsilon = 0.01$, and $0.05$
% using Equation \eqref{powershuffunshuff} 
to calculate power for $n=500$ vertices.
(Bottom row) The power of the test after using \texttt{SGM} to match the $k$ vertices in $U_{n,k}$ under the null, versus the $\ell$ vertices truly shuffled in the alternative.
Error bars are $\pm2$s.e.}
\label{fig:RDPGshuff_unshuff}
\end{center}
\end{figure}
These results play out in more general random graph settings as well, with Figure 
\ref{fig:RDPGshuff_unshuff}, showing the corresponding results in the RDPG setting of Section \ref{sec:ASE}; i.e., each graph is $n=500$ vertices with the rows of $\bX$ are i.i.d. Dirichlet(1,1,1).  The error here is added directly to the latent positions, as $\widetilde \bX=\bX+\epsilon \mathbf{J}_{n,3}$ for $\epsilon=0.01$, and $0.05$.
In the left panels, where we are fixing the latent positions and testing $H_0:\mathcal{L}(A)=\mathcal{L}(B)$, the results are similar: unshuffling recovering the testing power, and over-matching providing artificially high testing power.
In the right panel, we show results for testing $H_0:\mathcal{L}(\bX)=\mathcal{L}(\bY)$ where the latent positions are random and the rows of $\bX$ are i.i.d.  In this case, shuffling has no effect on testing power (as seen in the upper panel), while over-matching can still detrimentally add artificial power (as seen in the lower panel).

%%%%%%%%%%%%%%%%%%%%%%%%%%%%%%%%%%%%%%%%%%%%%%%%%%%%
%%%%%%%%%%%%%%%%%%%%%%%%%%%%%%%%%%%%%%%%%%%%%%%%%%%%

\subsection{Brain networks shuffling and matching}
\label{sec:brain_shuff_unshuff}

We explore the dual effects of shuffling and matching on testing power in the motivating brain testing example of Section \ref{sec:motive} (again with 200 bootstrap samples and averaged over nMC=100 Monte Carlo trials).
In Figures \ref{fig:unshuffbrains}, we plot (in the left panel) the power of the test when the brains are shuffled, then matched before testing (the power here is estimated using the parametric bootstrap procedure outlined in Section \ref{sec:motive}); note, as we first match the graphs we use the unshuffled statistic for this hypothesis test, 
i.e., we use $T(\bA_1, \bA_2)=\|\widehat \bP_1 -\widehat \bP_2\|_F$ when estimating the null critical value, and $T(\bA_1, \bA_3) = \|\widehat \bP_1 -\widehat \bP_3\|_F$ when estimating the testing power.  
Again, we show the effect of unshuffling all $k$ unknown vertices in the null and only the $\ell$ vertices truly shuffled in the alternative.
As in the simulations, from the $k=\ell$ case, we see that matching recovers the lost shuffling power \emph{while} maintaining the desired testing level.  From the $\ell<k$ setting, we see that this recovered power is at the expense of testing level much greater than the desired $\alpha=0.05$.
In the event that the matching is viewed as a pre-processing step, matching more vertices in the null than the alternative could increase the true level of the test resulting in heightened type-I error risk (resulting in possibly artificially high testing power) as seen in the right panel of Figure \ref{fig:unshuffbrains}.

\begin{figure}[t!]
\begin{center}
\includegraphics[width=1\textwidth]{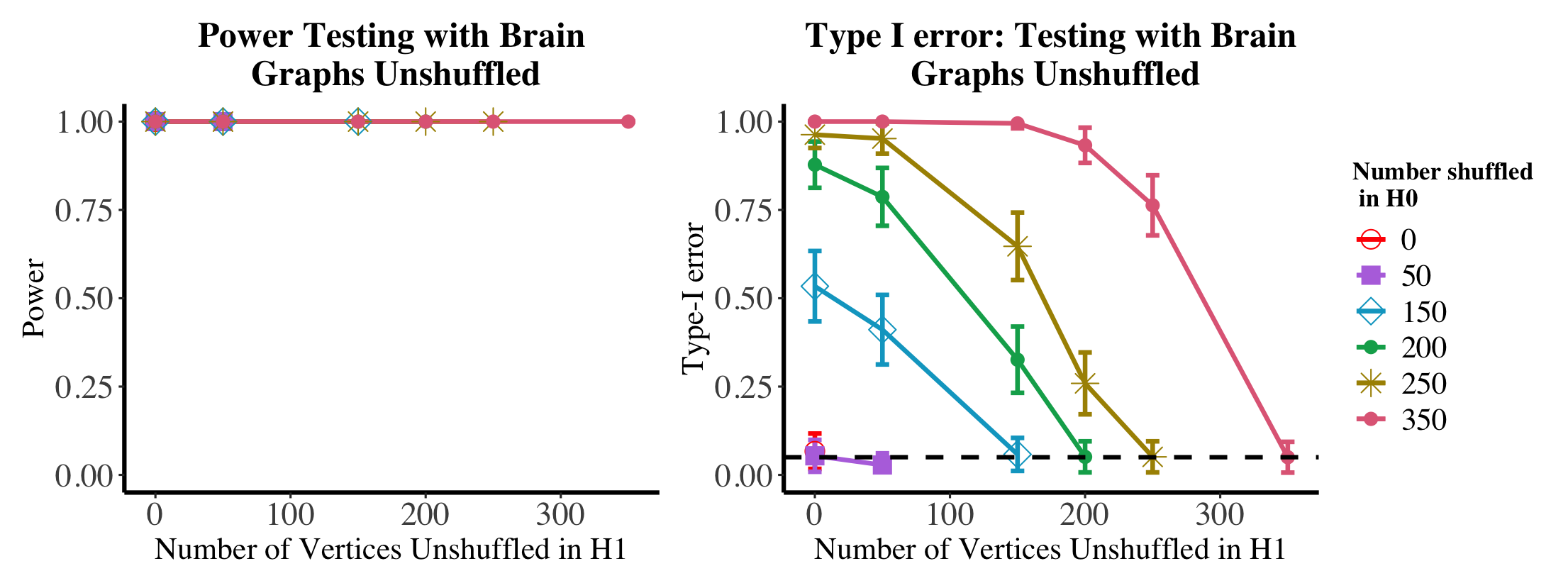}
\caption{Results for 200 bootstrap replicates to estimate power in the shuffled-then-unshuffled brain test.
In the figure the $x$-axis represents the number of vertices unshuffled via $\textbf{Q}_\ell$ (from $0$ to $k$) while the curve colors represent the maximum number of vertices potentially shuffled via $\Pi_{k,n}$, here all shuffled by $\bQ_k$.
The left panel displays testing power, the right type-I error with the dashed line at $y=0.05$; results are averaged over 100 Monte Carlo iterates (error bars are $\pm2$s.e.).}
\label{fig:unshuffbrains}
\end{center}
\end{figure}

%---------- Section----------------
\section{Conclusions}
\label{sec:conc}

As network data has become more commonplace, a number of statistical tools tailored for handling network data have been developed, including methods for hypothesis testing, goodness-of-fit analysis, clustering, and classification among others.  
Classically, 
many of the paired or multiple network inference tasks have assumed the vertices are aligned across networks, and this known alignment is often leveraged to improve subsequent inference.
Exploring the impact of shuffled/misaligned vertices on inferential performance---here on testing power---is an increasingly important problem as these methods gain more traction in noisy network domains.
By systematically breaking down a simple Frobenius-norm hypothesis test, we uncover and numerically analyze the decline of power as a function of both the distributional difference in the alternative and the number of shuffled nodes. 
Further analysis in a pair of real data settings reinforce our findings in giving practical guidance for the level of tolerable noise in paired, vertex-aligned testing procedures.

Our most thorough analysis of power loss is done in the context of random dot product and stochastic block models; bolstered by extensive simulation with real data experiments backing up our findings.  
While the goal of our research is to test the robustness of multiple network hypothesis testing methodologies, there still remains work to do in extending our findings to more general network models and to more complex network testing paradigms. 
A natural next step is to lift the simple Frobenius norm hypothesis test analysis to more broad and complex models, as well as test misalignment of vertices in the Omnibus and Semiparametric testing settings (see Section \ref{sec:ASE}).
Within the context of SBM's we aim to see how our power analysis changes for more esoteric shufflings (more blocks, different number flipped between blocks, etc.). 
Further extensions include extending the theory to non-edge independent graph models, and to shuffling in two sample tests where there are multiple graphs per sample, and where there is an interplay to explore in shuffling within and across populations.
In non-testing inference tasks, often the vertices are assumed aligned as well (e.g., tensor factorization, multiple graph embeddings, etc.), and exploring the inferential performance loss due to shuffled vertices in these settings is a natural next step.

In the event that vertex labels are incorrectly known, it is natural to use a graph matching/network alignment  methods to align the networks before proceeding with subsequent inference.
There are a host of matching procedures in the literature that could be applied to recover the true vertex alignment, and in doing so recover the lost testing power \cite{lyzinski2018information}.
We explore this in the context of our simulations and find that while matching recovers the lost power, ``over matching'' can result in artificially high power.  Care must be taken when using alignment tools for data pre-processing, as graph matching methods can induce artificial signal across even disparate network pairs (this is related to the ``phantom alignment strength'' phenomena of \cite{fishkind2021phantom}). 
Natural next steps in this direction include the following questions (among others): how the signal in an imperfectly recovered matching affects power loss as opposed to a random misalignment; 
how a probabilistic alignment (where the unknown in vertex labels is encoded into a stochastic matrix giving probabilities of alignment) can be incorporated into the testing framework; and how to use matching metrics (e.g., alignment strength \cite{fishkind2019alignment,fishkind2021phantom}) to estimate the size and membership of $U_{n,k}$ when this is unknown a priori.

\vspace{5mm}
\noindent\textbf{Acknowledgements:}
This material is based on research sponsored by the Air Force Research
Laboratory (AFRL) and Defense Advanced Research Projects Agency
(DARPA) under agreement number FA8750-20-2-1001. The U.S. Government is authorized to reproduce and distribute reprints for Governmental
purposes notwithstanding any copyright notation thereon. The views and
conclusions contained herein are those of the authors and should not be
interpreted as necessarily representing the official policies or endorsements,
either expressed or implied, of the AFRL and DARPA or the U.S.
Government.

 \bibliographystyle{elsarticle-num} 
 \bibliography{cite.bib}
 
\newpage

\appendix

%%%%%%%%%%%%%%%%%%%%%%%%%%%%%%%%%%%%%%%%

%%%%%%%%%%%%%%%%%%%%%%%%%%%
%%%%%%%%%%%%%%%%%%%%%%%%%%%

\section{Proofs of main and supporting results}
Herein, we collect the proofs of the main and supporting results from the paper.
We first present a useful theorem and corollary that will be used throughout.

In the shuffled testing analysis that we consider herein, we will use the following ASE consistency result from \cite{Rubin-Delanchy2017,du2023hypothesis}.
Note that we write here that a random variable $X$ is $\op(g(n))$ if for any constant $A>0$, there exists $n_0\in\mathbb{Z}^+$ and constant $B>0$ (both possibly depending on $A$) such that $\mathbb{P}(|X|\leq Bg(n))\geq 1-n^{-A}$ for all $n\geq n_0$.
\begin{theorem}
\label{thm:asecon}
Given Assumption \ref{ass:ass1}, let $\bA_n\sim RDPG(\bX_n,\nu_n)$ be a sequence of $d$-dimensional RDPGs, and let the adjacency spectral embedding of $\bA_n$ be given by $\hbX_n\sim$ASE$(\bA_n,d)$.
There exists a sequence of
orthogonal matrices $\mathbf{W}_n\in\mathcal{O}_d$ 
and a universal constant $c>0$ such that if the sparsity factor $\nu_n=\omega\left(\frac{\log^{4c}n}{n}\right)$, then (suppressing the dependence of $\bX$ and $\hbX$ on $n$)
\begin{equation}
\label{eq:2inf}
\max_{i=1,2,\ldots,n}\|\mathbf{W}_n\widehat X_i-\nu_n^{1/2}X_i\|=O_{\mathbb{P}}\left( \frac{\log^{c}n }{n^{1/2}}\right)
\end{equation}
\end{theorem}
From Eq. \ref{eq:2inf}, we can derive the following rough (though sufficient for our present needs) estimation bound on $\|\hbP-\bP\|_F$.
\begin{cor}
\label{cor:P-phat}
With notation and assumptions as in \hyperref[thm:asecon]{Theorem A.1}, let $\bP=\nu\bX\bX^T$ and $\hbP=\hbX\hbX^T$ (where $\hbX\sim$ASE$(\bA,d)$).  We then have 
\begin{equation}
    \label{eq:PPhat}
    \|\hbP-\bP\|_F=O_{\mathbb{P}}\left( \sqrt{n\nu_n}\log^{c}n\right)
\end{equation}
\end{cor}
\begin{proof}
Let $(\bW_n)$ be the sequence of $\bW$'s from \hyperref[thm:asecon]{Theorem A.1}.
Note first that (suppressing the subscript dependence on $n$)
\begin{align*}
|\nu X_i^TX_j-\hX_i^T\hX_j|&\leq 
|\nu^{1/2} X_i^T(\nu^{1/2}X_j-\bW\hX_j)|+|(\nu^{1/2}X_i^T-\hX_i^T\bW^T)\bW\hX_j|\\
&\leq \|\nu^{1/2}X_j-\bW\hX_j\|_2\|\nu^{1/2} X_i^T\|_2+
\|\nu^{1/2}X_i-\bW\hX_i\|_2\|\bW\hX_j\|_2\\
&\leq O_{\mathbb{P}}\left( \frac{\nu^{1/2}\log^{c}n }{n^{1/2}}\right)+O_{\mathbb{P}}\left( \frac{\log^{c}n }{n^{1/2}}\right)(\|\bW\hX_j-\nu^{1/2}X_j\|_2+\|\nu^{1/2}X_j\|_2)\\
&=O_{\mathbb{P}}\left( \frac{\nu^{1/2}\log^{c}n }{n^{1/2}}\right)
\end{align*}
Applying this entry-wise to $\|\hbP-\bP\|_F^2$ we get
\begin{align*}
\|\hbP-\bP\|_F^2&=\sum_{i,j} |\nu X_i^TX_j-\hX_i^T\hX_j|^2=O_{\mathbb{P}}\left( n\nu\log^{2c}n\right)
\end{align*}
as desired.
\end{proof}

%%%%%%%%%%%%%%%%%%%%%%%%%%%%%%%%%%%%

\subsection{Proof of Theorems 
\ref{thm:power_to_1} and \ref{thm:power_to_12}
}
\label{pf:power_to_1}
These proofs will proceed by using \hyperref[cor:P-phat]{Corollary A.1} to sharply bound the critical value of the level $\alpha$ test in terms of the error between the models (i.e., the difference of the $\bP$ matrices) and the sampling error (i.e., the difference between $\bP$ and $\widehat{\bP}$).  Growth rate analysis on the difference of the $\bP$ matrices will then allow for the detailed power analysis.

We will begin by recalling/establishing some notation.
Let $\bP_1=\mathbb{E}(A_1)$, $\bP_2=\mathbb{E}(A_2)$.
To ease notation moving forward, we will define
\begin{itemize}
\item[i.] For $i=1,2$, let $\hbP_i$ be
the ASE-based estimate of $\bP_i$ derived from $\bA_i$;
\item[ii.] For $i=1,2$, for any $\bQ\in\Pi_{n,k}$, let $\hbP_{i,\bQ}$ be
the ASE-based estimate of  $\bP_{i,\bQ}=\bQ\bP_2\bQ^T$ derived from $\bQ\bA_i\bQ^T$;
\end{itemize}
Given $\bA_1$ and $\bB_2=\tilde\bQ\bA_2(\tilde\bQ)^T$ (where $\tilde \bQ$ shuffles $\ell\leq k$ vertices of $U_{n,k}$), a valid (conservative) level-$\alpha$ test using the Frobenius norm test statistic would correctly reject $H_0$ if $$T(\bA_1,\bB_2):=\|\hbP_1-\hbP_{2,\tilde\bQ}\|_F \geq  \max_{\bQ\in\Pi_{n,k}}c_{\alpha,\bQ}.$$ 
Our first Proposition will bound $\max_{\bQ\in\Pi_{n,k}}c_{\alpha,\bQ}$ in terms of $\max_{\bQ\in\Pi_{n,k}}\|\bP_1-\bP_{1,\bQ}\|_F$ as follows: 

\begin{prop}
\label{prop:crit-value}
With notation and setup as above, we have that for any fixed $\alpha>0$, there exists an $M>0$ such that for all sufficiently large $n$, 
$$\max_{\bQ\in\Pi_{n,k}}c_{\alpha,\bQ}\leq \max_{\bQ\in\Pi_{n,k}}\|\bP_1-\bP_{1,\bQ}\|_F+M \sqrt{n\nu_n}\log^{c}n$$
and
$$\max_{\bQ\in\Pi_{n,k}}c_{\alpha,\bQ}\geq \max_{\bQ\in\Pi_{n,k}}\|\bP_1-\bP_{1,\bQ}\|_F-M \sqrt{n\nu_n}\log^{c}n.$$
\end{prop}
\begin{proof}
Note that under the null hypothesis, $\bP_1=\bP_2$.
As these critical values are computed under under the null hypothesis assumption, we shall make use of this throughout.
Note, however that $\hbP_1\neq\hbP_2$ in general, as these are estimated from $\bA_1$ and $\bA_2$, which are equal only in distribution.
Note that for the ASE-based estimate of $\bP_{2,\bQ}=\bP_{1,\bQ}$, we have
$\hbP_{2,\bQ}=\bQ\hbP_2\bQ^T$.
We then have
\begin{align*}
    \max_{\bQ\in\Pi_{n,k}}\|\hbP_1-\hbP_{1,\bQ}\|_F&\leq \max_{\bQ\in\Pi_{n,k}}\left( \|\hbP_1-\bP_1\|_F+
    \|\bP_1-\bP_{2,\bQ}\|_F+
    \|\bP_{2,\bQ}-\hbP_{2,\bQ}\|_F
    \right)\\
    &= \max_{\bQ\in\Pi_{n,k}}\left( \|\hbP_1-\bP_1\|_F+
    \|\bP_1-\bP_{2,\bQ}\|_F+
    \|\bQ(\bP_{2}-\hbP_{2})\bQ^T\|_F\right)\\
     &= \left(\max_{\bQ\in\Pi_{n,k}}\|\bP_1-\bP_{2,\bQ}\|_F\right) + \|\hbP_1-\bP_1\|_F+
    \|\bP_{2}-\hbP_2\|_F\\
    &= \max_{\bQ\in\Pi_{n,k}}\|\bP_1-\bP_{1,\bQ}\|_F+\op(\sqrt{n\nu_n}\log^{c}n)
\end{align*}
where the last line follows from \hyperref[cor:P-phat]{Corollary A.1}.
Therefore, for any $\epsilon\in(0,\alpha)$ there exists an $M_1>0$ and $N_1>0$ such that for any $n\geq N_1$, and any $\widehat\bQ\in\Pi_{n,k}$ we have that
\begin{align*}
\mathbb{P}_{H_0}&\left( \|\hbP_1-\hbP_{2,\widehat\bQ}\|_F> \max_{\bQ\in\Pi_{n,k}}\|\bP_1-\bP_{1,\bQ}\|_F+M_1 \sqrt{n\nu_n}\log^{c}n\right)\\
&\leq 
\mathbb{P}_{H_0}\left( \max_{\bQ\in\Pi_{n,k}}\|\hbP_1-\hbP_{2,\bQ}\|_F> \max_{\bQ\in\Pi_{n,k}}\|\bP_1-\bP_{1,\bQ}\|_F+M_1 \sqrt{n\nu_n}\log^{c}n\right)\leq \epsilon,
\end{align*}
implying (as $\widehat\bQ$ was chosen arbitrarily)
\begin{align*}
c_{\alpha,\widehat\bQ}&\leq \max_{\bQ\in\Pi_{n,k}}\|\bP_1-\bP_{1,\bQ}\|_F+M_1 \sqrt{n\nu_n}\log^{c}n\\
&\Rightarrow \max_{\bQ\in\Pi_{n,k}}c_{\alpha,\bQ}\leq \max_{\bQ\in\Pi_{n,k}}\|\bP_1-\bP_{1,\bQ}\|_F+M_1 \sqrt{n\nu_n}\log^{c}n.
\end{align*}

For the lower bound, recall that for $\widehat\bQ\in\Pi_{n,k}$, we have $c_{\alpha,\widehat\bQ}$ is the smallest value such that $\mathbb{P}_{H_0}(\|\hbP_1-\hbP_{2,\widehat\bQ}\|_F\geq c_{\alpha,\widehat \bQ})\leq \alpha$.
From the triangle inequality, we have that
\begin{align*}
    \|\hbP_1-\hbP_{2,\bQ}\|_F&\geq  -\|\hbP_1-\bP_1\|_F+
    \|\bP_1-\bP_{2,\bQ}\|_F-
    \|\bP_{2,\bQ}-\hbP_{2,\bQ}\|_F\\
    &\Leftrightarrow  \|\bP_{2,\bQ}-\hbP_{2,\bQ}\|_F +\|\hbP_1-\bP_1\|_F\geq  
    \|\bP_1-\bP_{2,\bQ}\|_F-\|\hbP_1-\hbP_{2,\bQ}\|_F
\end{align*}
so that for any $\epsilon_2>0$, there exists an $M_2$ and $N_2$ such that for all $n\geq N_2$ (recalling $\bP_1=\bP_2$ by assumption), 
\begin{align*}
\mathbb{P}_{H_0}&\left( 
\|\bP_{2,\bQ}-\hbP_{2,\bQ}\|_F +\|\hbP_1-\bP_1\|_F \leq  M_2\sqrt{n\nu_n}\log^{c}n\right)\geq1-\epsilon_2.\\
&\Rightarrow\mathbb{P}_{H_0}\left( 
\|\bP_1-\bP_{2,\bQ}\|_F-\|\hbP_1-\hbP_{2,\bQ}\|_F\leq  M_2\sqrt{n\nu_n}\log^{c}n\right)\geq1-\epsilon_2.\\
&\Leftrightarrow \mathbb{P}_{H_0}\left( \|\hbP_1-\hbP_{2,\bQ}\|_F\geq  \|\bP_1-\bP_{1,\bQ}\|_F-M_2\sqrt{n\nu_n}\log^{c}n\right)\geq1-\epsilon_2.
\end{align*}
This then implies that (for a well chosen $\epsilon_2$), that there exists an $M_2$ and $N_2$ such that for all $n\geq N_2$
$$\max_{\bQ\in\Pi_{n,k}} c_{\alpha,\bQ}\geq \max_{\bQ\in\Pi_{n,k}}\|\bP_1-\bP_{1,\bQ}\|_F-M_2 \sqrt{n\nu_n}\log^{c}n$$
Letting $M=\max(M_1,M_2)$ yields the desired result. 
\end{proof}

\noindent For ease of notation let $\bQ^*\in\text{argmax}_{\bQ\in\Pi_{n,k}}\|\bP_1-\bP_{1,\bQ}\|_F$, and define
\begin{align*}
T_{1,k,\ell}:=&\|\bP_1-\bP_{1,\bQ^*}\|^2_F-\|\bP_1-\bP_{1,\widetilde\bQ}\|^2_F;\\
T_{2,\ell,\ell}:=&\|\bP_1-\bP_{2,\widetilde\bQ}\|^2_F-\|\bP_1-\bP_{1,\widetilde\bQ}\|^2_F.
\end{align*}
We have then that (under the assumptions of Theorem \ref{thm:power_to_1}) there exists constants $C_1, C_2>0$ and an integer $n_0$ such that for $n\geq n_0$, the following holds with probability at least $1-n^{-2}$ under $H_1$,
\begin{align*}
\|\hbP_1-\hbP_{2,\widetilde\bQ}\|_F&\geq
% \|\bP_1-\bP_{2,\ell}\|_F -O\left( \sqrt{n\nu_n}\log^{c}n\right)
\left(\|\bP_1-\bP_{1,\bQ^*}\|^2_F-
T_{1,k,\ell}+T_{2,\ell,\ell}\right)^{1/2}-C_1 \sqrt{n\nu_n}\log^{c}n;\\
\|\hbP_1-\hbP_{2,\widetilde\bQ}\|_F&\leq
% \|\bP_1-\bP_{2,\ell}\|_F -O\left( \sqrt{n\nu_n}\log^{c}n\right)
\left(\|\bP_1-\bP_{1,\bQ^*}\|^2_F-
T_{1,k,\ell}+T_{2,\ell,\ell}\right)^{1/2}+C_2 \sqrt{n\nu_n}\log^{c}n.
\end{align*}
Recalling the form of $\mathbf{E}$, we have the following simplification of $T_{2,\ell,\ell}$ under $H_1$; 
first note that
\begin{align*}
\|\bP_1-\bP_{2,\widetilde\bQ}\|^2_F&=\|\bP_1-\bP_{1,\widetilde\bQ}\|^2_F+2\text{trace}\left(\left(\bP_{1,\widetilde\bQ^T}-\bP_{1}\right)^T\bE\right)+\|\bE\|_F^2,
\end{align*}
so that (where $C>0$ and $c>0$ are constants that can change line--to--line)
\begin{align*}
T_{2,\ell,\ell}\geq \begin{cases}
C n r\nu^2\epsilon^2-c nr\nu^2\epsilon &\text{ if }\ell\geq r\\
C n r\nu^2\epsilon^2-c n\ell\nu^2\epsilon&\text{ if }\ell\leq r
\end{cases}\quad\quad
T_{2,\ell,\ell}\leq \begin{cases}
C n r\nu^2\epsilon^2+c nr\nu^2\epsilon &\text{ if }\ell\geq r\\
C n r\nu^2\epsilon^2+c n\ell\nu^2\epsilon&\text{ if }\ell\leq r
\end{cases}
\end{align*}
We are now ready to prove Theorem \ref{thm:power_to_1}, which we state here for completeness.  Recall, we are concerned with showing that for all sufficiently large $n$,
\begin{equation}
\label{eq:power_to_1_2ap}
\mathbb{P}_{H_1}\left(\|\hbP_1-\hbP_{2,\tilde\bQ}\|_F>\max_{\bQ\in\Pi_{n,k}}c_{\alpha,\bQ}\right)\geq 1-n^{-2}.
\end{equation}
\begin{manualtheorem}{\ref{thm:power_to_1}}
\label{thm:power2}
With notation as above, assume there exist $\alpha\in(0,1]$ such that $r=\Theta(n^{\alpha})$ and 
$k,\ell\ll n^\alpha$, and that $$\frac{\|\bP_1-\bP_{1,\bQ^*}\|_F^2-\|\bP_1-\bP_{1,\widetilde\bQ}\|_F^2 }{\nu^2 n}=O(k).$$
In the sparse setting, consider $\nu\gg\frac{\log^{4c}(n)}{n^\beta}$ for $\beta\in(0,1]$ where $\alpha\geq\beta$.  If either \begin{itemize}
\item[i.]$k=O\left(\frac{n^\beta}{\log^{2c}n}\right)$ 
and
$\epsilon\gg  \sqrt{\frac{n^{\beta-\alpha}}{\log^{2c}(n)}}$;
\item[ii.] $k\gg \frac{n^\beta}{\log^{2c}(n)}$ and $\epsilon\gg \sqrt{\frac{k}{n^\alpha}}$
\end{itemize}
then Eq. \ref{eq:power_to_1_2ap} holds for all $n$ sufficiently large.
In the dense case where $\nu=1$, if either
\begin{itemize}
\item[i.]$k\gg\log^{2c}(n)\text{ and }\epsilon\gg \sqrt{k/n^{\alpha}};$ 
\item[ii.]$k\ll\log^{2c}(n)\text{ and }\epsilon\gg\sqrt{(\log^{c}(n))/ n^{\alpha}},$
\end{itemize}
then Eq. \ref{eq:power_to_1_2ap} holds for all $n$ sufficiently large.
\end{manualtheorem}
\begin{proof}
Note that we will see below that we will require $k< c_2 n^\alpha$ for an appropriate constant $c_2$, so the assumption that $k\ll n^\alpha$ is not overly stringent.
We begin by noting that for sufficiently large $n$ (as $\ell\ll r$)
\begin{align}
T_{2,\ell,\ell}-T_{1,k,\ell}\geq& 
C n r\nu^2\epsilon^2-c n\ell\nu^2\epsilon 
-\|\bP_1-\bP_{1,\bQ^*}\|^2_F
+\|\bP_1-\bP_{1,\widetilde\bQ}\|^2_F.
\label{eq:offset12}
\end{align}
Now, for power to be asymptotically almost surely 1 (i.e., bounded below by $1-n^{-2}$ for all $n$ sufficiently large), it suffices that under $H_1$ (as the critical value for the hypothesis test is bounded above by Proposition \ref{prop:crit-value} by $\|\bP_1-\bP_{1,\bQ^*}\|_F+M \sqrt{n\nu}\log^{c}n$),
\begin{align*}
\|\hbP_1-\hbP_{2,\widetilde\bQ}\|_F&\geq
\left(\|\bP_1-\bP_{1,\bQ^*}\|^2_F-
T_{1,k,\ell}+T_{2,\ell,\ell}\right)^{1/2}-C_1 \sqrt{n\nu}\log^{c}n\\
&\geq \|\bP_1-\bP_{1,\bQ^*}\|_F+M \sqrt{n\nu}\log^{c}n,
\end{align*}
which is implied by
\begin{align}
C n r\nu^2\epsilon^2\geq&
c n\ell\nu^2\epsilon 
+(\|\bP_1-\bP_{1,\bQ^*}\|^2_F
-\|\bP_1-\bP_{1,\widetilde\bQ}\|^2_F) +(M+C_1)^2 n\nu\log^{2c}n\notag\\
&\hspace{5mm}+2\|\bP_1-\bP_{1,\bQ^*}\|_F(M+C_1) \sqrt{n\nu}\log^{c}n.
\label{eq:keypower22}
\end{align}
To show Eq. \ref{eq:keypower22} holds, it suffices that all of the following hold 
\begin{align}
\epsilon \gg \frac{\ell}{n^\alpha};\quad\quad 
\epsilon\gg \sqrt{\frac{k}{n^{\alpha}}  }
;\quad\quad
\epsilon\gg  \sqrt{\frac{\log^{2c}n}{n^\alpha\nu}};\quad\quad \epsilon\gg
\sqrt{\frac{k^{1/2} \log^{c}n  }{n^{\alpha}\nu^{1/2}}} 
\label{eq:eqneed14k}
\end{align}
In the sparse setting where $\nu\gg\frac{\log^{4c}(n)}{n^\beta}$ for $\beta\in(0,1]$, we have
Eq. \ref{eq:eqneed14k} is implied by the following (as $k,\ell\ll n^\alpha$ and $k\geq \ell$ so that $\frac{k}{n^\alpha}\ll 1$ and then $\sqrt{\frac{k}{n^\alpha}}\gg\frac{k}{n^\alpha}\geq \frac{\ell}{n^\alpha}$)
\begin{align}
    \epsilon\gg\sqrt{\frac{k}{n^\alpha}};\quad\quad\quad\epsilon\gg  \sqrt{\frac{n^{\beta-\alpha}}{\log^{2c}(n)}};\quad\quad\quad
    \epsilon\gg\sqrt{\frac{ n^{\beta-\alpha}}{\log^{c}(n)} \cdot\frac{k^{1/2} }{n^{\beta/2}}};
    \label{eq:eqneed143k}
\end{align}
note that for these equations to be possible, we must have $\alpha\geq\beta$. 
If $k=O\left(\frac{n^\beta}{\log^{2c}n}\right)$, then 
Eq. \ref{eq:eqneed143k} is implied by 
$\epsilon\gg  \sqrt{\frac{n^{\beta-\alpha}}{\log^{2c}(n)}}$.  Consider next $k\gg \frac{n^\beta}{\log^{2c}(n)}$, in which case Eq. \ref{eq:eqneed143k} are implied by  $\epsilon\gg\sqrt{\frac{k}{n^\alpha}}$.

In the dense case where $\nu=1$, we have that Eq. \ref{eq:eqneed14k} is implied by the following
\begin{align}
\epsilon\gg\sqrt{\frac{k}{n^{\alpha}}  };\quad\quad
\epsilon&\gg  \sqrt{\frac{\log^{2c}n}{n^\alpha}};\quad\quad
 \epsilon\gg
\sqrt{\frac{k^{1/2} \log^{c}n  }{n^{\alpha}}}
\label{eq:eqneed144k}
\end{align}
Eq. \ref{eq:eqneed144k} is then implied by
\begin{align*}
\epsilon\gg \sqrt{\frac{k}{n^{\alpha}}}&\,\,\text{ if }\,\,k\gg\log^{2c}(n);\\
\epsilon\gg\sqrt{\frac{\log^{c}(n) }{n^{\alpha}}}&\,\,\text{ if }\,\,k\ll\log^{2c}(n).
\end{align*} 
as desired.
\end{proof}

\noindent We next turn our attention to Theorem \ref{thm:power_to_12}.
\begin{manualtheorem}
{\ref{thm:power_to_12}}
\label{thm:power3}
With notation as above, assume there exist $\alpha\in(0,1]$ such that $r=\Theta(n^{\alpha})$ and $k,\ell\ll n^{\alpha}$, and that $$\frac{\|\bP_1-\bP_{1,\bQ^*}\|_F^2-\|\bP_1-\bP_{1,\widetilde\bQ}\|_F^2 }{\nu^2 n}=O(k-\ell).$$
In the sparse setting where $\nu\gg\frac{\log^{4c}(n)}{n^\beta}$ for $\beta\in(0,1]$ where $\alpha\geq\beta$, if $k\gg \frac{n^\beta}{\log^{2c}(n)}$ and either
\begin{itemize}
\item[i.] $\frac{k-\ell}{k^{1/2}}\geq\frac{n^{\beta/2}}{\log^{2c}(n)}$; $\epsilon\gg \frac{\ell}{n^\alpha}$; and $\epsilon\gg \sqrt{\frac{k-\ell}{n^\alpha}}$; or
\item[ii.]  $\frac{k-\ell}{k^{1/2}}\leq\frac{n^{\beta/2}}{\log^{2c}(n)}$;  
$\epsilon\gg\frac{\ell}{n^\alpha}$; and $\epsilon\gg\sqrt{\frac{n^{\beta/2}}{\log^{2c}(n)}\frac{k^{1/2} }{n^{\alpha}}}$
\end{itemize}
then Eq. \ref{eq:power_to_1_2ap} holds for all $n$ sufficiently large.
In the dense case where $\nu=1$ and $k=\omega(\log^{2c}n)$, if either
\begin{itemize}
\item[i.] $\frac{k-\ell}{k^{1/2}}\geq\log^{c}(n)$; $\epsilon\gg \frac{\ell}{n^\alpha}$; and $\epsilon\gg \sqrt{\frac{(k-\ell)}{n^{\alpha}}}$; or
\item[ii.] $\frac{k-\ell}{k^{1/2}}\leq\log^{c}(n)$; $\epsilon\gg \frac{\ell}{n^\alpha}$; and $\epsilon\gg \sqrt{\frac{k^{1/2}\log^{c}(n) }{n^{\alpha}}}$,
\end{itemize}
then Eq. \ref{eq:power_to_1_2ap} holds for all $n$ sufficiently large.
\end{manualtheorem}

\begin{proof}
Mimicking the proof of Theorem \ref{thm:power_to_1}, for Eq. \ref{eq:power_to_1_2ap} to hold, it suffices that all of the following hold 
\begin{align}
\epsilon \gg \frac{\ell}{n^\alpha};\quad\quad 
\epsilon\gg \sqrt{\frac{k-\ell}{n^{\alpha}}  }
;\quad\quad
\epsilon\gg  \sqrt{\frac{\log^{2c}n}{n^\alpha\nu}};\quad\quad \epsilon\gg
\sqrt{\frac{k^{1/2} \log^{c}n  }{n^{\alpha}\nu^{1/2}}}
\label{eq:eqneed14}
\end{align}
In the sparse setting where $\nu\gg\frac{\log^{4c}(n)}{n^\beta}$ for $\beta\in(0,1]$, we have
Eq. \ref{eq:eqneed14} is implied by the following 
\begin{align}
\epsilon\gg \frac{\ell}{n^\alpha};\quad\quad
\epsilon\gg \sqrt{\frac{k-\ell}{n^\alpha}};\quad\quad
    \epsilon\gg  \sqrt{\frac{n^{\beta-\alpha}}{\log^{2c}(n)}};\quad\quad
    \epsilon\gg\sqrt{\frac{ n^{\beta-\alpha}}{\log^{c}(n)} \cdot\frac{k^{1/2} }{n^{\beta/2}}};
    \label{eq:eqneed143}
\end{align}
note that for these equations to be possible, we must have $\alpha\geq\beta$. 
Recalling the assumption that $k\gg \frac{n^\beta}{\log^{2c}(n)}$, the behavior in this case hinges on $\ell$ as well, as Eq. \ref{eq:eqneed143} is implied by
\begin{align*}
% \epsilon\gg\sqrt{\frac{k}{n^\alpha}}&\,\,\text{ if }\,\,\frac{k-\ell}{k}=\Theta(1);\\
\epsilon\gg \frac{\ell}{n^\alpha},\quad \epsilon\gg \sqrt{\frac{k-\ell}{n^\alpha}}&\,\,\text{ if }\,\,
% \frac{k-\ell}{k}=o(1),\, 
\frac{k-\ell}{k^{1/2}}\geq \frac{n^{\beta/2}}{\log^{2c}(n)};\\
\epsilon\gg \frac{\ell}{n^\alpha},\quad \epsilon\gg\sqrt{\frac{n^{\beta/2}}{\log^{2c}(n)}\frac{k^{1/2} }{n^{\alpha}}}&\,\,\text{ if }
% \,\,\frac{k-\ell}{k}=o(1),\, 
\frac{k-\ell}{k^{1/2}}\leq\frac{n^{\beta/2}}{\log^{2c}(n)};
\end{align*}
In the dense case where $\nu=1$, we have that Eq. \ref{eq:eqneed14} is implied by the following
\begin{align}
\epsilon \gg\frac{\ell}{n^\alpha};\quad\quad
\epsilon\gg\sqrt{\frac{k-\ell}{n^{\alpha}}  };\quad\quad
\epsilon\gg  \sqrt{\frac{\log^{2c}n}{n^\alpha}};\quad\quad
 \epsilon\gg
\sqrt{\frac{k^{1/2} \log^{c}n  }{n^{\alpha}}}
\label{eq:eqneed144}
\end{align}
Under the assumption that $k\gg \log^{2c}n$, Eq. \ref{eq:eqneed144} is implied by 
\begin{align*}
\epsilon\gg\frac{\ell}{n^\alpha},\quad\epsilon\gg\sqrt{\frac{(k-\ell)}{n^{\alpha}}}&\,\,\text{ if }\,\,\frac{k-\ell}{k^{1/2}}\geq\log^{c}(n);\\
\epsilon\gg\frac{\ell}{n^\alpha},\quad \epsilon\gg\sqrt{\frac{k^{1/2}\log^{c}(n) }{n^{\alpha}}}&\,\,\text{ if }\,\,\frac{k-\ell}{k^{1/2}}\leq\log^{c}(n),
\end{align*} 
as desired.
\end{proof}

\subsection{Proof of Theorem \ref{thm:power_to_0}}
\label{pf:power_to_0}

Recall that 
$$T_{2,\ell,\ell}-T_{1,k,\ell}\leq \begin{cases}
C n r\nu^2\epsilon^2+c nr\nu^2\epsilon-\|\bP_1-\bP_{1,\bQ^*}\|^2_F
+\|\bP_1-\bP_{1,\widetilde\bQ}\|^2_F &\text{ if }\ell\geq r\\
C n r\nu^2\epsilon^2+c n\ell\nu^2\epsilon-\|\bP_1-\bP_{1,\bQ^*}\|^2_F
+\|\bP_1-\bP_{1,\widetilde\bQ}\|^2_F&\text{ if }\ell\leq r
\end{cases}
$$
so that the assumption $$\frac{\|\bP_1-\bP_{1,\bQ^*}\|_F^2-\|\bP_1-\bP_{1,\widetilde\bQ}\|_F^2 }{\nu^2 n}=\Omega(k-\ell)$$ yields that there exists a constant $R>0$ such that for $n$ sufficiently large
$$T_{2,\ell,\ell}-T_{1,k,\ell}\leq \begin{cases}
C n r\nu^2\epsilon^2+c nr\nu^2\epsilon-R n\nu^2(k-\ell) &\text{ if }\ell\geq r\\
C n r\nu^2\epsilon^2+c n\ell\nu^2\epsilon-R n\nu^2(k-\ell)&\text{ if }\ell\leq r
\end{cases}
$$
Now, for power to be asymptotically almost surely 0 (i.e., bounded below by $n^{-2}$ for all $n$ sufficiently large), it suffices that under $H_1$ we have, where $\mathfrak{c}$ is an appropriate constant that can change line-to-line (as the critical value for the hypothesis test is bounded below by Proposition \ref{prop:crit-value} by $\|\bP_1-\bP_{1,\bQ^*}\|_F-M \sqrt{n\nu}\log^{c}n$)
\begin{align}
\|\hbP_1-\hbP_{2,\widetilde\bQ}\|_F&\leq
\left(\|\bP_1-\bP_{1,\bQ^*}\|^2_F-
T_{1,k,\ell}+T_{2,\ell,\ell}\right)^{1/2}+C_2 \sqrt{n\nu}\log^{c}n\notag\\
&\leq \|\bP_1-\bP_{1,\bQ^*}\|_F-M \sqrt{n\nu}\log^{c}n
% &\Leftarrow T_{2,\ell,\ell}-
% T_{1,k,\ell}\leq (C_2+M)^2n\nu \log^{2c}n
% -2(M+C_2) \sqrt{n\nu}\log^{c}n\cdot\|\bP_1-\bP_{2,\bQ^*}\|_F
\label{eq:cutline}
\end{align}
Eq. \ref{eq:cutline} is then implied by 
$$
T_{2,\ell,\ell}-
T_{1,k,\ell}\leq (C_2+M)^2n\nu \log^{2c}n
-2(M+C_2) \sqrt{n\nu}\log^{c}n\cdot\|\bP_1-\bP_{2,\bQ^*}\|_F
$$
which is implied by
\begin{align}
\label{eq:pwr0}
\begin{cases}
n^{1+\alpha}\nu^2\epsilon^2+ n^{1+\alpha}\nu^2\epsilon\leq \mathfrak{c} n\nu^2(k-\ell) + \mathfrak{c} n\nu \log^{2c}n
-\mathfrak{c} n\nu^{3/2}\sqrt{k}\log^{c}n &\text{ if }\ell\geq r\\
n^{1+\alpha}\nu^2\epsilon^2+ n\ell\nu^2\epsilon\leq \mathfrak{c} n\nu^2(k-\ell) + \mathfrak{c} n\nu \log^{2c}n
-\mathfrak{c} n\nu^{3/2}\sqrt{k}\log^{c}n &\text{ if }\ell\leq r
\end{cases}
\end{align}
Suppose further that
$$\frac{k-\ell}{\sqrt{k}}\gg \frac{\log^cn}{\sqrt{\nu}}.$$
In this case Eq. \ref{eq:pwr0} is implied by
\begin{align}
 n^{1+\alpha}\nu^2\epsilon^2&+ n^{1+\alpha}\nu^2\epsilon\ll  n\nu\log^{2c}n+n(k-\ell)\nu^2 \quad\text{ if }\ell\geq r\notag\\
 n^{1+\alpha}\nu^2\epsilon^2&+ n\ell\nu^2\epsilon\ll  n\nu\log^{2c}n+n(k-\ell)\nu^2\quad\text{ if }\ell\leq r\notag\\
&\Leftrightarrow
\begin{cases}
 \epsilon^2+ \epsilon\ll  \frac{\log^{2c}n}{n^\alpha\nu}+\frac{k-\ell}{n^\alpha} \quad&\text{ if }\ell\geq r\\
 n^{\alpha}\epsilon^2+ \ell\epsilon\ll  \frac{\log^{2c}n}{\nu}+(k-\ell)\quad&\text{ if }\ell\leq r
\end{cases}\notag\\
&\Leftarrow
\begin{cases}
 \epsilon\ll  \frac{k-\ell}{n^\alpha} \quad&\text{ if }\ell\geq r\\
\epsilon\ll  \sqrt{\frac{k-\ell}{n^\alpha}};\,\, \epsilon\ll  \frac{k-\ell}{\ell}\quad&\text{ if }\ell\leq r
\end{cases}
\end{align}
as desired.
%%%%%%%%%%%%%%%%%%%%%%%%%%%%%%%%%%%%%%%%%%

\subsection{Proof of Proposition \ref{prop:AP}}
\label{sec:APpf}

To ease notation, define $T_{h,A}=T_A(\bA_1,\bA_{2,h}):=\frac{1}{2}\|\bA_1-\bA_{2,h}\|_F^2$.
We will adopt the following notations for the entries of the shuffled edge expectation matrices: for $\eta=1,2$, the $(i,j)$-th entry of $\bP_{\eta,h}=\bQ_{2h}\bP_\eta\bQ_{2h}^T$ is denoted
via $p_{ij}^{(\eta,h)}$ (and where $p_{i,j}^{(\eta)}$ will refer to the $(i,j)$-th entry of $\bP_\eta$).
We will also define (recall, $\sum_{\{i,j\}}$ signifies the sum over unordered pairs of elements of $[n]$)
$\mu_{1}:=\sum_{\{i,j\}}p_{i,j}^{(1)}$, and  $\mu_E :=\sum_{\{i,j\}} e_{ij}.$
  Then, we have that (where we recall that $\bE_\ell=\bQ_{2\ell} \bE\bQ_{2\ell}^T=[e_{i,j}^{(\ell)}]$ is the shuffled noise matrix)
  \begin{align*}
\mathbb{E}_{H_0}(T_{k,A})
 &=2\mu_1-2\sum_{\{i,j\}} p_{ij}^{(1)}p_{ij}^{(1,k)}\\
 \mathbb{E}_{H_1}(T_{\ell,A}) 
 &=2\mu_1+\mu_E-2\sum_{\{i,j\}} p^{(1)}_{ij}(p^{(1,\ell)}_{ij}+e^{(\ell)}_{ij})
\end{align*}
So that in the dense setting (i.e., $\nu_n= 1$ for all $n$), letting $\xi_{ij}:=(2p^{(1)}_{ij}-1)e^{(\ell)}_{ij}$ and 
$\mu_\xi:=\sum_{\{ij\}}\xi_{ij}$, we have that (where to ease notation, we define $\delta=\max_i |\Lambda_{1i}-\Lambda_{2i}|$)
\begin{align}
    \mathbb{E}_{H_0}(T_{k,A})-&\mathbb{E}_{H_1}(T_{\ell,A})=2\sum_{\{ij\}} p^{(1)}_{ij}(p^{(1,\ell)}_{ij}-p^{(1,k)}_{ij})-\mu_E+ 2 \sum_{\{ij\}}p^{(1)}_{ij}e^{(\ell)}_{ij}\notag\\
    &=2\bigg(\sum_{h>2}n_h(k-\ell)(\Lambda_{1h}-\Lambda_{2h})^2+\binom{k-\ell}{2}(\Lambda_{11}-\Lambda_{22})^2\label{eq:rowdiff1}\\
    &\hspace{5mm}+(n_1-k)(k-\ell)(\Lambda_{11}-\Lambda_{21})^2+(n_2-k)(k-\ell)(\Lambda_{22}-\Lambda_{12})^2 \label{eq:rowdiff2}\\
    &\hspace{5mm}- 2\ell(k-\ell)(\Lambda_{11}-\Lambda_{12})(\Lambda_{22}-\Lambda_{21})\label{eq:rowdiff3}\bigg)\\
    &\hspace{5mm}-\mu_E+ 2 \sum_{\{ij\}}p^{(1)}_{ij}e^{(\ell)}_{ij}\label{eq:rowdiff4}\\
    &\leq 2n(k-\ell)\delta^2+\mu_\xi
\end{align}
To see this, we first focus
on bounding the terms in 
Eqs. (\ref{eq:rowdiff1})--(\ref{eq:rowdiff3}).
To ease notation, let $x:=\Lambda_{11}-\Lambda_{21}$ and $y:=\Lambda_{22}-\Lambda_{12}$.
For the desired bound, it suffices to show that 
$$
\binom{k-\ell}{2}(x-y)^2-k(k-\ell)(x^2+y^2)-2\ell(k-\ell)xy\leq 0.
$$
To see this, note that 
\begin{align*}
\binom{k-\ell}{2}&(x-y)^2-k(k-\ell)(x^2+y^2)-2\ell(k-\ell)xy\\
&=\frac{k-\ell}{2}\left((k-\ell-1)(x^2+y^2)-2(k-\ell-1)xy
-2k(x^2+y^2)-4\ell xy \right)\\
&=\frac{k-\ell}{2}\left((-k-\ell-1)(x^2+y^2)-2(k-\ell-1)xy-4\ell xy \right)\leq 0,
\end{align*} as this yields the terms in Eqs. (\ref{eq:rowdiff1})--(\ref{eq:rowdiff3}) are bounded above by
$2\sum_{h\geq 1}n_h(k-\ell)(\Lambda_{1h}-\Lambda_{2h})^2$
as desired.
For a lower bound, we have that (where $\gamma:=|\Lambda_{11}-\Lambda_{22}|$)
\begin{align*}
\mathbb{E}_{H_0}(T_{k,A})-\mathbb{E}_{H_1}(T_{\ell,A})&
\geq 2(k-\ell)\left[
\sum_{h\geq 1}n_h(\Lambda_{1h}-\Lambda_{2h})^2
-(k+\ell)\delta^2-\gamma^2/2
\right]
+\mu_{\xi}\\
&\geq2(k-\ell)[\min_i n_i-(k+\ell)]\delta^2-(k-\ell)\gamma^2+\mu_{\xi}
\end{align*}
To derive the desired lower bound on the terms in Eqs. (\ref{eq:rowdiff1})--(\ref{eq:rowdiff3}), we see that
\begin{align*}
\binom{k-\ell}{2}&(x-y)^2-k(k-\ell)(x^2+y^2)-2\ell(k-\ell)xy\\
&=-\frac{k-\ell}{2}((k+\ell)(x+y)^2+(x-y)^2)\\
&\geq-\frac{k-\ell}{2}((k+\ell)2\delta^2+\gamma^2)=-(k-\ell)(k+\ell)\delta^2-(k-\ell)\gamma^2/2
\end{align*}
as desired.

Under our assumptions on $\Lambda$ and $\bE$, we have that 
\begin{align*}
2\eta(1-\eta)\mathbb{E}_{H_0}(T_{k,A}) &\leq \text{Var}_{H_0}(T_{k,A})\leq
(1-2\eta(1-\eta)) \mathbb{E}_{H_0}(T_{k,A})\\
2(\eta-\hat\eta)(1-\eta+\hat\eta)\mathbb{E}_{H_1}(T_{\ell,A}) &\leq \text{Var}_{H_1}(T_{\ell,A})\leq
(1-2(\eta-\hat\eta)(1-\eta+\hat\eta)) \mathbb{E}_{H_1}(T_{\ell,A})
\end{align*}
To see this for $T_{k,A}$ (with $T_{\ell,A}$ being analogous), let 
$\sigma_k$ (resp., $\sigma_\ell$) be the permutation associated with $\bQ_{2k}$ (resp., $\bQ_{2\ell}$), and define
  \begin{align*}
\csp&:=\{\, \{i,j\}\text{ s.t. }\{i,j\}\neq \{\sigma_k(i),\sigma_k(j)\}\}\\
\csq&:=\{\, \{i,j\}\text{ s.t. }\{i,j\}\neq \{\sigma_\ell(i),\sigma_\ell(j)\}\}.
  \end{align*}
For ease of notation, define
\begin{align*}
\mathfrak{a}_{ij}:&= p^{(1)}_{i,j}(1-p_{ij}^{(1,k)}) + p_{ij}^{(2,k)}(1-p^{(1)}_{ij})\\
\mathfrak{b}_{ij}:&= 2p^{(1)}_{i,j}(1-{p}^{(1)}_{ij}) 
\end{align*}
Then
\begin{align*}
\text{Var}_{H_0}(T_{k,A})= \mathbb{E}_{H_0}(T_{k,A})-\sum_{\{ij\}\in\csp}\mathfrak{a}_{ij}^2 -\sum_{\{ij\}\notin\csp} \mathfrak{b}_{ij}^2;
\end{align*}
noting that 
\begin{align*}
2\eta(1-\eta)&\leq \mathfrak{a}_{ij}\leq 1-2\eta(1-\eta)\\
2\eta(1-\eta)&\leq \mathfrak{b}_{ij} \leq 1-2\eta(1-\eta)
\end{align*}
yields
\begin{align*}
\text{Var}_{H_0}(T_{k,A})&\geq \mathbb{E}_{H_0}(T_{k,A})-(1-2\eta(1-\eta))\!\!\sum_{\{ij\}\in\csp}\!\!\mathfrak{a}_{ij} -(1-2\eta(1-\eta))\!\!\sum_{\{ij\}\notin\csp} \!\!\mathfrak{b}_{ij}\\
&= 2\eta(1-\eta)\mathbb{E}_{H_0}(T_{k,A})
\end{align*}
and, similarly, $\text{Var}_{H_0}(T_{k,A})
\leq(1-2\eta(1-\eta))\mathbb{E}_{H_0}(T_{k,A})$.

Stein's method (see \cite{stein1986approximate,ross2011fundamentals}) yields that under both $H_0$ (resp., $H_1$) $T_A(\bA_1,\bA_{2,k})$ (resp., $T_A(\bA_1,\bA_{2,\ell})$) are aymptotically normally distributed, and hence the testing power is asymptotically equal to (where $\widetilde\Phi$ is the standard normal tail CDF, and $C>0$ is a constant that can change line-to-line, and $n_*=\min_i n_i$) 
\begin{align*}
\mathbb{P}_{H_1}&(T_{\ell,A}\geq \mathfrak{c}_{\alpha,k})\approx\mathbb{P}_{H_1}(T_{\ell,A} \geq z_{\alpha}\sqrt{\text{Var}_{H_0}(T_{k,A})} + \mathbb{E}_{H_0}(T_{k,A})) \\
    &= \mathbb{P}_{H_1}\left(\frac{T_{\ell,A}-\mathbb{E}_{H_1}(T_{\ell,A})}{\sqrt{\text{Var}_{H_1}(T_{\ell,A})}}\geq \frac{z_{\alpha}\sqrt{\text{Var}_{H_0}(T_{k,A})} + \mathbb{E}_{H_0}(T_{k,A}) - \mathbb{E}_{H_1}(T_{\ell,A})}{\sqrt{\text{Var}_{H_1}(T^{(a)}_{\ell,A})}}\right)\\
    &\approx\widetilde\Phi\left( \frac{z_{\alpha}\sqrt{\text{Var}_{H_0}(T_{k,A})} + \mathbb{E}_{H_0}(T_{k,A}) - \mathbb{E}_{H_1}(T_{\ell,A})}{\sqrt{\text{Var}_{H_1}(T_{\ell,A})}}\right)\\
    &\leq\widetilde\Phi\left( C\cdot\left(
        2(k-\ell)\left(\frac{n_*}{n}-\frac{k+\ell}{n}\right)\delta^2-\frac{k-\ell}{n}\gamma^2+\frac{\mu_{\xi}}{n}\right)\right)\\
    &\leq\widetilde\Phi\left( C\cdot\left(
        (k-\ell)\frac{n_*}{n}-\frac{k^2-\ell^2}{n}\delta^2-\frac{k-\ell}{n}\gamma^2+\frac{\mu_{\xi}}{n}\right)\right)
\end{align*}
Now, we have that power is asymptotically negligible if 
$$(k-\ell)\frac{n_*}{n}-\frac{k^2-\ell^2}{n}\delta^2-\frac{k-\ell}{n}\gamma^2+\frac{\mu_{\xi}}{n}\gg0$$
as desired.

%%%%%%%%%%%%%%%%%%%%%%%%%%%%%%%%%%%%%%%
\newpage
\subsection{Additional Experiments}
\label{app:addexp}

Herein, we include the additional experiments from Sections \ref{sec:PhatvsA} and \ref{sec:ASE}.
We first show an example of the $\widehat\bP$ test in a sparser regime than that considered in Section \ref{sec:PhatvsA}.
As in Section \ref{sec:PhatvsA}, we consider
$b(v)=2-\mathds{1}\{v\in\{1,2,\cdots,250\}$, we consider two $n=500$ vertex SBMs defined via
\begin{equation}
\label{eq:AandE2}
\bA\sim \text{SBM}\left(2, 
\begin{bmatrix}
0.05 & 0.01\\
0.01 & 0.04
\end{bmatrix}
,b,1\right);\quad
\bB\sim \text{SBM}\left(2, 
\begin{bmatrix}
0.05 & 0.01\\
0.01 & 0.04
\end{bmatrix}+\textbf{E}_{\epsilon},b,1\right)
\end{equation}
where $\textbf{E}_{\epsilon}=\epsilon J_{500}$ for $\epsilon=0.01, 0.05, -0.005$.
Results are displayed in Figure \ref{sparsePhat}; here we see the same trend in play for the modestly sparse regime (when $\epsilon=0.01,0.05$). 
Indeed, in this sparse regime, the testing power is high even for low noise ($\epsilon=0.01$), as the signal--to--noise ratio is still favorable for inference.  As in the dense case, if the error is too small (here $\epsilon=-0.005$), the test is unable to distinguish the two networks irregardless of the shuffling effect.

\begin{figure}[t!]
\centering
\includegraphics[width=1\textwidth]{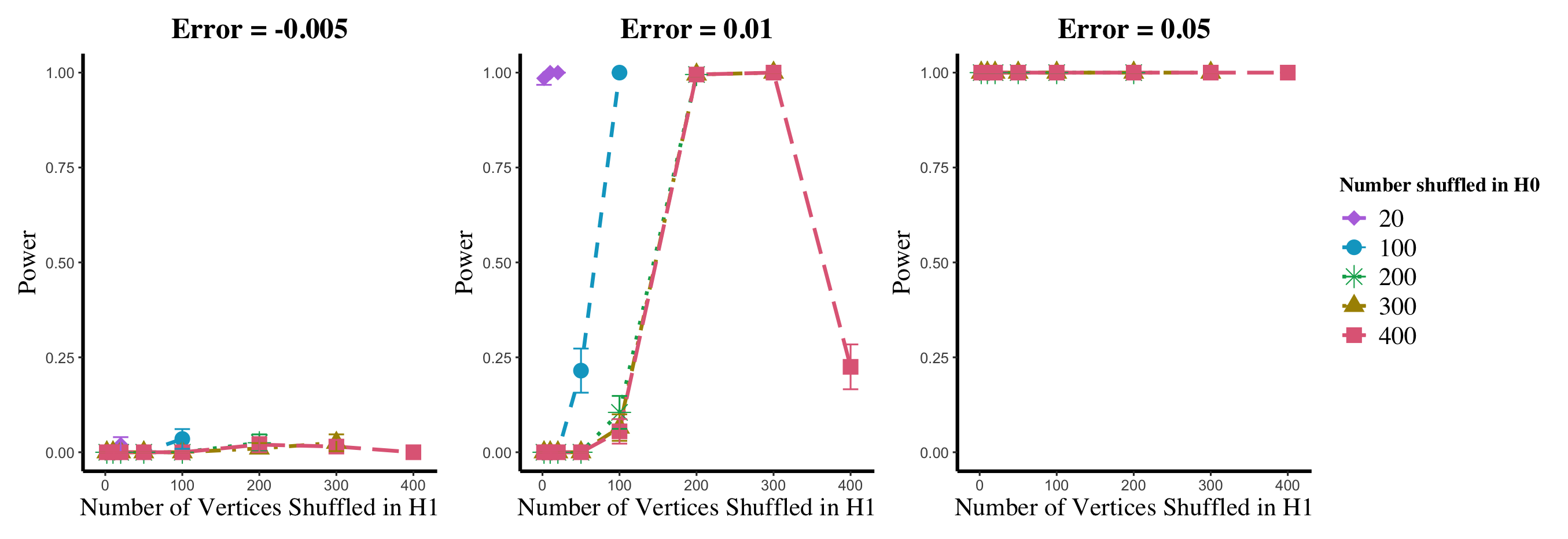}
\caption{For the experimental setup considered in Section \ref{app:addexp}, we plot the empirical testing power in the presence of shuffling for the $\hbP$-based test in the sparse regime.
In the figure
the x-axis represents the number of vertices actually shuffled in $U_{n,k}$ (i.e., the number shuffled in the alternative) while the curve colors
represent the maximum number of vertices potentially shuffled via in $U_{n,k}$; error bars are $\pm2$s.e.).}
\label{sparsePhat}
\end{figure}

We next consider additional experiments in the ASE-based testing of Section \ref{sec:ASE}.  Herein, we show results for more values of $k,\ell, and \lambda$.

\begin{center}
\begin{figure}[h!]
\centering
\includegraphics[width=1\textwidth]{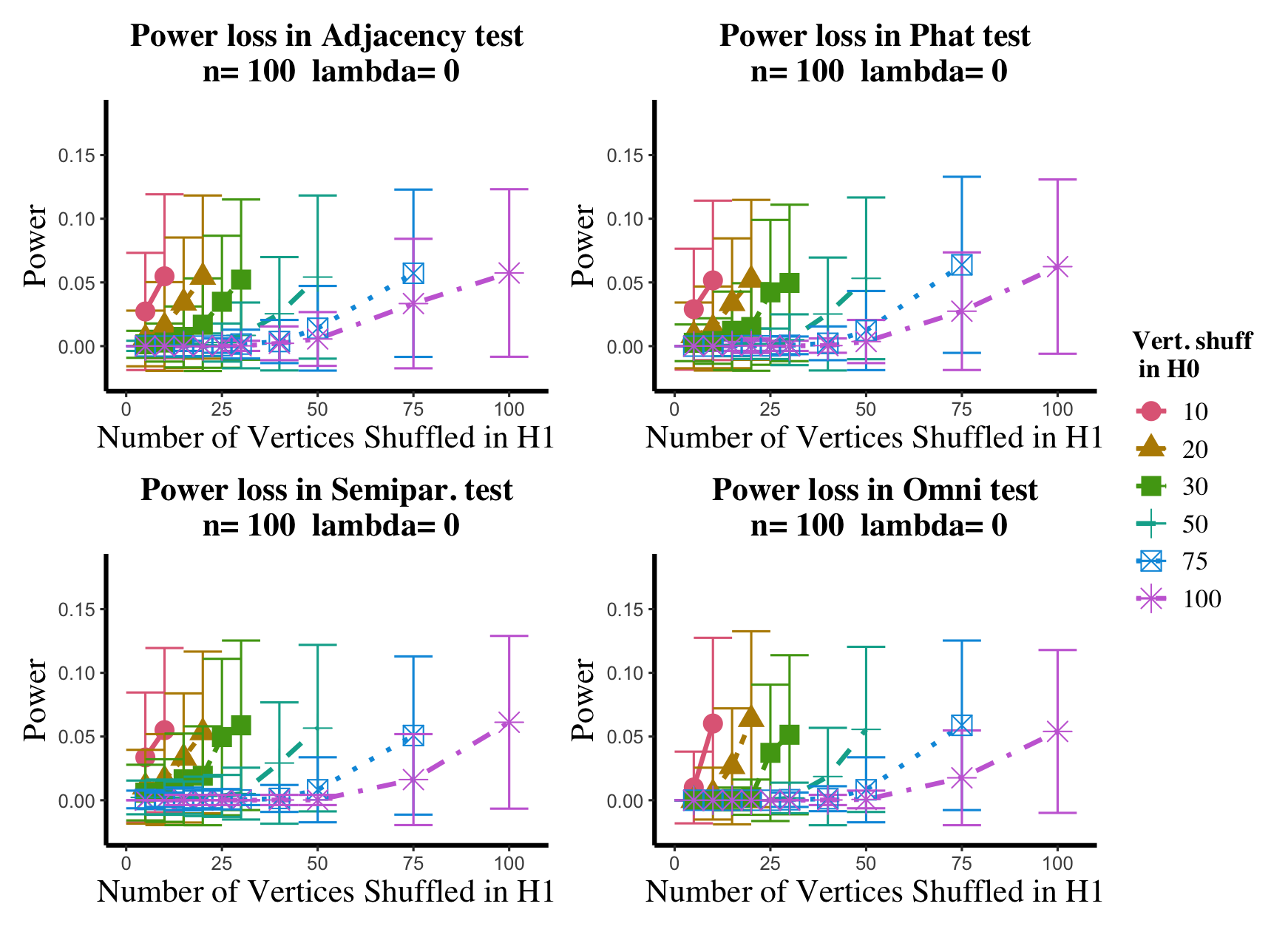}
\caption{For the experimental setup considered in Section \ref{sec:ASE}, we plot the empirical testing power in the presence of shuffling for the four tests: the Frobenius norm difference between the adjacency-matrices, between $\hbP$'s, $T_{\text{Omni}}$ and $T_{\text{Semipar}}$.
In the figure
the x-axis represents the number of vertices actually shuffled in $U_{n,k}$ (i.e., the number shuffled in the alternative) while the curve colors
represent the maximum number of vertices potentially shuffled via in $U_{n,k}$.}
\label{levinvsomniA1}
\end{figure}
\end{center}
\newpage

\begin{center}
\begin{figure}[h!]
\centering
\includegraphics[width=1\textwidth]{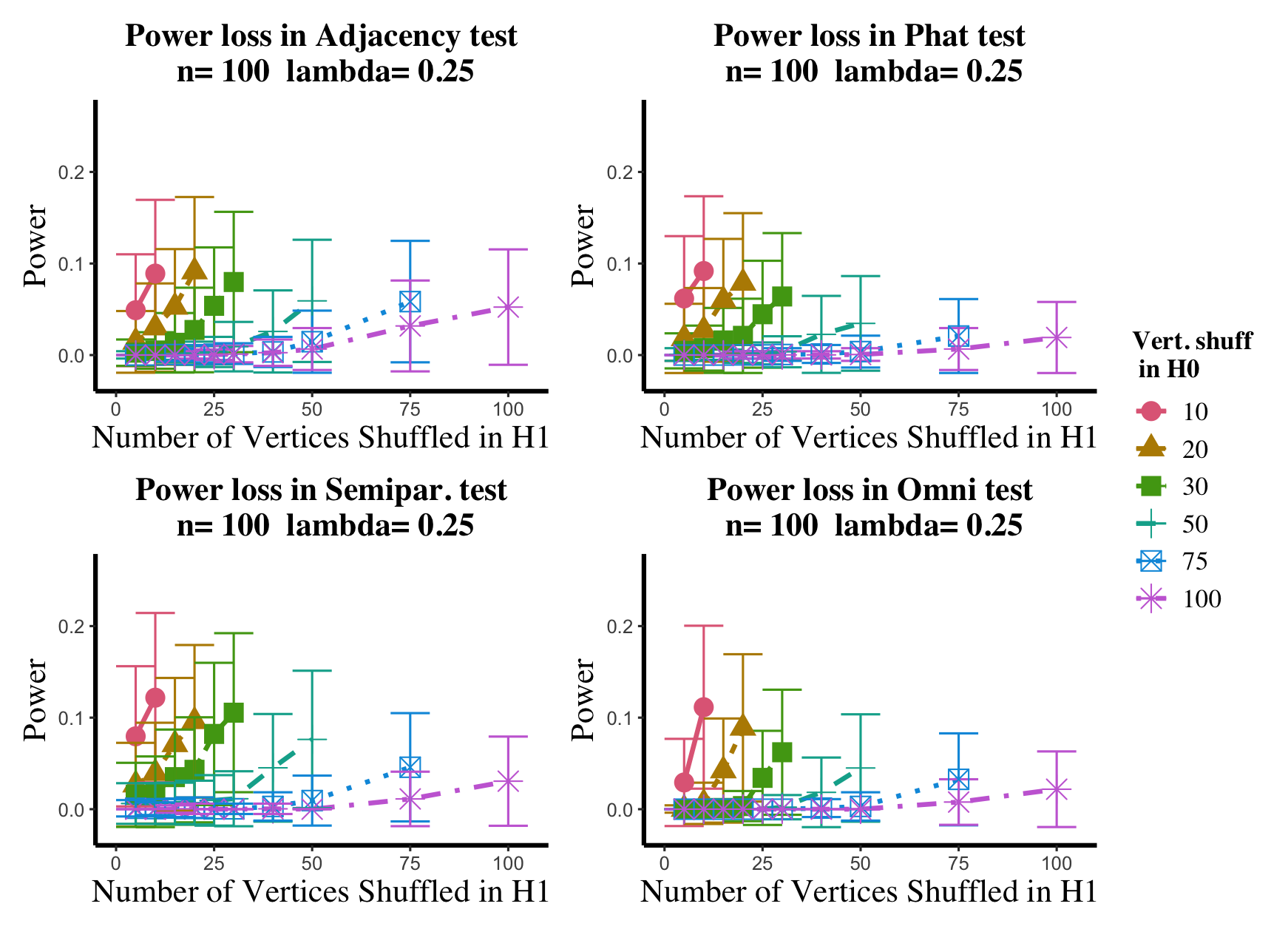}
\caption{For the experimental setup considered in Section \ref{sec:ASE}, we plot the empirical testing power in the presence of shuffling for the four tests: the Frobenius norm difference between the adjacency-matrices, between $\hbP$'s, $T_{\text{Omni}}$ and $T_{\text{Semipar}}$.
In the figure
the x-axis represents the number of vertices actually shuffled in $U_{n,k}$ (i.e., the number shuffled in the alternative) while the curve colors
represent the maximum number of vertices potentially shuffled via in $U_{n,k}$.}
\label{levinvsomniA2}
\end{figure}
\end{center}
\newpage

\begin{figure}[h!]
\centering
\includegraphics[width=1\textwidth]{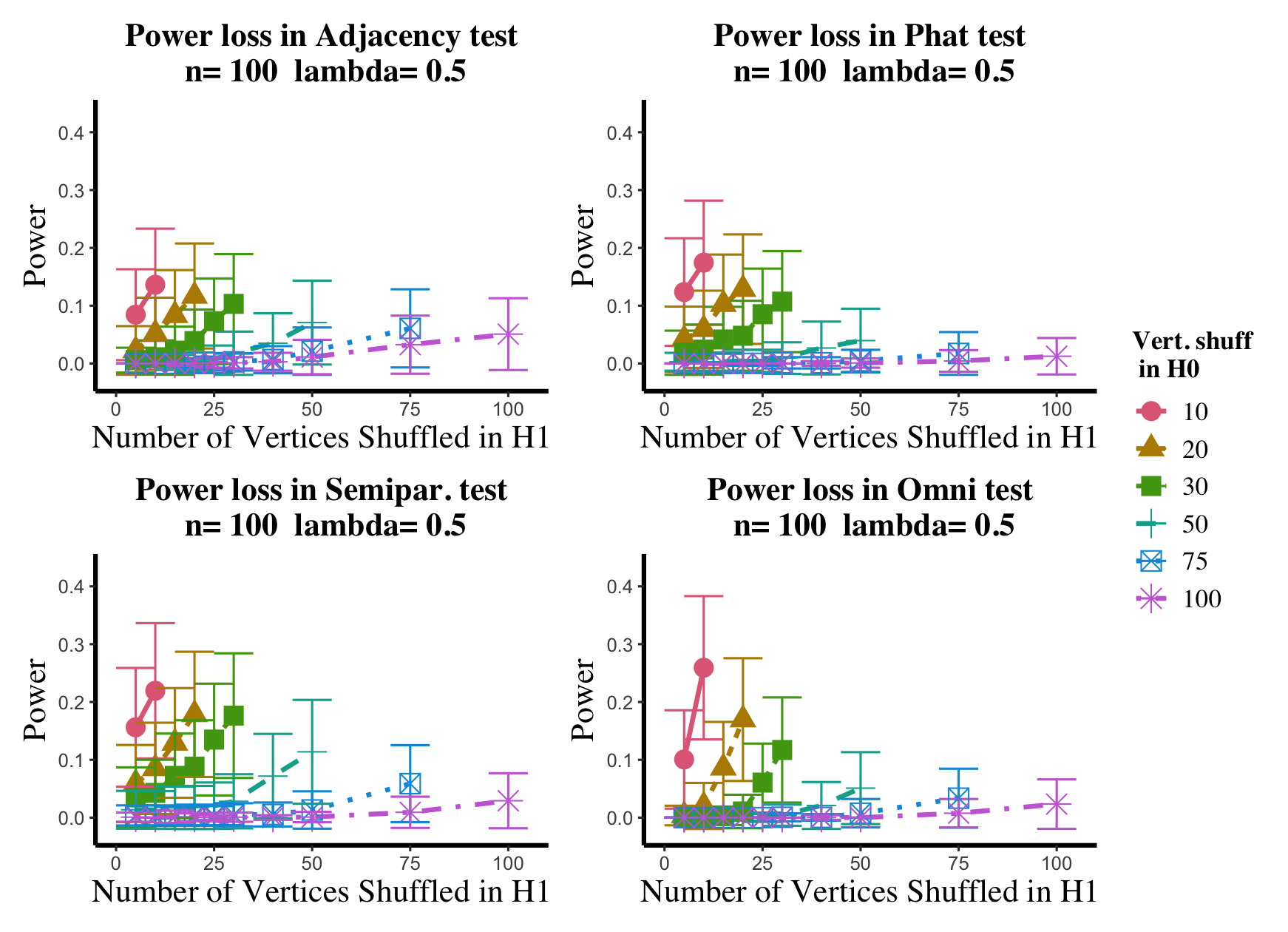}
\caption{For the experimental setup considered in Section \ref{sec:ASE}, we plot the empirical testing power in the presence of shuffling for the four tests: the Frobenius norm difference between the adjacency-matrices, between $\hbP$'s, $T_{\text{Omni}}$ and $T_{\text{Semipar}}$.
In the figure
the x-axis represents the number of vertices actually shuffled in $U_{n,k}$ (i.e., the number shuffled in the alternative) while the curve colors
represent the maximum number of vertices potentially shuffled via in $U_{n,k}$.}
\label{levinvsomniA3}
\end{figure}

\begin{figure}[h!]
\centering
\includegraphics[width=1\textwidth]{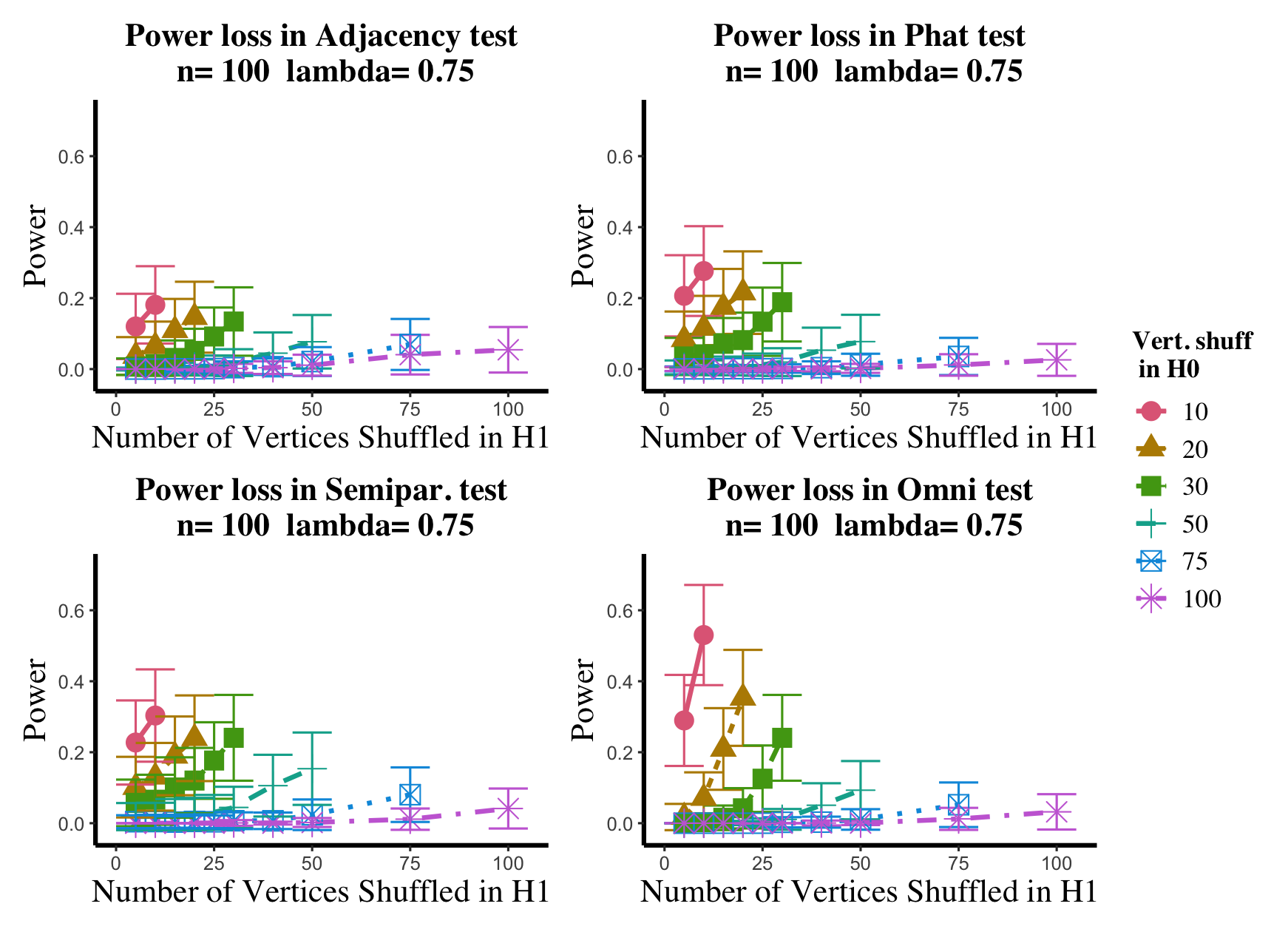}
\caption{For the experimental setup considered in Section \ref{sec:ASE}, we plot the empirical testing power in the presence of shuffling for the four tests: the Frobenius norm difference between the adjacency-matrices, between $\hbP$'s, $T_{\text{Omni}}$ and $T_{\text{Semipar}}$.
In the figure
the x-axis represents the number of vertices actually shuffled in $U_{n,k}$ (i.e., the number shuffled in the alternative) while the curve colors
represent the maximum number of vertices potentially shuffled via in $U_{n,k}$.}
\label{levinvsomniA4}
\end{figure}

\begin{figure}[h!]
\centering
\includegraphics[width=1\textwidth]{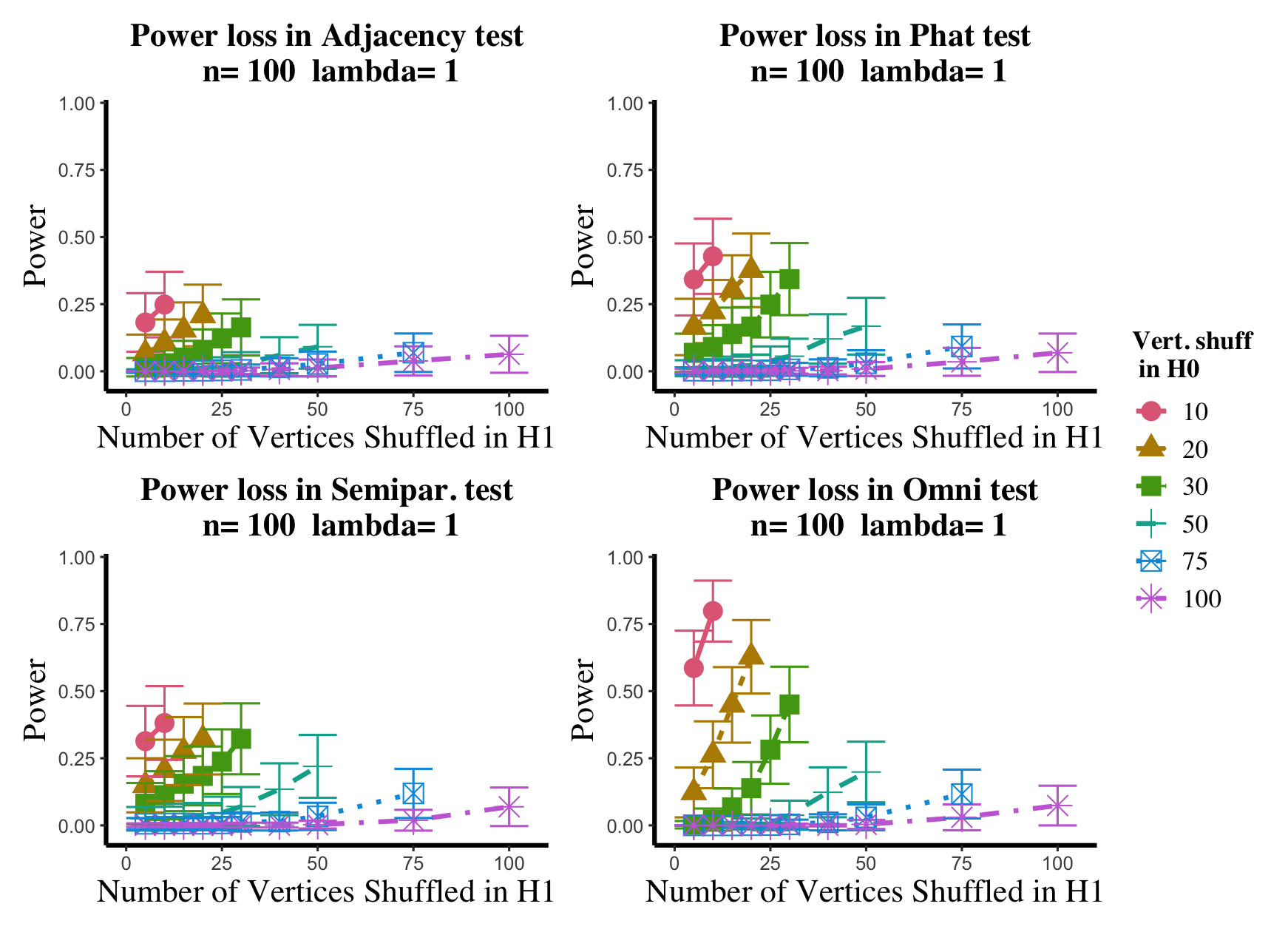}
\caption{For the experimental setup considered in Section \ref{sec:ASE}, we plot the empirical testing power in the presence of shuffling for the four tests: the Frobenius norm difference between the adjacency-matrices, between $\hbP$'s, $T_{\text{Omni}}$ and $T_{\text{Semipar}}$.
In the figure
the x-axis represents the number of vertices actually shuffled in $U_{n,k}$ (i.e., the number shuffled in the alternative) while the curve colors
represent the maximum number of vertices potentially shuffled via in $U_{n,k}$.}
\label{levinvsomniA5}
\end{figure}


 % \bibliographystyle{elsarticle-num} 
 % \bibliography{cite.bib}





% \end{thebibliography}
\end{document}